\documentclass{article}
\usepackage{comment}
\usepackage{mathtools}
\usepackage{mathrsfs}
\usepackage{amsfonts}
\usepackage{amssymb}
\usepackage{graphicx}
\usepackage{amsmath}
\usepackage{amsthm}
\usepackage[colorlinks=true,linkcolor=blue,citecolor= red,bookmarksopen=true]{hyperref}
\usepackage{amsmath}
\usepackage{amsfonts}
\usepackage{amsthm}

\nonstopmode
\newtheorem{theorem}{Theorem}[section]
\newtheorem{proposition}[theorem]{Proposition}
\newtheorem{lemma}[theorem]{Lemma}

\newtheorem{corollary}[theorem]{Corollary}

\newtheorem{assumption}[theorem]{Assumption}
\newtheorem{definition}[theorem]{Definition}

\theoremstyle{remark}
\newtheorem{example}[theorem]{Example}
\newtheorem{remark}[theorem]{Remark}
\newcommand{\probp}{P}
\newcommand{\probq}{Q}
\newcommand{\R}{{\mathbb R}}

\begin{document}

\title{Systemic Optimal Risk Transfer Equilibrium}
\date{\today }
\author{Francesca Biagini\thanks{%
Department of Mathematics, University of Munich, Theresienstra{\ss }e 39,
80333 Munich, Germany, \emph{francesca.biagini@math.lmu.de.}} \and %
Alessandro Doldi \thanks{%
Dipartimento di Matematica, Universit\`a degli Studi di Milano, Via Saldini
50, 20133 Milano, Italy, $\,\,$\emph{alessandro.doldi@unimi.it}. } \and %
Jean-Pierre Fouque \thanks{%
Department of Statistics \& Applied Probability, University of California,
Santa Barbara, CA 93106-3110, \emph{fouque@pstat.ucsb.edu}. Work supported
by NSF grant DMS-1814091.} \and Marco Frittelli\thanks{%
Dipartimento di Matematica, Universit\`a degli Studi di Milano, Via Saldini
50, 20133 Milano, Italy, \emph{marco.frittelli@unimi.it}.} \and Thilo
Meyer-Brandis\thanks{%
Department of Mathematics, University of Munich, Theresienstra{\ss }e 39,
80333 Munich, Germany, \emph{meyerbr@math.lmu.de}.}}
\maketitle

\begin{abstract}
\noindent We propose a novel concept of a Systemic Optimal Risk Transfer
Equilibrium (SORTE), which is inspired by the B\"{u}hlmann's classical
notion of an Equilibrium Risk Exchange. We provide sufficient general
assumptions that guarantee existence, uniqueness, and Pareto optimality of
such a SORTE. In both the B\"{u}hlmann and the SORTE definition, each agent
is behaving rationally by maximizing his/her expected utility given a budget
constraint. The two approaches differ by the budget constraints. In B\"{u}%
hlmann's definition the vector that assigns the budget constraint is given a
priori. On the contrary, in the SORTE approach, the vector that assigns the
budget constraint is endogenously determined by solving a systemic utility
maximization. SORTE gives priority to the systemic aspects of the problem,
in order to optimize the overall systemic performance, rather than to
individual rationality.
\end{abstract}

\begin{equation*}
\,
\end{equation*}

\newline

\parindent=0em \noindent

\section{Introduction}

We introduce the concept of Systemic Optimal Risk Transfer Equilibrium,
denoted by SORTE, that conjugates the classical B\"{u}hlmann's notion of an
equilibrium risk exchange with capital allocation based on systemic expected
utility optimization.

The capital allocation and risk sharing equilibrium that we consider can be
applied to many contexts, such as: equilibrium among financial institutions,
agents, or countries; insurance and reinsurance markets; capital allocation
among business units of a single firm; wealth allocation among investors.

In this paper we will refer to a participant in these problems (financial
institution or firms or countries) as an \textbf{agent}; the class
consisting of these $N$ agents as the \textbf{system}; the individual risk
of the agents (or the random endowment or future profit and loss) as the 
\textbf{risk vector} $\mathbf{X}:=(X^{1},...,X^{N})$; the amount $\mathbf{Y}%
:=(Y^{1},...,Y^{N})$ that can be exchanged among the agents as random 
\textbf{allocation}. We will generically refer to a central regulator
authority, or CCP, or executive manager as a \textbf{central bank} (CB).

We now present the main concepts of our approach and leave the details and
the mathematical rigorous presentation to the next sections. In a one period
framework, we consider $N$ agents, each one characterized by a concave,
strictly monotone utility function $u_{n}:\mathbb{R\rightarrow R}$ and by
the original risk $X^{n}\in L^{0}(\Omega ,\mathcal{F},P),$ for $n=1,...,N$.
Here, $(\Omega ,\mathcal{F},P\mathbf{)}$ is a probability space and $%
L^{0}(\Omega ,\mathcal{F},P)$ is the vector space of real valued $\mathcal{F}
$-measurable random variables. The sigma-algebra $\mathcal{F}$ represents
all possible measurable events at the final time $T$. $\mathbb{E}\left[
\cdot \right] $ denotes the expectation under $P$. Given another probability
measure $Q$, $E_{Q}\left[ \cdot \right] $ denotes the expectation under $Q$.
For the sake of simplicity and w.l.o.g., we are assuming zero interest rate.
We will use the bold notation to denote vectors.

\begin{enumerate}
\item \textbf{B\"{u}hlmann's risk exchange equilibrium}

We recall B\"{u}hlmann's definition of a risk exchange equilibrium in a pure
exchange economy (or in a reinsurance market). The initial wealth of agent $%
n $ is denoted by $x^{n}\in \mathbb{R}$ and the variable $X^{n}$ represents
the original risk of this agent. In this economy each agent is allowed to
exchange risk with the other agents. Each agent has to agree to receive (if
positive) or to provide (if negative) the amount $\widetilde{Y}^{n}(\omega )$
at the final time in exchange of the amount $E_{Q}[\widetilde{Y}^{n}]$ paid
(if positive) or received (if negative) at the initial time, where $Q$ is
some pricing probability measure. Hence $\widetilde{Y}^{n}$ is a time $T$
measurable random variable. In order that at the final time this risk
sharing procedure is indeed possible, the exchange variables $\widetilde{Y}%
^{n}$ have to satisfy the \emph{clearing condition}

\begin{equation*}
\sum_{n=1}^{N}\widetilde{Y}^{n}=0\text{ \ }P\text{-}a.s.\,\,\,.
\end{equation*}

As in B\"{u}hlmann \cite{Buhlmann1} and \cite{Buhlmann}, we say that a pair (%
$\widetilde{\mathbf{Y}}_{\mathbf{X}},Q_{\mathbf{X}})$ is an \textbf{risk
exchange equilibrium} if:

(a) for each $n$, $\widetilde{Y}_{\mathbf{X}}^{n}$ maximizes: $\mathbb{E}%
\left[ u_{n}(x^{n}+X^{n}+\widetilde{Y}^{n}-E_{Q_{\mathbf{X}}}[\widetilde{Y}%
^{n}])\right] $ among all variables $\widetilde{Y}^{n}$;

(b) $\sum_{n=1}^{N}\widetilde{Y}_{\mathbf{X}}^{n}=0$ $P-$a.s. .

It is clear that only for some particular choice of the equilibrium pricing
measure $Q_{\mathbf{X}}$, the optimal solutions $\widetilde{Y}_{\mathbf{X}%
}^{n}$ to the problems in (a) will also satisfy the condition in (b). 

In addition it is evident that the clearing condition in (b) requires that all agents accept to exchange the amount $\widetilde{Y}_{%
\mathbf{X}}^{n}(\omega )$ at the final time $T$.


Define 
\begin{equation}
\mathcal{C}_{\mathbb{R}}:=\left\{ \mathbf{Y}\in (L^{0}(\Omega ,\mathcal{F}%
,P))^{N}\mid \sum_{n=1}^{N}Y^{n}\in \mathbb{R}\right\}   \label{Cr}
\end{equation}%
that is, $\mathcal{C}_{\mathbb{R}}$ is the set of random vectors such that
the sum of the components is $P$-a.s. a deterministic number.

Observe that with the change of notations $Y^{n}:=x^{n}+\widetilde{Y}%
^{n}-E_{Q_{\mathbf{X}}}[\widetilde{Y}^{n}]$, we obtain variables with $E_{Q_{%
\mathbf{X}}}[Y^{n}]=x^{n}$ for each $n,$ and an optimal solution $Y_{\mathbf{%
X}}^{n}$ still belonging to $\mathcal{C}_{\mathbb{R}}$ and satisfying%
\begin{equation}
\sum_{n=1}^{N}Y_{\mathbf{X}}^{n}=\sum_{n=1}^{N}x^{n}\quad P\text{-}a.s.\,\,.
\label{sumY}
\end{equation}%
As can be easily checked%
\begin{equation*}
\sup_{\widetilde{Y}^{n}}\mathbb{E}\left[ u_{n}(x^{n}+X^{n}+\widetilde{Y}%
^{n}-E_{Q_{\mathbf{X}}}[\widetilde{Y}^{n}])\right] =\sup_{Y^{n}}\left\{ 
\mathbb{E}\left[ u_{n}(X^{n}+Y^{n})\right] \mid E_{Q_{\mathbf{X}%
}}[Y^{n}]\leq x^{n}\right\} .
\end{equation*}%
Hence the two above conditions in the definition of a risk exchange
equilibrium may be equivalently reformulated as

(a') for each $n$, $Y_{\mathbf{X}}^{n}$ maximizes: $\mathbb{E}\left[
u_{n}(X^{n}+Y^{n})\right] $ among all variables satisfying $E_{Q_{\mathbf{X}%
}}[Y^{n}]\leq x^{n}$;

(b') $\mathbf{Y}_{\mathbf{X}}\in \mathcal{C}_{\mathbb{R}}$\quad and\quad $%
\sum_{n=1}^{N}Y_{\mathbf{X}}^{n}=\sum_{n=1}^{N}x^{n}$ $P$-a.s.

We remark that here the quantity $x^{n}\in \mathbb{R}$ is preassigned to
each agent.

\item \textbf{Systemic Optimal (deterministic) Allocation}

To simplify the presentation, we now suppose that the initial wealth of each
agent is already absorbed in the notation $X^{n}$, so that $X^{n}$
represents the initial wealth plus the original risk of agent $n$. We assume
that the system has at disposal a total amount of capital $A\in \mathbb{R}$
to be used at a later time in case of necessity. This amount could have been
assigned by the Central Bank, or could have been the result of the previous
trading in the system, or could have been collected ad hoc by the agents. 
The amount A could represent an insurance pot or a fund collected (as guarantee for future investments) in a community of homeowners.
For further interpretation of $A$, see also the related discussion in
Section 5.2 of Biagini et al. \cite{bffm}. In any case, we consider the
quantity $A$ as exogenously determined. This amount is allocated among the
agents in order to optimize the overall systemic satisfaction. If we denote
with $a^{n}\in \mathbb{R}$ the cash received (if positive) or provided (if
negative) by agent $n$, then the time $T$ wealth at disposal of agent $n$
will be $(X^{n}+a^{n})$. The optimal vector $\mathbf{a_{\mathbf{X}}\in }%
\mathbb{R}^{N}$ could be determined according to the following aggregate
time-$T$ criterion 
\begin{equation}
\sup \left\{ \sum_{n=1}^{N}\mathbb{E}\left[ u_{n}(X^{n}+a^{n})\right] \mid 
\mathbf{a\in }\mathbb{R}^{N}\text{ s.t. }\sum_{n=1}^{N}a^{n}=A\right\} .
\label{ProblemOptimum}
\end{equation}%
Note that each agent is not optimizing his own utility function. As the vector $%
\mathbf{a\in }\mathbb{R}^{N}$ is deterministic, it is known at time $t=0$
and therefore the agents have to agree to provide or receive money only at such
initial time.
\end{enumerate}

However, under the assumption that also at the final time the agents 
 have confidence in the overall reliability of
the other agents, one can combine the two approaches outlined in Items 1 and
2 above to further increase the optimal total expected systemic utility and
simultaneously guarantee that each agent will optimize his/her own single
expected utility, taking into consideration an aggregated budget constraint
assigned by the system. Of course an alternative assumption to 
trustworthiness could be that the rules are enforced by the CB.

We denote with $\mathcal{L}^{n}\subseteq L^{0}(\Omega ,\mathcal{F},P)$ a
space of admissible random variables and assume that $\mathcal{L}^{n}+%
\mathbb{R=}\mathcal{L}^{n}$. We will consider maps $p^{n}:\mathcal{L}%
^{n}\rightarrow \mathbb{R}$ that represent the pricing or cost functionals,
one for each agent $n$. As we shall see, in some relevant cases, all agents
will adopt the same functional $p^{1}=...=p^{N}$, which will then be
interpreted as the equilibrium pricing functional, as in B\"{u}hlmann's
setting above, where $p^{n}(\cdot ):=E_{Q}[\cdot ]$ for all $n$. However, we
do not have to assume this a priori. Instead we require that the maps $p^n$
satisfy for all $n=1,...,N$:

i) $p^{n}$ is monotone increasing;

ii) $p^{n}(0)=0;$

iii) $p^{n}(Y+c)=p^{n}(Y)+c$ for all $c\in \mathbb{R}$ and $Y\in \mathcal{L}%
^{n}$.

Such assumptions in particular imply $p^{n}(c)=c$ for all constants $c\in 
\mathbb{R}$. A relevant example of such functionals are 
\begin{equation}
p^{n}(\cdot ):=E_{Q^{n}}[\cdot ]\,,  \label{pricingFunctional}
\end{equation}%
where $Q^n$ are probability measures for $n=1,...,N$. Another example could be $p^{n}=-\rho^{n}$, for convex risk measures $\rho^{n}$.

Now we will apply both approaches, outlined in Items 1 and 2 above, to
describe the concept of a Systemic Optimal Risk Transfer Equilibrium.

\begin{enumerate}
\item[3.] \textbf{Systemic Optimal Risk Transfer Equilibrium.}

As explained in Item 1, given some amount $a^{n}$ assigned to agent $n$,
this agent may buy $\widetilde{Y}^{n}$ at the price $p^{n}(\widetilde{Y}%
^{n}) $ in order to optimize 
\begin{equation*}
\mathbb{E}\left[ u_{n}(a^{n}+X^{n}+\widetilde{Y}^{n}-p^{n}(\widetilde{Y}%
^{n}))\right] \text{.}
\end{equation*}%
The pricing functionals $p^{n}$, $n=1,...,N$ have to be selected so that the
optimal solution verifies the clearing condition 
\begin{equation*}
\sum_{n=1}^{N}\widetilde{Y}^{n}=0\quad P\text{-a.s.}
\end{equation*}%
However, as in Item 2, $a^{n}$ is not exogenously assigned to each agent,
but only the total amount $A$ is at disposal of the whole system. Thus the
optimal way to allocate $A$ among the agents is given by the solution $(%
\widetilde{Y}_{\mathbf{X}}^{n},p_{\mathbf{X}}^{n},a_{\mathbf{X}}^{n})$ of
the following problem:%
\begin{eqnarray}
&&\sup_{\mathbf{a\in }\mathbb{R}^{N}}\left\{ \sum_{n=1}^{N}\sup_{\widetilde{Y%
}^{n}}\left\{ \mathbb{E}\left[ u_{n}(a^{n}+X^{n}+\widetilde{Y}^{n}-p_{%
\mathbf{X}}^{n}(\widetilde{Y}^{n}))\right] \right\} \,\middle%
|\sum_{n=1}^{N}a^{n}=A\right\} ,  \label{Problem2} \\
&&\sum_{n=1}^{N}\widetilde{Y}_{\mathbf{X}}^{n}=0\,\,\,P-\text{a.s.}\,\,.
\label{Problem2a}
\end{eqnarray}%
From \eqref{Problem2} and %
\eqref{Problem2a} it easily follows that an optimal solution $(\widetilde{Y}_{%
\mathbf{X}}^{n},p_{\mathbf{X}}^{n},a_{\mathbf{X}}^{n})$ fulfills 
\begin{equation}
\sum_{n=1}^{N}p_{\mathbf{X}}^{n}(\widetilde{Y}_{\mathbf{X}}^{n})=0.
\label{PriceZero}
\end{equation}%
Further, letting $Y^{n}:=a^{n}+\widetilde{Y}^{n}-p_{\mathbf{X}}^{n}(%
\widetilde{Y}^{n})$, from the cash additivity of $p_{\mathbf{X}}^{n}$ we
deduce $p_{\mathbf{X}}^{n}(Y^{n})=a^{n}+p_{\mathbf{X}}^{n}(\widetilde{Y}%
^{n})-p_{\mathbf{X}}^{n}(\widetilde{Y}^{n})=a^{n}$ and $\sum_{n=1}^{N}Y_{%
\mathbf{X}}^{n}=\sum_{n=1}^{N}a^{n}+\sum_{n=1}^{N}\widetilde{Y}_{\mathbf{X}%
}^{n}-\sum_{n=1}^{N}p_{\mathbf{X}}^{n}(\widetilde{Y}_{\mathbf{X}%
}^{n})=\sum_{n=1}^{N}a^{n}$ and, as before, the above optimization problem
can be reformulated as 
\begin{eqnarray}
&&\sup_{\mathbf{a\in }\mathbb{R}^{N}}\left\{
\sum_{n=1}^{N}\sup_{Y^{n}}\left\{ \mathbb{E}\left[ u_{n}(X^{n}+Y^{n})\right]
\mid p_{\mathbf{X}}^{n}(Y^{n})\leq a^{n}\right\} \,\middle%
|\sum_{n=1}^{N}a^{n}=A\right\} ,  \label{Problem} \\
&&\sum_{n=1}^{N}Y_{\mathbf{X}}^{n}=A\,\,\,P-\text{a.s.}\,\,,
\label{ProblemY}
\end{eqnarray}%
where analogously to \eqref{PriceZero} we have that a solution $(Y_{\mathbf{X%
}}^{n},p_{\mathbf{X}}^{n},a_{\mathbf{X}}^{n})$ satisfies $\sum_{n=1}^{N}p_{%
\mathbf{X}}^{n}(Y_{\mathbf{X}}^{n})=A$, by \eqref{Problem} and %
\eqref{ProblemY}.

The two optimal values in (\ref{Problem2}) and (\ref{Problem}) coincide. We
see that while each agent is behaving optimally according to his
preferences, the budget constraint $p_{\mathbf{X}}^{n}(Y^{n})\leq a^{n}$ are
not a priori assigned, but are endogenously determined through an aggregate
optimization problem. The optimal value $a_{\mathbf{X}}^{n}$ determines the
optimal risk allocation of each agent. It will turn out that $a_{\mathbf{X}%
}^{n}=p_{\mathbf{X}}^{n}(Y_{\mathbf{X}}^{n})$. Obviously, the optimal value
in (\ref{Problem2}) is greater than (or equal to) the optimal value in (\ref%
{ProblemOptimum}), which can be economically translated into the statement
that \emph{allowing for exchanges also at terminal time increases the systemic performance}.

In addition to the condition in (\ref{ProblemY}), we introduce further
possible constraints on the optimal solution, by requiring that 
\begin{equation}
\mathbf{Y}_{\mathbf{X}}\in \mathcal{B},  \label{BB}
\end{equation}%
where $\mathcal{B}\subseteq \mathcal{C}_{\mathbb{R}}$.
\end{enumerate}

In the paper, see Section \ref{SecSORTE}, we formalize the above discussion
and show the existence of the solution $(Y_{\mathbf{X}}^{n}$, $p_{\mathbf{X}%
}^{n},a_{\mathbf{X}}^{n})$ to (\ref{Problem}), (\ref{ProblemY}) and (\ref{BB}%
), which we call Systemic Optimal Risk Transfer Equilibrium (SORTE). We show
that $p_{\mathbf{X}}^{n}$ can be chosen to be of the particular form $p_{%
\mathbf{X}}^{n}(\cdot ):=E_{Q_{\mathbf{X}}^{n}}[\cdot ]$, for a probability
vector $\mathbf{Q}_{\mathbf{X}}=(Q_{\mathbf{X}}^{1},...,Q_{\mathbf{X}}^{N})$%
. The crucial step, Theorem \ref{thmoptimumexists}, is the proof of the dual
representation and the existence of the optimizer of the associated problem (%
\ref{defpi}). The optimizer of the dual formulation provides the optimal
probability vector $\mathbf{Q}_{\mathbf{X}}$ that determines the functional $%
p_{\mathbf{X}}^{n}(\cdot ):=E_{Q_{\mathbf{X}}^{n}}[\cdot ]$. The
characteristics of the optimal $\mathbf{Q}_{\mathbf{X}}$ depend on the
feasible allocation set $\mathcal{B}$. When no constraints are enforced,
i.e., when $\mathcal{B}=\mathcal{C}_{\mathbb{R}}$, then all the components
of $\mathbf{Q}_{\mathbf{X}}$ turn out to be equal. Hence we find that the
implicit assumption of one single equilibrium pricing measure, made in the B%
\"{u}hlmann's framework, is in our theory a consequence of the particular
selection $\mathcal{B}=\mathcal{C}_{\mathbb{R}}$, but for general $\mathcal{B%
}$ this in not always the case.\
At this point it might be convenient for the reader to have at hand the example of the
exponential utility function that is described in Section \ref{expnew} and Section \ref{sectionexp}, where we obtain
an explicit formulation of the optimal solution $\mathbf{Y}_{\mathbf{X}}$,
of the equilibrium pricing measure $\mathbf{Q}_{\mathbf{X}}$ and of the
optimal vector $\mathbf{a}_{\mathbf{X}}$.

\begin{remark}
We emphasize that the existence of multiple equilibrium pricing measures $%
\mathbf{Q}_{\mathbf{X}}=(Q_{\mathbf{X}}^{1},...,Q_{\mathbf{X}}^{N})$ is a
natural consequence of the presence of the - non trivial - constraints set $%
\mathcal{B}$. Indeed, even in the B\"{u}hlmann setting, if we add
constraints, of a very simple nature, a single equilibrium pricing measure
might not exists any more. Consider the following extension of a B\"{u}%
hlmann risk exchange equilibrium.

Let $\mathcal{B}\subseteq \mathcal{C}_{{\mathbb{R}}}$ be fixed. We say that
a pair ($\widetilde{\mathbf{Y}}_{\mathbf{X}},Q_{\mathbf{X}})$ is a \textbf{%
constrained }risk exchange equilibrium if:

(a2) for each $n$, $\widetilde{Y}_{\mathbf{X}}^{n}$ maximizes: $\mathbb{E}%
\left[ u_{n}(x^{n}+X^{n}+\widetilde{Y}^{n}-E_{Q_{\mathbf{X}}}[\widetilde{Y}%
^{n}])\right] $ among all variables $\widetilde{Y}^{n}$;

(b2) $\widetilde{\mathbf{Y}}_{\mathbf{X}}\in \mathcal{B}$ and $\sum_{n=1}^{N}%
\widetilde{Y}_{\mathbf{X}}^{n}=0$ $P-$a.s. .

We show with the next example that such an equilibrium (with one single
probability $Q_{\mathbf{X}})$ \textbf{does not exist in general}. The
example we present is rather simple, yet instructive, since it shows that
the absence of the equilibrium arises not from technical assumptions, like
integrability conditions, but is rather a structural problem caused by the
presence of additional constraints. Here we provide the intuition for it.
Suppose that two isolated systems of agents have, under suitable
assumptions, their own (unconstrained) equilibria, and that such two
equilibria do not coincide. As shown in the next example, we might then
consider the two systems as one single larger system consisting of two
isolated clusters, expressing this latter property with the addition of
constraints. Then it is evident that an equilibrium (with a unique pricing
measure) cannot exist for such unified system.
\end{remark}

\begin{example}
In order to ignore all integrability issues, in this example we assume that $%
\Omega $ is a finite set, endowed with the sigma algebra of all its subsets
and the uniform probability measure. Consider $N=4,$ $u_{n}(x):=(1-e^{-%
\alpha _{n}x}),$ $\alpha _{n}>0,$ $\,n=1,\dots ,4$, \ and some vectors $%
\mathbf{x}\in {\mathbb{R}}^{4}$, and $\mathbf{X}\in (L^{\infty })^{4}$.
Moreover take 
\begin{equation*}
\mathcal{B}=\left\{ \mathbf{Y}\in \mathcal{C}_{{\mathbb{R}}}\mid
Y^{1}+Y^{2}=0\text{, }Y^{3}+Y^{4}=0\right\} .
\end{equation*}%
Thus $\mathbf{X}$ and $\mathcal{B}$ model a single system of $4$ agents
which can exchange the risk only in a restricted way (agent $1$ with agent $%
2 $, and agent $3$ with agent $4)$, so that in effect the system consists of
two isolated clusters of agents.\ Then a constrained risk exchange
equilibrium in general does not exists. By contradiction, suppose that ($%
\widetilde{\mathbf{Y}}_{\mathbf{X}},Q_{\mathbf{X}})$ is a constrained risk
exchange equilibrium. It is easy to verify that $([\widetilde{{Y}}_{\mathbf{X%
}}^{1},\widetilde{{Y}}_{\mathbf{X}}^{2}],Q_{\mathbf{X}})$ is a
(unconstrained) risk exchange equilibrium with respect to $[X^{1},X^{2}]$
and $[x^{1},x^{2}]$ (i.e. it satisfies (a) and (b) for $N=2$). Similarly, $([%
\widetilde{{Y}}_{\mathbf{X}}^{3},\widetilde{{Y}}_{\mathbf{X}}^{4}],Q_{%
\mathbf{X}})$ is a (unconstrained) risk exchange equilibrium with respect to 
$[X^{3},X^{4}]$ and $[x^{3},x^{4}]$. This implies using equation (2) in B\"{u}hlmann \cite%
{Buhlmann} that 
\begin{equation*}
\frac{\exp {\left( \eta (X^{1}+X^{2})\right) }}{\mathbb{E}\left[ \exp {%
\left( \eta (X^{1}+X^{2})\right) }\right] }=\frac{\mathrm{d}Q_{\mathbf{X}}}{%
\mathrm{d}P}=\frac{\exp {\left( \theta (X^{3}+X^{4})\right) }}{\mathbb{E}%
\left[ \exp {\left( \theta (X^{3}+X^{4})\right) }\right] },\,\,\,\,\,\eta =%
\frac{1}{\alpha _{1}}+\frac{1}{\alpha _{2}},\text{ }\theta =\frac{1}{\alpha
_{3}}+\frac{1}{\alpha _{4}},
\end{equation*}%
which clearly gives a contradiction, since $\mathbf{X}$ is arbitrary.

Observe, however, that in this example a constrained equilibrium exists 
\textbf{if we allow for possibly different pricing measures}, namely if we
may replace the measure $Q_{\mathbf{X}}$ with a vector $\mathbf{Q}_{\mathbf{X%
}}$. This would amount to replacing (a2) with (a3) below, namely to require
that:

(a3) for each $n$, $\widetilde{Y}_{\mathbf{X}}^{n}$ maximizes: $\mathbb{E}%
\left[ u_{n}(x^{n}+X^{n}+\widetilde{Y}^{n}-E_{Q_{\mathbf{X}}^{n}}[\widetilde{%
Y}^{n}])\right] $ among all variables $\widetilde{Y}^{n}$;

(b2) $\widetilde{\mathbf{Y}}_{\mathbf{X}}\in \mathcal{B}$ and $\sum_{n=1}^{N}%
\widetilde{Y}_{\mathbf{X}}^{n}=0$ $P-$a.s. .

Then such an equilibrium exists. Indeed, by the results in B\"{u}hlmann \cite%
{Buhlmann}, we can guarantee the existence of the risk exchange equilibrium $%
([\widetilde{{Y}}_{\mathbf{X}}^{1},\widetilde{{Y}}_{\mathbf{X}}^{2}],Q_{%
\mathbf{X}}^{12})$ with respect to $[X^{1},X^{2}]$ and $[x^{1},x^{2}]$, and
the risk exchange equilibrium $([\widetilde{{Y}}_{\mathbf{X}}^{3},\widetilde{%
{Y}}_{\mathbf{X}}^{4}],Q_{\mathbf{X}}^{34})$ with respect to $[X^{3},X^{4}]$
and $[x^{3},x^{4}]$. Then \newline
$([\widetilde{{Y}}_{\mathbf{X}}^{1},\widetilde{{Y}}_{\mathbf{X}}^{2},%
\widetilde{{Y}}_{\mathbf{X}}^{3},\widetilde{{Y}}_{\mathbf{X}}^{4}],[Q_{%
\mathbf{X}}^{12},Q_{\mathbf{X}}^{12},Q_{\mathbf{X}}^{34},Q_{\mathbf{X}%
}^{34}])$ satisfies (a3) and (b2). \textbf{The conclusion is that, even in
the B\"{u}hlmann case, the presence of constraints implies multiple
equilibrium pricing measures}.

From the mathematical point of view, this fact is very easy to understand in
our setup, described in Assumption \ref{A00}. More constraints implies
a smaller set $\mathcal{B}_{0}$ of feasible vectors $\widetilde{\mathbf{Y}}%
\in \mathcal{B}$\ such that $\sum_{n=1}^{N}\widetilde{Y}_{\mathbf{X}}^{n}=0$
and this in turn implies a larger polar set of $\mathcal{B}_{0}$ (which we
will denote with $\mathcal{Q}$, see the definition in Section \ref%
{secproof} item 4. The equilibrium exists only if we are allowed to pick
the pricing vector $\mathbf{Q}_{\mathbf{X}}$ in this larger set $\mathcal{Q}$%
, but the elements in $\mathcal{Q}$ don't need to have all equal components.
Economically, multiple pricing measures may arise because the risk exchange
mechanism may be restricted to clusters of agents, as in this example, and
agents from different clusters may well adopt a different equilibrium
pricing measure. For further details on clustering, see the Examples \ref%
{exCh} and \ref{exmixtures}.
\end{example}

\bigskip

B\"{u}hlmann's equilibrium ($\mathbf{Y}_{\mathbf{X}}$) satisfies two
relevant properties: \emph{Pareto optimality} (there are no feasible
allocation $\mathbf{Y}$ such that all agents are equal or better off -
compared with $\mathbf{Y}_{\mathbf{X}}$ - and at least one of them is better
off) and \emph{Individual Rationality} (each agent is better off with $Y_{%
\mathbf{X}}^{n}$ than without it). Any feasible allocation satisfying these
two properties is called an \emph{optimal risk sharing rule}, see Barrieu
and El Karoui \cite{BE05} or Jouini et al. \cite{JST07}.

We show that a SORTE is unique (once the class of pricing functionals is
restricted to those in the form $p^{n}(\cdot )=E_{Q^{n}}[\cdot ]$). We also
prove Pareto optimality, see the Definition \ref{defparetos} and the exact
formulation in Theorem \ref{thmsorteuniqueA}.

However, a SORTE lacks Individual Rationality. This is shown in the toy
example of Section \ref{Ex}, but it is also evident from the expression in
equation (\ref{Problem}). As already mentioned, each agent is performing
rationally, maximizing her expected utility, but under a budget constraint $%
p_{\mathbf{X}}^{n}(Y^{n})\leq a_{\mathbf{X}}^{n}$\ that is determined
globally via an additional systemic maximization problem ($\sup_{\mathbf{%
a\in }\mathbb{R}^{N}}\{...\mid \sum_{n=1}^{N}a^{n}=A\}$) that assigns
priority to the systemic performance, rather than to each individual agent.
In the SORTE we replace individual rationality with such a \emph{systemic
induced individual rationality}, which also shows the difference between the
concepts of SORTE and of an optimal risk sharing rule. We also point out
that the participation in the risk sharing mechanism may be appropriately
mitigated or enforced by the use of adequate sets $\mathcal{B}$, see e.g.
Example \ref{exmixtures} for risk sharing restricted to subsystems. From the
technical point of view, we will not rely on any of the methods and results
related to the notion of inf-convolution, which is a common tool to prove
existence of optimal risk sharing rules (see for example \cite{BE05} or \cite%
{JST07}) in the case of monetary utility functions, as we do not require the
utility functions to be cash additive. Our proofs are based on the dual
approach to (systemic) utility maximization. This is summarized in Section %
\ref{SecScheme}. Furthermore, the exponential case is treated in detail in
Section \ref{sectionexp}.

\begin{remark}
\label{remgamma} As customary in the literature on general equilibrium and
risk sharing, we could have considered, in place of \eqref{Problem} and %
\eqref{ProblemY}, the more general problem 
\begin{eqnarray}
&&\sup_{\mathbf{a\in }\mathbb{R}^{N}}\left\{
\sum_{n=1}^{N}\sup_{Y^{n}}\left\{ \mathbb{E}\left[ \gamma
_{n}u_{n}(X^{n}+Y^{n})\right] \mid p_{\mathbf{X}}^{n}(Y^{n})\leq
a^{n}\right\} \,\middle|\sum_{n=1}^{N}a^{n}=A\right\} ,  \label{Problemgamma}
\\
&&\sum_{n=1}^{N}Y_{\mathbf{X}}^{n}=A\,\,\,P-\text{a.s.}\,\,,
\label{ProblemYgamma}
\end{eqnarray}%
where the positive weights $\gamma =(\gamma _{1},...,\gamma _{N})\in \mathbb{%
R}^{N}$ could have been selected exogenously, say by a social planner. In
such more general problems, equilibria will generally depend on the selected
weights. However, in this paper we are focused on existence, uniqueness and
Pareto optimality of the equilibrium and for this analysis we may restrict,
without loss of generality, our attention to the utilitarian choice $\gamma
_{1}=...=\gamma _{N}=1$, as we now explain. It is easy to check that given $%
u_{1},\dots ,u_{N}$ satisfying our assumptions (namely Assumption \ref{A00}%
.(a)), the associated functions $x\mapsto u_{n}^{\gamma }(x):=\gamma
_{n}u_{n}(x),\,n=1,\dots ,N$ will satisfy the same Assumption \ref{A00}.(a)
and so \eqref{Problemgamma} can be written as%
\begin{equation}
\sup_{\mathbf{a\in }\mathbb{R}^{N}}\left\{ \sum_{n=1}^{N}\sup_{Y^{n}}\left\{ 
\mathbb{E}\left[ u_{n}^{\gamma }(X^{n}+Y^{n})\right] \mid p_{\mathbf{X}%
}^{n}(Y^{n})\leq a^{n}\right\} \,\middle|\sum_{n=1}^{N}a^{n}=A\right\} ,
\label{Problemgamma1}
\end{equation}%
Thus, technically speaking, the study of the existence, uniqueness and
Pareto optimality of the equilibrium in a non-utilitarian setup ($\gamma
\neq 1$) boils down to the one in \eqref{Problem} and \eqref{ProblemY}. Of
course it could be of interest to study the dependence of the optimal
solution from the vector $\gamma $ and to analyze the stability properties
of the equilibrium with respect to the utility functions. In Section \ref%
{secdeponweight} we address this problem for exponential utility functions,
but the general case is left for future investigation.
\end{remark}

\textbf{Review of literature:} This paper originates from the systemic risk
approach developed in Biagini et al. \cite{bffm0} and \cite{bffm}. In \cite{bffm} 
 the main focus was the analysis of the systemic risk measure  
\begin{equation}
\rho (\mathbf{X}):=\inf_{\mathbf{Y}\in \mathcal{B}\subset \mathcal{C}_{%
\mathbb{R}}}\left\{ \sum_{n=1}^{N}Y^{n}\mid \mathbb{E}\left[
\sum_{n=1}^{N}u_{n}(X^{n}+Y^{n})\right] \geq B\right\} ,\text{ }B\in \mathbb{%
R},  \label{99}
\end{equation}%
which computes systemic risk as the minimal capital $\sum_{n=1}^{N}Y^{n}\in 
\mathbb{R}$ that secures the aggregated system $(\mathbb{E}\left[
\sum_{n=1}^{N}u_{n}(X^{n}+Y^{n})\right] \geq B)$ by injecting the random
allocation $Y^{n}$ into the single institution $X^{n}$. 

The notion of a SORTE is inspired by the following utility maximization
problem, associated to the risk minimization problem (\ref{99}),    
\begin{equation}
\sup_{\mathbf{Y}\in \mathcal{B}\subset \mathcal{C}_{\mathbb{R}}}\left\{ 
\mathbb{E}\left[ \sum_{n=1}^{N}u_{n}(X^{n}+Y^{n})\right] \mid
\sum_{n=1}^{N}Y^{n}\leq A\right\} ,\text{ }A\in \mathbb{R},  \label{piAbis}
\end{equation}%
that was also introduced in \cite{bffm}. Related papers on systemic risk
measures are Feinstein et al. \cite{FeinsteinRudloffWeber}, Acharya et al. 
\cite{pedersen}, Armenti et al. \cite{Drapeau}, Chen et al. \cite{chen},
Kromer et al. \cite{Kromer}. For an exhaustive overview on the literature on
systemic risk, see Hurd \cite{Hurd} and Fouque and Langsam \cite{JP_Langsam}.

For a review on Arrow-Debreu Equilibrium (see Debreu \cite{Debreu}; Mas
Colell and Zame \cite{MasColell} for the infinite dimensional case) we refer
to Section 3.6 of F\"{o}llmer and Schied \cite{Follmer}, which is close to
our setup. In the spirit of the Arrow-Debreu Equilibrium, B\"{u}hlmann \cite%
{Buhlmann1} and \cite{Buhlmann} proved the existence of risk exchange
equilibria in a pure exchange economy. Such risk sharing equilibria had been
studied in different forms starting from the seminal papers of Borch \cite%
{Borch}, where Pareto-optimal allocations were proved to be comonotonic for
concave utility functions, and 
B\"{u}hlmann and Jewell \cite{BJ79}. The differences with B\"{u}hlmann's
setup and our approach have been highlighted before in detail. \newline
In Barrieu and El Karoui \cite{BE05} inf-convolution of convex risk measures
has been introduced as a fundamental tool for studying risk sharing.
Existence of optimal risk sharing for law-determined monetary utility
functions is obtained in Jouini et al. \cite{JST07} and then generalized to
the case of non-monotone risk measures by Acciaio \cite{Acciaio} and Filipovi%
\'{c} and Svindland \cite{Filipovic_Svindland}, to multivariate risks by
Carlier and Dana \cite{Carlier1} and Carlier et al. \cite{Carlier2}, to
cash-subadditive and quasi-convex measures by Mastrogiacomo and Rosazza
Gianin \cite{MRG14}. Further works on risk sharing are also Dana and Le Van 
\cite{Dana_LeVan}, Heath and Ku \cite{Heath_Ku}, Tsanakas \cite{Tsanakas},
Weber \cite{Weber17}. Risk sharing problems with quantile-based risk
measures are studied in Embrechts et al. \cite{ELW18} by explicit
construction, and in \cite{ELW182} for heterogeneous beliefs. In Filipovi%
\'{c} and Kupper \cite{Kupper} Capital and Risk Transfer is modelled as
(deterministically determined) redistribution of capital and risk by means
of a finite set of non deterministic financial instruments. Existence issues
are studied and related concepts of equilibrium are introduced. Recent
further extensions have been obtained in Liebrich and Svindland \cite{LS18}.



\section{Notations}

Let $(\Omega ,\mathcal{F},P\mathbf{)}$ be a probability space and consider
the following set of probability vectors on $(\Omega ,\mathcal{F})$ 
\begin{equation*}
\mathcal{P}^{N}:=\left\{ \mathbf{Q=(}Q^{1},...,Q^{N}\mathbf{)}\mid \text{
such that }Q^{j}\ll P\text{ for all }j=1,...,N\right\} .
\end{equation*}

For a vector of probability measures $\mathbf{Q}$ we write $\mathbf{Q}\ll P$
to denote $Q^{1}{\ll P},\dots ,Q^{N}\ll P$. Similarly for $\mathbf{Q}\sim P $%
. Set $L^{0}(\Omega ,\mathcal{F},P;\mathbb{R}^{N}\mathbf{)}=(L^{0}(P))^{N}$.
For $Q\in \mathcal{P}^{1}$ let $L^{1}(Q\mathbf{):=}L^{1}(\Omega ,\mathcal{F}%
,Q;\mathbb{R}\mathbf{)}$ be the vector space of $Q-$ integrable random
variables and $L^\infty(Q):=L^{\infty}(\Omega ,\mathcal{F},Q;\mathbb{R}%
\mathbf{)}$ be the space of $Q-$ essentially bounded random variables. Set $%
L^1_+(Q)=\left\{Z\in L^1(Q)\,\,\middle|Z\geq 0\,\,Q-\text{a.s.}\right\}$ and 
$L^\infty_+(Q)=\left\{Z\in L^\infty(Q)\,\,\middle|\,\,Z\geq 0\,\,Q-\text{a.s.%
}\right\}$. For $\mathbf{Q}\in \mathcal{P}^{N}$ let%
\begin{equation*}
L^{1}(\mathbf{Q):=}L^1(Q^{1}\mathbf{)\times ...\times }L^1(Q^{N})\,,\,\,\,\,%
\,\,\,\,L^{1}_+(\mathbf{Q):=}L^1_+(Q^{1}\mathbf{)\times ...\times }%
L^1_+(Q^{N})\,,
\end{equation*}

\begin{equation*}
L^\infty(\mathbf{Q}):=L^\infty(Q^1)\times\dots\times
L^\infty(Q^N)\,,\,\,\,\,\,\,\,\,L^\infty_+(\mathbf{Q}):=L^\infty_+(Q^1)%
\times\dots\times L^\infty_+(Q^N).
\end{equation*}

For each $j=1,...,N$ consider a vector subspace $\mathcal{L}^{j}$ with ${%
\mathbb{R}}\subseteq \mathcal{L}^{j}\subseteq L^{0}(\Omega ,\mathcal{F},P;%
\mathbb{R}\mathbf{)}$ and set 
\begin{equation*}
\mathcal{L}\mathbf{:=}\mathcal{L}^{1}\times ...\times \mathcal{L}^{N}\mathbf{%
\subseteq }(L^{0}(P))^{N}.
\end{equation*}%
Consider now a subset $\mathscr{\probq}\subseteq \mathcal{P}^{N}$ and assume
that the pair $(\mathcal{L},\mathscr{Q})$ satisfies that for every $\mathbf{Q%
}\in \mathscr{Q}$ 
\begin{equation*}
\mathcal{L}\subseteq L^{1}(\mathbf{Q}).
\end{equation*}%
One could take as $\mathcal{L}^{j}$, for example, $L^{\infty }$ or some
Orlicz space. Our optimization problems will be defined on the vector space $%
\mathcal{L}$ to be specified later.

For each $n=1,...,N$, let $u_{n}:\mathbb{R\rightarrow R}$ be concave and
strictly increasing. Fix $\mathbf{X=(}X^{1},...,X^{N}\mathbf{)\in }\mathcal{L%
}$.

For $(\mathbf{Q},\,\mathbf{a,}\,A)\in \mathscr{Q}\mathbf{\times }\mathbb{R}%
^{N}\mathbf{\times }\mathbb{R}$ define 
\begin{eqnarray}
U_{n}^{Q^{n}}(a^{n}) &:&=\sup \left\{ \mathbb{E}\left[ u_{n}(X^{n}+Y)\right]
\mid Y\in \mathcal{L}^{n}\text{, }E_{Q^{n}}[Y]\leq a^{n}\right\} \,,
\label{UQ} \\
S^{\mathbf{Q}}(A) &:&=\sup \left\{ \sum_{n=1}^{N}U_{n}^{Q^{n}}(a^{n})\mid 
\mathbf{a\in }\mathbb{R}^{N}\text{ s.t. }\sum_{n=1}^{N}a^{n}\leq A\right\}
\,,  \label{SQ} \\
\Pi ^{\mathbf{Q}}(A) &:&=\sup \left\{ \mathbb{E}\left[
\sum_{n=1}^{N}u_{n}(X^{n}+Y^{n})\right] \mid \mathbf{Y}\in \mathcal{L}%
,\,\sum_{n=1}^{N}E_{Q^{n}}[Y^{n}]\leq A\right\} \,.  \label{PiQ}
\end{eqnarray}%
Obviously, such quantities depend also on $\mathbf{X}$, but as $\mathbf{X}$
will be kept fixed throughout most of the analysis, we may avoid to
explicitly specify this dependence in the notations. As $u_{n}$ is
increasing we can replace, in the definitions of $U_{n}^{Q^{n}}(a^{n}),$ $S^{%
\mathbf{Q}}(A)$ and $\Pi ^{\mathbf{Q}}(A)$ the inequality in the budget
constraint with an equality.

When a vector $\mathbf{Q}\in \mathscr{Q}$ is assigned, we can consider two
problems. First, for each $n$, $U_{n}^{Q^{n}}(a^{n})$ is the optimal value
of the classical one dimensional expected utility maximization problem with
random endowment $X^{n}$ under the budget constraint $E_{Q^{n}}[Y]\leq a^{n}$%
, determined by the real number $a^{n}$ and the valuation operator $%
E_{Q^{n}}[\cdot ]$ associated to $Q^{n}$. Second, if we interpret the
quantity $\sum_{n=1}^{N}u_{n}(\cdot )$ as the aggregated utility of the
system, then $\Pi ^{\mathbf{Q}}(A)$ is the maximal expected utility of the
whole system $\mathbf{X,}$ among all $\mathbf{Y}\in {\mathcal{L}}$
satisfying the overall budget constraint $\sum_{n=1}^{N}E_{Q^{n}}\left[ Y^{n}%
\right] \leq A$. Notice that in these problems the vector $\mathbf{Y}$ is
not required to belong to $\mathcal{C}_{\mathbb{R}}$, but only to the vector
space $\mathcal{L}$. We will show in Lemma \ref{lemmalinkpiS} the quite
obvious equality $S^{\mathbf{Q}}(A)=\Pi ^{\mathbf{Q}}(A).$

\section{On several notions of Equilibrium}

\label{seconseveralnotions}

\subsection{Pareto Allocation}

\begin{definition}
\label{defparetos}Given a set of feasible allocations $\mathscr{V}\subseteq 
\mathcal{L}$ and a vector $\mathbf{X}\in \mathcal{L}$, $\widehat{\mathbf{Y}}%
\in \mathscr{V}$ is a \emph{Pareto allocation for} $\mathscr{V}$ if 
\begin{equation}
\mathbf{Y}\in \mathscr{V}\text{\quad and\quad }\mathbb{E}\left[
u_{n}(X^{n}+Y^{n})\right] \geq \mathbb{E}\left[ u_{n}(X^{n}+\widehat{Y}^{n})%
\right] \text{ for all }n  \label{eqpareto}
\end{equation}%
imply $\mathbb{E}\left[ u_{n}(X^{n}+Y^{n})\right] =\mathbb{E}\left[
u_{n}(X^{n}+\widehat{Y}^{n})\right] $ for all $n$.
\end{definition}

In general Pareto allocations are not unique and, not surprisingly, the
following version of the First Welfare Theorem holds true. Define the
optimization problem 
\begin{equation}
\Pi (\mathscr{V}):=\sup_{Y\in \mathscr{V}}\sum_{n=1}^{N}\mathbb{E}\left[
u_{n}(X^{n}+Y^{n})\right] .  \label{piV}
\end{equation}

\begin{proposition}
\label{ABQ} Whenever $\widehat{\mathbf{Y}}\in \mathscr{V}$ is the unique
optimal solution of $\Pi (\mathscr{V})$, then it is a Pareto allocation for $%
\mathscr{V}$.
\end{proposition}

\begin{proof}
Let $\widehat{\mathbf{Y}}$ be optimal for $\Pi (\mathscr{V})$, so that $%
\mathbb{E}\left[ \sum_{n=1}^{N}u_{n}(X^{n}+\widehat{Y}^{n})\right] =\Pi (%
\mathscr{V}).$ Suppose that there exists $\mathbf{Y}$ such that (\ref%
{eqpareto}) holds true. As $\mathbf{Y}\in \mathscr{V}$ we have:%
\begin{equation*}
\mathbb{E}\left[ \sum_{n=1}^{N}u_{n}(X^{n}+\widehat{Y}^{n})\right] =\Pi (%
\mathscr{V})\geq \mathbb{E}\left[ \sum_{n=1}^{N}u_{n}(X^{n}+Y^{n})\right]
\geq \mathbb{E}\left[ \sum_{n=1}^{N}u_{n}(X^{n}+\widehat{Y}^{n})\right] ,
\end{equation*}%
by (\ref{eqpareto}). Hence also $\mathbf{Y}$ is an optimal solution to $\Pi (%
\mathscr{V})$. Uniqueness of the optimal solution implies $\mathbf{Y}=%
\widehat{\mathbf{Y}}$.
\end{proof}

\subsection{Systemic utility maximization\label{secMax}}

The next definition is the utility maximization problem, in the case of a
system of $N$ agents.

\begin{definition}
\label{def}Fix $\mathbf{Q}\in \mathscr{Q}$. The pair $(\mathbf{Y}_{\mathbf{X}%
},\mathbf{a}_{\mathbf{X}})\in \mathcal{L}\mathbf{\times }\mathbb{R}^{N}$ is
a $\mathbf{Q}-$\textbf{Optimal Allocation} with budget $A\in \mathbb{R}$ if

1) for each $n$, $Y_{\mathbf{X}}^{n}$ is optimal for $U_{n}^{Q^{n}}(a_{%
\mathbf{X}}^{n}),$

2) $\mathbf{a}_{\mathbf{X}}$ is optimal for $S^{\mathbf{Q}}(A),$

3) $\mathbf{Y}_{\mathbf{X}}\in \mathcal{L}$.
\end{definition}

Note that in the above definition the vector $\mathbf{Q}\in \mathscr{Q}$ is
exogenously assigned. Given a total budget $A\in \mathbb{R}$, the vector $%
\mathbf{a}_{\mathbf{X}}\in \mathbb{R}^{N}$ maximizes the systemic utility $%
\sum_{n=1}^{N}U_{n}^{Q^{n}}(a^{n})$ among all feasible $\mathbf{a\in }%
\mathbb{R}^{N}$ ( $\sum_{n=1}^{N}a^{n}\leq A$) and $Y_{\mathbf{X}}^{n}$
maximizes the single agent expected utility $\mathbb{E}\left[ u_{n}(X^{n}+Y)%
\right] $ among all feasible allocations $Y\in \mathcal{L}^{n}$ s.t. $%
E_{Q^{n}}[Y]\leq a_{\mathbf{X}}^{n}$. Since $\mathbf{Q}\in \mathscr{Q}$ is
given, the budget constraint $E_{Q^{n}}[Y]\leq a_{\mathbf{X}}^{n}$ is well
defined for all $\mathbf{Y}\in \mathcal{L}$ and we do not need additional
conditions of the form $\mathbf{Y}\in \mathcal{C}_{\mathbb{R}}$. A
generalization of the classical single agent utility maximization yields the
following existence result.

\begin{proposition}
\label{propUtmax}Under Assumption \ref{A00} (a) select $\mathscr{Q}=\{%
\mathbf{Q}\}$ for some $\mathbf{Q}\in \mathcal{Q}_{v}$ (see \eqref{defqv})
with $\mathbf{Q}\sim P$. Set $\mathcal{L}=L^{1}(Q^{1})\times \dots \times
L^{1}(Q^{N})$ and let $\mathbf{X} \in M^{\Phi}$ (see \eqref{orly-product}).
Then a $\mathbf{Q}-$Optimal Allocation exists.
\end{proposition}

\begin{proof}
The proof can be obtained with the same arguments employed in Section 4.2 \cite{bffm}. \end{proof}

Let $(\mathbf{Y}_{\mathbf{X}},\mathbf{a}_{\mathbf{X}})\in \mathcal{L}\mathbf{%
\times }\mathbb{R}^{N}$ be a $\mathbf{Q}-$Optimal Allocation. Due to Lemma %
\ref{lemmalinkpiS}, $\Pi ^{\mathbf{Q}}(A)=S^{\mathbf{Q}}(A)$ and 
\begin{eqnarray*}
\Pi ^{\mathbf{Q}}(A) &=&S^{\mathbf{Q}}(A)=\sup_{\mathbf{a}\in \mathbb{R}^{N}%
\text{, }\sum_{n=1}^{N}a^{n}=A}\sum_{n=1}^{N}\sup_{Y^{n}\in \mathcal{L}%
^{n}}\left\{ \mathbb{E}\left[ u_{n}(X^{n}+Y^{n})\right] \mid
E_{Q^{n}}[Y^{n}]=a^{n}\right\} \\
&=&\sum_{n=1}^{N}\sup_{Y^{n}\in \mathcal{L}^{n}}\left\{ \mathbb{E}\left[
u_{n}(X^{n}+Y^{n})\right] \mid E_{Q^{n}}[Y^{n}]=a_{X}^{n}\right\} ,
\end{eqnarray*}%
where we replaced the inequalities with equalities in the budget constraints, as $u_{n}$ are monotone.
Hence the systemic utility maximization problem $\Pi ^{\mathbf{Q}}(A)$ with
overall budget constraint $A$ reduces to the sum of $n$ single agent
maximization problems, where, however, the budget constraint of each agents
is assigned by $a_{\mathbf{X}}^{n}=E_{Q^{n}}[Y_{\mathbf{X}}^{n}]$ and the
vector $\mathbf{a}_{\mathbf{X}}$\ maximizes the overall performance of the
system. We will also recover this feature in the notion of a SORTE, where
the probability vector $\mathbf{Q}$ will be endogenously determined, instead
of being a priori assigned, as in this case.

\subsection{Risk Exchange Equilibrium\label{SecEq}}

We here formalize B\"{u}hlmann's risk exchange equilibrium in a pure
exchange economy, \cite{Buhlmann1} and \cite{Buhlmann}, already mentioned in
conditions (a') and (b'), Item 1 of the Introduction. Let $\mathcal{Q}^{1}$ be the set of vectors of probability
measures having all components equal:
\begin{equation*}
\mathcal{Q}^{1}:=\left\{ \mathbf{Q}\in \mathcal{P}^{N}\mid
Q^{1}=...=Q^{N}\right\} .
\end{equation*}

To be consistent
with Definition \ref{def} we keep the same numbering for the corresponding
conditions. %

\begin{definition}
\label{def4} Fix $A\in{\mathbb{R}},\,\mathbf{a}\in \mathbb{R}^{N}$ such that 
$\sum_{n=1}^{N}a^{n}=A$. The pair $(\mathbf{Y}_{\mathbf{X}},\mathbf{Q}_{%
\mathbf{X}})\in \mathcal{L}\mathbf{\times }\mathcal{Q}^{1}$ is a \textbf{%
risk exchange equilibrium (with budget $A$ and allocation $\mathbf{a}\in 
\mathbb{R}^{N}$)} if:

1) for each $n$, $Y_{\mathbf{X}}^{n}$ is optimal for $U_{n}^{Q_{\mathbf{X}%
}^{n}}(a^{n})$,

3) $\mathbf{Y}_{\mathbf{X}}\in \mathcal{C}_{\mathbb{R}}$, $\sum_{n=1}^{N}Y_{%
\mathbf{X}}^{n}=A$ $P\text{-}a.s.$
\end{definition}

\begin{theorem}[B\"{u}hlmann, \protect\cite{Buhlmann}]
\label{ThBuhlmann}For twice differentiable, concave, strictly increasing
utilities $u_{1},\dots ,u_{n}:{\mathbb{R}}\rightarrow {\mathbb{R}}$ such
that their risk aversions are positive Lipschitz and for $\mathcal{L}%
=(L^{\infty }(P))^{N},$ $\mathscr{Q}=\mathcal{Q}^{1}$ and $\mathbf{X} \in 
\mathcal{L }$, there exists a unique risk exchange equilibrium that is
Pareto optimal.
\end{theorem}

\begin{proof}
See \cite{Buhlmann}.
\end{proof}

In a risk exchange equilibrium with budget $A$, the vector $\mathbf{a}\in 
\mathbb{R}^{N}$ such that $\sum_{n=1}^{N}a^{n}=A$ is exogenously assigned,
while both the optimal exchange variable $\mathbf{Y}_{\mathbf{X}}$ and the
equilibrium price measure $Q_{\mathbf{X}}$ are endogenously determined. On
the contrary, in a $\mathbf{Q}-$Optimal Allocation the pricing measure is
assigned \emph{a priori}, while the optimal allocation $\mathbf{Y}_{\mathbf{X%
}}$ and optimal budget $\mathbf{a}_{\mathbf{X}}$ are endogenously
determined. We shall now introduce a notion which requires to endogenously
recover the triple $(\mathbf{Y}_{\mathbf{X}},\mathbf{Q}_{\mathbf{X}},\mathbf{%
a}_{\mathbf{X}})$ from the systemic budget $A$.

\subsection{Systemic Optimal Risk Transfer Equilibrium (SORTE)}

\label{SecSORTE} The novel equilibrium concept presented in equations (\ref%
{Problem}) (\ref{ProblemY}) and (\ref{BB}) can now be formalized as follows.
To this end, recall from (\ref{Cr}) the definition of $\mathcal{C}_{\mathbb{R%
}}$ and fix a convex cone 
\begin{equation*}
\mathcal{B}\subseteq \mathcal{C}_{{\mathbb{R}}}
\end{equation*}%
of admissible allocations such that ${\mathbb{R}}^{N}+\mathcal{B}=\mathcal{B}
$.

\begin{definition}[SORTE]
\label{carte}The triple $(\mathbf{Y}_{\mathbf{X}},\mathbf{Q}_{\mathbf{X}},%
\mathbf{a}_{\mathbf{X}})\in \mathcal{L}\mathbf{\times }\mathscr{Q}\mathbf{%
\times }\mathbb{R}^{N}$ is a \textbf{Systemic Optimal Risk Transfer
Equilibrium} with budget $A\in \mathbb{R}$ if:

1) for each $n$, $Y_{\mathbf{X}}^{n}$ is optimal for $U_{n}^{Q_{\mathbf{X}%
}^{n}}(a_{\mathbf{X}}^{n}),$

2) $\mathbf{a}_{\mathbf{X}}$ is optimal for $S^{\mathbf{Q}_{\mathbf{X}}}(A),$

3) $\mathbf{Y}_{\mathbf{X}}\in \mathcal{B}\subseteq \mathcal{C}_{\mathbb{R}}$
\ and $\sum_{n=1}^{N}Y_{\mathbf{X}}^{n}=A$ $P\text{-}a.s.$
\end{definition}

\begin{remark}
\label{remarkmonot} It follows from the monotonicity of each $u_{n}$ that $%
\sum_{n=1}^{N}a_{\mathbf{X}}^{n}=A$ and $E_{Q_{\mathbf{X}}^{n}}[Y_{\mathbf{X}%
}^{n}]=a_{\mathbf{X}}^{n}.$ Hence 
\begin{equation*}
\sum_{n=1}^{N}E_{Q_{\mathbf{X}}^{n}}[Y_{\mathbf{X}}^{n}]=\sum_{n=1}^{N}a_{%
\mathbf{X}}^{n}=A,
\end{equation*}%
and 
\begin{equation}
\sum_{n=1}^{N}Y_{\mathbf{X}}^{n}=\sum_{n=1}^{N}E_{Q_{\mathbf{X}}^{n}}[Y_{%
\mathbf{X}}^{n}] \text{ } P\text{-}a.s.   \label{eqformulaa}
\end{equation}
\end{remark}

The main aim of the paper is to provide sufficient general assumptions that
guarantee existence and uniqueness as well as good properties of a SORTE.

\begin{remark}
We will show the existence of a triple $(\mathbf{Y}_{\mathbf{X}},\mathbf{Q}_{%
\mathbf{X}},\mathbf{a}_{\mathbf{X}})\in \mathcal{L}\mathbf{\times }%
\mathscr{Q}\mathbf{\times }\mathbb{R}^{N}$ verifying the three conditions in
Definition \ref{carte}. Hence, we also obtain the existence of the SORTE in
the formulations given in (\ref{Problem2}), (\ref{Problem2a}), (\ref{BB}) or
in (\ref{Problem}), (\ref{ProblemY}), (\ref{BB}), for generic functional $%
p^{n}$ verifying the conditions (i), (ii) and (iii) stated in the
Introduction (see also Remark \ref{remQ}).
\end{remark}



In the sequel we will work under the following Assumption \ref{A00}.

\begin{assumption}
\label{A00} $\,$

\begin{enumerate}
\item[(a)] {\textbf{Utilities}: $u_{1},\dots ,u_{N}:{\mathbb{R}}\rightarrow {%
\mathbb{R}}$ are concave, strictly increasing differentiable functions with 
\begin{equation*}
\lim_{x\rightarrow -\infty }\frac{u_{n}(x)}{x}=+\infty
\,\,\,\,\,\lim_{x\rightarrow +\infty }\frac{u_{n}(x)}{x}=0, \text{    for any } n\in
\{1,\dots ,N\}.
\end{equation*}%
Moreover we assume that the following property holds: for any $n\in
\{1,\dots ,N\}$ and $Q^{n}\ll P$ 
\begin{equation}
\mathbb{E}\left[ v_{n}\left( \lambda \frac{\mathrm{d}Q^{n}}{\mathrm{d}P}%
\right) \right] <+\infty \text{ for some }\lambda
>0\,\,\,\Longleftrightarrow \,\,\,\mathbb{E}\left[ v_{n}\left( \lambda \frac{%
\mathrm{d}Q^{n}}{\mathrm{d}P}\right) \right] <+\infty \text{ for all }%
\lambda >0,  \label{RAE}
\end{equation}%
}where $v_{n}(y):=\sup_{x\in \mathbb{R}}\left\{ u_{n}(x)-xy\right\} $
denotes the convex conjugate of $u_{n}$.

\item[(b)] {\textbf{Constraints}: $\mathcal{B}\subseteq \mathcal{C}_{\mathbb{%
R}}$ is a convex cone, closed in probability, such that ${\mathbb{R}}^{N}+%
\mathcal{B}=\mathcal{B}$.}
\end{enumerate}
\end{assumption}

\begin{remark}
In particular, Assumptions \ref{A00} (b) implies that all constant vectors
belong to $\mathcal{B}$. The condition (\ref{RAE}) is related to the
Reasonable Asymptotic Elasticity condition on utility functions, which was
introduced in \cite{S01}. This assumption, even though quite weak (see \cite%
{bf} Section 2.2), is fundamental to guarantee the existence of the optimal
solution to classical utility maximization problems (see \cite{bf} and \cite%
{S01}).
\end{remark}

\begin{theorem}
\label{TH1} A \textbf{Systemic
Optimal Risk Transfer Equilibrium } $(\mathbf{Y}_{\mathbf{X}},\mathbf{Q}_{%
\mathbf{X}}, \mathbf{a}_{\mathbf{X}})$ exists, with ${Q}^1_{\mathbf{X}%
},\dots,{Q}^N_{\mathbf{X}}$ equivalent to $P$.
\end{theorem}

\begin{theorem}
\label{TH2}Under the
additional Assumption that $\mathcal{B}$ is closed under truncation
(Definition \ref{DefTrunc}) the \textbf{Systemic Optimal Risk Transfer
Equilibrium } is unique and is a Pareto optimal allocation.
\end{theorem}

The formal statements and proofs are postponed to Section \ref{secproof},
Theorem \ref{thmsorteexistsA} and Theorem \ref{thmsorteuniqueA}.

\begin{remark}
A priori there are no reasons why a $\mathbf{Q}$-optimal allocation $\mathbf{%
Y}_{\mathbf{X}}$ in Definition \ref{def} would also satisfy the constraint $%
\sum_{n=1}^{N}Y_{\mathbf{X}}^{n}\in \mathbb{R}$. The existence of a SORTE is
indeed the consequence of the existence of a probability measure $\mathbf{Q}%
_{\mathbf{X}}$ such that the $\mathbf{Q}_{\mathbf{X}}$-optimal allocation $%
\mathbf{Y}_{\mathbf{X}}$ in Definition \ref{def} satisfies also the
additional risk transfer constraint $\sum_{n=1}^{N}Y_{\mathbf{X}}^{n}=A$ $P%
\text{-}a.s.$ .
\end{remark}

\begin{remark}
Without the additional feature expressed by 2) in the Definition \ref{carte}%
, for all choices of $\mathbf{a}_{\mathbf{X}}$ satisfying $\sum_{n=1}^{N}a_{%
\mathbf{X}}^{n}=A$ there exists an equilibrium $(\mathbf{Y}_{\mathbf{X}},%
\mathbf{Q}_{\mathbf{X}})$ in the sense of Definition \ref{def4} (see Theorem %
\ref{ThBuhlmann}). The uniqueness of a SORTE is then a consequence of the
uniqueness of the optimal solution in condition 2).
\end{remark}

\begin{remark}
Depending on which one of the three objects $(\mathbf{Y},\mathbf{Q},\mathbf{a%
})\in \mathcal{L}\mathbf{\times }\mathscr{Q}\times \mathbb{R}^{N}$ we keep a
priori fixed, we get a different notion of equilibrium (see the various
definitions above). The characteristic features of the risk exchange
equilibriums and of a SORTE, compared with the more classical utility
optimization problem in the systemic framework of Section \ref{secMax}, are
the condition $\sum_{n=1}^{N}Y_{\mathbf{X}}^{n}=A$ $P\text{-}a.s.$ and the
existence of the equilibrium pricing vector $\mathbf{Q}_{\mathbf{X}}$.
\end{remark}

For both concepts of equilibrium (Definitions \ref{def4} and SORTE), each
agent is behaving rationally by maximizing his expected utility given a
budget constraint. The two approaches differ by the budget constraints. In B%
\"{u}hlmann's definition the vector $\mathbf{a}\in \mathbb{R}^{N}$ that
assigns the budget constraint ($E_{Q_{\mathbf{X}}^{n}}[Y^{n}]\leq a_{n})$ is
prescribed a priori$\mathbf{.}$ On the contrary, in the SORTE approach, the
vector $\mathbf{a}\in \mathbb{R}^{N}$, with $\sum_{n=1}^{N}a_{n}=A,$ that
assigns the budget constraint $E_{Q_{\mathbf{X}}^{n}}[Y^{n}]\leq a_{n}$ is
determined by optimizing the problem in condition 2), hence by taking into
account the optimal systemic utility $S^{\mathbf{Q}_{\mathbf{X}}}(A)$, which
is (by definition) larger than the systemic utility $%
\sum_{n=1}^{N}U_{n}^{Q_{X}^{n}}(a^{n})$ in B\"{u}hlmann's equilibrium. The
SORTE\ gives priority to the systemic aspects of the problem in order to
optimize the overall systemic performance. A toy example showing the
difference between a B\"{u}hlmann's equilibrium and a SORTE is provided in
Section \ref{Ex}.

\begin{example}
\label{exCh}We now consider the example of a cluster of agents, already
introduced in \cite{bffm}. For $h\in \left\{ 1,\cdots ,N\right\} ,$ let $%
\mathbf{I}:=(I_{m})_{m=1,...,h}$ be some partition of $\{1,\cdots ,N\}$. We
introduce the following family{\small 
\begin{equation}
\mathcal{B}^{(\mathbf{I})}=\left\{ \mathbf{Y}\in L^{0}(\mathbb{R}^{N})\mid \
\exists \ d=(d_{1},\cdots ,d_{h})\in \mathbb{R}^{h}\text{\ }:\text{\ }%
\sum_{i\in I_{m}}Y^{i}=d_{m}\text{ for\ }m=1,\cdots ,h\right\} \subseteq 
\mathcal{C}_{\mathbb{R}}.  \label{C0}
\end{equation}%
}For a given $\mathbf{I}$, the values $(d_{1},\cdots ,d_{h})$ may change,
but the elements in each of the $h$ groups $I_{m}$ is fixed by the partition 
$\mathbf{I}$. It is then easily seen that $\mathcal{B}^{(\mathbf{I})}$ is a
linear space containing $\mathbb{R}^{N}$ and closed with respect to
convergence in probability. We point out that the family $\mathcal{B}^{(%
\mathbf{I})}$ admits two extreme cases:

\begin{itemize}
\item[(i)] the strongest restriction occurs when $h=N,$ i.e., we consider
exactly $N$ groups, and in this case $\mathcal{B}^{(\mathbf{I})}=\mathbb{R}%
^{N}$ corresponds to no risk sharing;

\item[(ii)] on the opposite side, we have only one group $h=1$ and $\mathcal{%
B}^{(\mathbf{I})}=\mathcal{C}_{\mathbb{R}}$ is the largest possible class,
corresponding to risk sharing among all agents in the system. This is the
only case considered in B\"{u}hlmann's definition of equilibrium.
\end{itemize}
\end{example}

\begin{remark}
As already mentioned in the Introduction, one additional feature of a SORTE,
compared with the B\"{u}hlmann's notion, is the possibility to require, in
addition to $\sum_{n=1}^{N}Y^{n}=A$ that the optimal solution belongs to a
pre-assigned set $\mathcal{B}$ of admissible allocations, satisfying
Assumption \ref{A00} (b). In particular, we allow for the selection of the
sets $\mathcal{B}=\mathbb{R}^{N}$ or $\mathcal{B}=\mathcal{C}_{\mathbb{R}}$.
The characteristics of the optimal probability $\mathbf{Q}_{\mathbf{X}}$
depend on the admissible set $\mathcal{B}$. For $\mathcal{B}=\mathcal{C}_{%
\mathbb{R}}$, all the components of $\mathbf{Q}_{\mathbf{X}}$ turn out to be
equal. We also know (see Lemma \ref{lemmaalleqonf}) that for $\mathcal{B}=%
\mathcal{B}^{(\mathbf{I})}$ all the components $Q_{\mathbf{X}}^{i}$ of $%
\mathbf{Q}_{\mathbf{X}}$ are equal for all $i\in I_{m}$, for each group $%
I_{m}$. Additional examples of sets $\mathcal{B}$ are provided in Section %
\ref{SecXB}.
\end{remark}

\subsection{Explicit Formulas in the Exponential Case} \label{expnew}
We believe it is now instructive to anticipate the explicit solution to the SORTE problem in the exponential case for $\mathcal{B}=\mathcal{C}_\R$. This is a particular case of a more general situation treated in detail in Section \ref{sectionexp}. 
\begin{theorem}
Take exponential utilities
\begin{equation*}
u_{n}(x):=1-\exp (-\alpha _{n}x),\,n=1,\dots ,N\,\,\,\,\,\text{ for }%
\,\,\,\,\,\alpha _{1},\dots ,\alpha _{N}>0.  
\end{equation*}%

Then the SORTE for $\mathcal{B}=\mathcal{C}_\R$ is given by 
\begin{equation}
\begin{cases}
\widehat{Y}^{k}=-X^{k}+\frac{1}{\alpha _{k}}\left( \frac{\overline{X}}{%
\beta }\right) +\frac{1}{\alpha _{k}}\left[ \frac{A}{%
\beta }+\ln \left( \alpha _{k}\right) -E_{R}\left[ \ln (\alpha )\right] %
\right] & \,\,\,\,\,k=1,\dots,N \\ 
\frac{\mathrm{d}\widehat{Q}^{k}}{\mathrm{d}P}=\frac{\exp \left( -\frac{%
\overline{X}}{\beta}\right) }{\mathbb{E}\left[ \exp \left( -\frac{%
\overline{X}}{\beta}\right) \right] }=:\frac{\mathrm{d}\widehat{Q}%
}{\mathrm{d}P} & \,\,\,\,\,k=1,\dots,N \\ 
\widehat{a}^{k}=E_{\widehat{Q}^{k}}[\widehat{Y}^{k}] & \,\,\,\,\,k=1,\dots ,N%
\end{cases}%
\end{equation}%
where $\beta:=\sum_{n=1}^N\frac{1}{\alpha _{n}}$, $\overline{X}%
:=\sum_{n=1}^NX^{n}$, $R(n):=\frac{\frac{1}{\alpha _{n}}}{\sum_{k=1}^{N}\frac{1}{\alpha _{k}}}$ for $n=1,...N$, $\alpha :=(\alpha _{1},...,\alpha _{N})$, $E_{R}\left[ \ln (\alpha )\right] =\sum_{n=1}^{N}R(n)\ln (\alpha _{n})$.

\end{theorem}

\section{Proof of Theorem \protect\ref{TH1} and Theorem \protect\ref{TH2}}

\label{secproof}

We need to introduce the following concepts and notations:

\begin{enumerate}
\item The u{tility functions in Assumption \ref{A00} induce an Orlicz Space
structure: see Appendix \ref{secorlicz} for the details and the definitions
of the functions }$\Phi ${\ and }$\Phi ^{\ast }${, the Orlicz space }$%
L^{\Phi }${\ and the Orlicz Heart }$M^{\Phi }$. Here we just recall the
following inclusions among the Banach Spaces $L^{\infty }(\mathbf{P}%
)\subseteq M^{\Phi }\subseteq L^{\Phi }\subseteq L^{1}(\mathbf{P})$ and that 
$\frac{d\mathbf{Q}}{d\mathbb{P}}\in L^{\Phi ^{\ast }}$ implies $L^{\Phi
}\subseteq L^{1}(\mathbf{Q})$. From now on we assume that $\mathbf{X} \in
M^{\Phi}$.

\item {For any $A\in {\mathbb{R}}$ we set }%
\begin{equation*}
{\mathcal{B}_{A}:=\mathcal{B}\cap \{\mathbf{Y}\in (L^{0}(P))^{N}\mid
\sum_{n=1}^{N}Y^{n}\leq A\text{ }P\text{-}a.s.\}.}
\end{equation*}%
{Observe that $\mathcal{B}_{0}\cap M^{\Phi }$ is a convex cone. }

\item We {introduce the following problem for $\mathbf{X}\in M^{\Phi }$ and
for a vector of probability measures $\mathbf{Q}\ll P,$ with $\frac{\mathrm{d%
}\mathbf{Q}}{\mathrm{d}P}\in L^{\phi ^{\ast }}$, }


\begin{equation}
\pi ^{\mathbf{Q}}(A):=\sup \left\{ \sum_{n=1}^{N}\mathbb{E}\left[
u_{n}\left( X^{n}+Y^{n}\right) \right] \,\middle|\,\mathbf{Y}\in M^{\Phi
},\,\sum_{n=1}^{N}E_{{Q}^{n}}\left[ Y^{n}\right] \leq A\right\} .
\label{defpiQ}
\end{equation}%
Notice that in ( \ref{defpiQ}) the vector $\mathbf{Y}$ is not required to
belong to $\mathcal{C}_{\mathbb{R}}$, but only to the vector space $M^{\Phi
} $. In order to show the existence of the optimal solution to the problem $%
\pi ^{\mathbf{Q}}(A)$, it is necessary to enlarge the domain in (\ref{defpiQ}%
).

\item {$\mathcal{Q}$ is the set of vectors of probability measures defined
by 
\begin{equation*}
\mathcal{Q}:=\left\{ {\mathbf{Q}\ll P}\middle|\left[ \frac{\mathrm{d}Q^{1}}{%
\mathrm{d}P},\dots ,\frac{\mathrm{d}Q^{N}}{\mathrm{d}P}\right] \in L^{\Phi
^{\ast }},\,\sum_{n=1}^{N}\mathbb{E}\left[ Y^{n}\frac{\mathrm{d}Q^{n}}{%
\mathrm{d}P}\right] \leq 0,\,\forall \,Y\in \mathcal{B}_{0}\cap M^{\Phi
}\right\} .
\end{equation*}%
Identifying Radon-Nikodym derivatives and measures in the natural way, $%
\mathcal{Q}$ turns out to be the set of normalized (i.e. with componentwise
expectations equal to $1$), non negative vectors in \emph{the polar of} $%
\mathcal{B}_{0}\cap M^{\Phi }$ \emph{in the dual system} $(M^{\Phi },L^{\Phi
^{\ast }})$. In our }$N$-dimensional {systemic one-period setting, the set }$%
\mathcal{Q}$ plays the same crucial role as the set of martingale measures
in multiperiod stochastic securities markets.

\item {We introduce the following convex subset of $\mathcal{Q}$: 
\begin{equation}
\mathcal{Q}_{v}:=\mathcal{Q}\,\,\cap \,\,\left\{ \left[ \frac{\mathrm{d}Q^{1}%
}{\mathrm{d}P},\dots ,\frac{\mathrm{d}Q^{N}}{\mathrm{d}P}\right] \in L^{\phi
^{\ast }}\,\middle|\,\frac{\mathrm{d}Q^{n}}{\mathrm{d}P}\geq 0\,\forall
\,n\in \{1,\dots ,N\},\,\sum_{n=1}^{N}\mathbb{E}\left[ v_{n}\left( \frac{%
\mathrm{d}Q^{n}}{\mathrm{d}P}\right) \right] <+\infty \right\} .
\label{defqv}
\end{equation}%
}

\item {Set 
\begin{equation}
\mathcal{L}:=\bigcap_{\mathbf{Q}\in \mathcal{Q}_{v}}L^{1}(Q^{1})\times \dots
\times L^{1}(Q^{N}),\,\,\,\,\,\,\,\,\,\,\,\,\,\,\,\mathscr{Q}:=\mathcal{Q}%
_{v}.  \label{defLQ}
\end{equation}%
}
\end{enumerate}

{Note that }$M^{\Phi }\subseteq \mathcal{L}$ and that $\mathcal{L}$ has the product structure $\mathcal{L}=\mathcal{L}^1\times\dots\times\mathcal{L}^N$: let $\text{Proj}_n$ denote the projection on the $n-$th component, defined on $\mathcal{Q}_v$, and take the corresponding image $\mathcal{Q}_n:=\text{Proj}_n(\mathcal{Q}_v)$ (consisting of a family of probability measures, all absolutely continuous with respect to $\probp$). Set $\mathcal{L}^n:=\bigcap_{Q\in\mathcal{Q}_n}L^1(Q)$. Then $\mathcal{L}=\mathcal{L}^1\times\dots\times\mathcal{L}^N$.

 We will consider the
optimization problems \eqref{UQ},\thinspace \eqref{SQ} and \eqref{PiQ} with
the particular choice of $(\mathcal{L},\mathscr{Q})$ in (\ref{defLQ}) and
will show that, with such choice, $\pi ^{\mathbf{Q}}(A)=\Pi ^{\mathbf{Q}%
}(A). ${\ Observe that if all utilities are bounded from above, the
requirement $\sum_{n=1}^{N}\mathbb{E}\left[ v_{n}\left( \frac{\mathrm{d}Q^{n}%
}{\mathrm{d}P}\right) \right] <+\infty $ is redundant, but it becomes
important if we allow utilities to be unbounded. } 

We also require some additional definitions and notations:

\begin{enumerate}
\item[a)] {$\overline{\mathcal{B}_{0}}$ is the polar of the cone $co(%
\mathcal{Q}_{v})$ in the dual pair 
\begin{equation*}
\left( L^{\Phi _{1}^{\ast }}\times \dots \times L^{\Phi _{N}^{\ast
}}\,,\,\bigcap_{\mathbf{Q}\in \mathcal{Q}_{v}}L^{1}(Q^{1})\times \dots
\times L^{1}(Q^{N})\right) ,
\end{equation*}%
that is 
\begin{equation*}
\overline{\mathcal{B}_{0}}:=\left\{ \mathbf{Y}\in \bigcap_{\mathbf{Q}\in 
\mathcal{Q}_{v}}L^{1}(Q^{1})\times \dots \times L^{1}(Q^{N})\,\,\middle%
|\,\,\sum_{n=1}^{N}E_{{Q}^{n}}\left[ Y^{n}\right] \leq 0,\,\forall \,\mathbf{%
Q}\in \mathcal{Q}_{v}\right\} .
\end{equation*}%
}It is easy to verify that 
\begin{equation*}
\mathcal{B}_{0}\cap M^{\Phi }\subseteq \overline{\mathcal{B}_{0}}.
\end{equation*}

\item[b)] {For any $A\in {\mathbb{R}}$ we define $\overline{\mathcal{B}_{A}}$
as the set 
\begin{equation*}
\overline{\mathcal{B}_{A}}:=\left\{ \mathbf{Y}\in \bigcap_{Q\in \mathcal{Q}%
_{v}}L^{1}(Q^{1})\times \dots \times L^{1}(Q^{N})\,\,\middle%
|\,\,\sum_{n=1}^{N}E_{{Q}^{n}}\left[ Y^{n}\right] \leq A,\,\forall \,\mathbf{%
Q}\in \mathcal{Q}_{v}\right\}
\end{equation*}%
}We will prove that $\overline{\mathcal{B}_{A}}$ is the correct enlargement
of the domain $\mathcal{B}_{A}\cap M^{\Phi }$ in order to obtain the
existence of the optimal solution of the primal problem.

\item[c)] $\left\{ \mathbf{e_{i}}\right\} _{i=1,...,N}$ is the canonical
base of ${\mathbb{R}}^{N}$.
\end{enumerate}

\begin{lemma}
\label{rempolarisnice}In the dual pair $(M^{\Phi },L^{\Phi ^{\ast }}),$
consider the polar $(\mathcal{B}_{0}\cap M^{\Phi })^{0}$ of $\mathcal{B}%
_{0}\cap M^{\Phi }$. Then $(\mathcal{B}_{0}\cap M^{\Phi })^{0}\cap
(L_{+}^{0})^{N}$ is the cone generated by $\mathcal{Q}$.
\end{lemma}

\begin{proof}
From the definition of $\mathcal{B}_{0}$ and the fact that $\mathcal{B}$
contains all constant vectors, we may conclude that all vectors in ${\mathbb{%
R}}^{N}$ of the form $\mathbf{e_{i}}-\mathbf{e_{j}}$ belong to $\mathcal{B}%
_{0}\cap M^{\Phi }$. Then for all $\mathbf{Z}\in (\mathcal{B}_{0}\cap
M^{\Phi })^{0}$ and for all $i,j\in \{1,\dots ,N\}$ we must have: $\mathbb{E}%
\left[ Z^{i}\right] -\mathbb{E}\left[ Z^{j}\right] \leq 0$. As a
consequence, $\mathbf{Z}\in (\mathcal{B}_{0}\cap M^{\Phi })^{0}$ implies $%
\mathbb{E}\left[ Z^{1}\right] =\dots =\mathbb{E}\left[ Z^{N}\right] $ and so 
\begin{equation}
(\mathcal{B}_{0}\cap M^{\Phi })^{0}\cap (L_{+}^{0})^{N}={\mathbb{R}}%
_{+}\cdot \mathcal{Q},  \label{polarityisnice}
\end{equation}%
where ${\mathbb{R}}_{+}:=\{b\in {\mathbb{R}},b\geq 0\}$.
\end{proof}

\begin{lemma}
\label{remqvok}$\mathcal{Q}_{v}^{e}:=\{\mathbf{Q}\in \mathcal{Q}_{v}\text{
s.t. }\mathbf{Q}\sim P\}\neq \varnothing .$
\end{lemma}

\begin{proof}
The condition $\mathcal{B}\subseteq \mathcal{C}_{\mathbb{R}}$ implies $%
\mathcal{B}_{0}\cap M^{\Phi }\subseteq (\mathcal{C}_{\mathbb{R}}\cap M^{\Phi
}\cap \{\sum_{n=1}^{N}Y^{n}\leq 0\})$, so that the polars satisfy the
opposite inclusion: $(\mathcal{C}_{\mathbb{R}}\cap M^{\Phi }\cap
\{\sum_{n=1}^{N}Y^{n}\leq 0\})^{0}\subseteq (\mathcal{B}_{0}\cap M^{\Phi
})^{0}$. Observe now that any vector $(Z,\dots ,Z)$, for $Z\in L_{+}^{\infty
}$, belongs to $(\mathcal{C}_{\mathbb{R}}\cap M^{\Phi }\cap
\{\sum_{n=1}^{N}Y^{n}\leq 0\})^{0}$. In particular, $(\mathcal{B}_{0}\cap
M^{\Phi })^{0}$ contains vectors in the form $\left(\frac{\varepsilon +Z}{%
1+\varepsilon },\dots ,\frac{\varepsilon +Z}{1+\varepsilon }\right)$ with $%
\varepsilon >0$ and $Z\in L_{+}^{\infty }$, $\mathbb{E}\left[ Z\right] =1$.
Each component of such a vector has expectation equal to $1$, belongs to $%
L_{+}^{\infty }$ and satisfies $\frac{\varepsilon +Z}{1+\varepsilon }\geq 
\frac{\varepsilon }{1+\varepsilon }$. All these conditions imply that there
exists a probability vector $\mathbf{Q\in }\mathcal{Q}$ such that $\frac{%
\mathrm{d}\mathbf{Q}}{\mathrm{d}P}>0$ $P-a.s.$ with $\sum_{n=1}^{N}\mathbb{E}%
\left[ v_{n}\left( \frac{\mathrm{d}Q^{n}}{\mathrm{d}P}\right) \right]
<\infty $, hence $\mathcal{Q}_{v}^{e}\neq \varnothing $.
\end{proof}

\subsection{Scheme of the proof\label{SecScheme}}

The proof of Theorem \ref{TH1} is inspired by the classical duality theory
in utility maximization, see for example \cite{CK} and \cite{KLSX} and by
the minimax approach developed in \cite{BF02}. More precisely, our road map
will be the following:

\begin{enumerate}
\item First we show, in Remark \ref{remarkfrom0toA}, how we may reduce the
problem to the case $A=0.$

\item {We consider }%
\begin{equation}
\pi (A):=\sup \left\{ \sum_{n=1}^{N}\mathbb{E}\left[ u_{n}\left(
X^{n}+Y^{n}\right) \right] \,\middle|\,\mathbf{Y}\in M^{\Phi }\cap \mathcal{B%
},\,\sum_{n=1}^{N}\ Y^{n}\leq A\text{ }P\text{-}a.s.\right\}.  \label{defpi}
\end{equation}%
In Theorem \ref{thmoptimumexists} we specialize the duality, obtained in
Theorem \ref{thmminimax} for a generic convex cone $\mathcal{C}$, for the
maximization problem $\pi (0)$ over the convex cone $\mathcal{C}=${$\mathcal{%
B}_{0}\cap $}$M^{\Phi }$ and prove: (i) the existence of the optimizer $%
\widehat{\mathbf{Y}}$ of $\pi (0)$, which belongs to {$\mathcal{B}_{0}$};
(ii) the existence of the optimizer $\widehat{\mathbf{Q}}$ to the dual
problem of $\pi (0)$. Here we need all the assumptions on the utility
functions and on the set {$\mathcal{B}$ and an auxiliary result stated in
Theorem \ref{thmweirdclosure} in Appendix.}

\item Proposition \ref{propfairprob} will show that also the elements in the
closure of {$\mathcal{B}\cap $}$M^{\Phi }$ satisfy the key condition $%
\sum_{n=1}^{N}E_{{Q}^{n}}\left[ Y^{n}\right] \leq \sum_{n=1}^{N}Y^{n}\in 
\mathbb{R}$ for all $Q\in \mathscr{Q}$.

\item Theorem \ref{thmminimax} is then again applied, to a different set $%
\mathcal{C=}\left\{ \mathbf{Y}\in M^{\Phi }\mid \sum_{n=1}^{N}E_{{Q}^{n}}%
\left[ Y^{n}\right] \leq 0\right\} $, to derive Proposition \ref{proppiq},
which establishes the duality for $\pi ^{\mathbf{Q}}(0)$ and $\pi ^{\mathbf{Q%
}}(A)$ in case a fixed probability vector $\mathbf{Q}$ is assigned.

\item The minimax duality: 
\begin{equation*}
\pi (A)=\min_{\mathbf{Q}\in \mathcal{Q}_{v}}\pi ^{\mathbf{Q}}(A)=\pi ^{%
\widehat{\mathbf{Q}}}(A),
\end{equation*}%
is then a simple consequence of the above results (see Corollary \ref%
{propminimum}). This duality is the key tool to prove the existence of a
SORTE (see Theorem \ref{thmsorteexistsA}).

\item Uniqueness and Pareto optimality are then proved in Theorem \ref%
{thmsorteuniqueA}.
\end{enumerate}

\begin{remark}
\label{remQ}Notice that in the definition of $\pi (A)$ there is no reference
to a probability vector $\mathbf{Q}$. However, the optimizer of the dual
formulation of $\pi (A)$ is a probability vector $\widehat{\mathbf{Q}}$
(that will be the equilibrium pricing vector in the SORTE). Even if in the
equations (\ref{Problem}), (\ref{ProblemY}), (\ref{BB}) we do not a priori
require pricing functional of the form $p^{n}(\cdot )=E_{Q^{n}}[\cdot ]$,
this particular linear expression naturally appears from the dual
formulation.
\end{remark}

\subsection{Minimax Approach}

\begin{remark}
\label{remarkfrom0toA}Only in this Remark, we need to change the notation a
bit: we make the dependence of our maximization problems on the initial
point explicit. To this end we will write%
\begin{equation*}
\pi _{\mathbf{X}}(A):=\sup \left\{ \sum_{j=1}^{N}\mathbb{E}\left[
u_{j}\left( X^{j}+Y^{j}\right) \right] \,\middle|\,\mathbf{Y}\in \mathcal{B}%
_{A}\cap M^{\Phi }\right\}\,,
\end{equation*}%
\begin{equation*}
\pi _{\mathbf{X}}^{\mathbf{Q}}(A):=\sup \left\{ \sum_{j=1}^{N}\mathbb{E}%
\left[ u_{j}\left( X^{j}+Y^{j}\right) \right] \,\middle|\,\mathbf{Y}\in
M^{\Phi },\,\sum_{j=1}^{N}E_{{Q}^{j}}\left[ Y^{j}\right] \leq A\right\} \,.
\end{equation*}%
It is possible to reduce the maximization problem expressed by $\pi _{%
\mathbf{X}}(A)$ (and similarly by $\pi _{\mathbf{X}}^{Q}(A)$) to the problem
related to $\pi _{\cdot }(0)$ (respectively, $\pi _{\cdot }^{Q}(0))$ by
using the following simple observation: for any $\mathbf{a_{0}}\in {\mathbb{R%
}}^{N} $ with $\sum_{j=1}^{N}a_{0}^{j}=A$ consider 
\begin{eqnarray*}
\pi _{\mathbf{X}}(A) &=&\sup \left\{ \sum_{j=1}^{N}\mathbb{E}\left[
u_{j}\left( X^{j}+Y^{j}\right) \right] \middle|\mathbf{Y}\in \mathcal{B}\cap
M^{\Phi },\sum_{j=1}^{N}Y^{j}\leq A\right\} \\
&=&\sup \left\{ \sum_{j=1}^{N}\mathbb{E}\left[ u_{j}\left(
X^{j}+a_{0}^{j}+(Y^{j}-a_{0}^{j})\right) \right] \middle|\mathbf{Y}\in 
\mathcal{B}\cap M^{\Phi },\sum_{j=1}^{N}\left( Y^{j}-a_{0}^{j}\right) \leq
0\right\} \\
&=&\sup \left\{ \sum_{j=1}^{N}\mathbb{E}\left[ u_{j}\left(
X^{j}+a_{0}^{j}+Z^{j}\right) \right] \middle|\mathbf{Z}\in \mathcal{B}%
_{0}\cap M^{\Phi }\right\} ,
\end{eqnarray*}%
where last equality holds as we are assuming that ${\mathbb{R}}^{N}+\mathcal{%
B}=\mathcal{B}$. The last line represents the original problem, but with $%
A=0 $ and a different initial point. This fact will be used in the
conclusion of the proof of Theorem \ref{thmoptimumexists}.
\end{remark}

In the following Theorem we follow a minimax procedure inspired by the
technique adopted in \cite{bf}.

\begin{theorem}
\label{thmoptimumexists} 

Under Assumption \ref{A00} we have 
\begin{equation}
\pi (A):=\sup_{\mathbf{Y}\in \mathcal{B}_{A}\cap M^{\Phi }}\sum_{j=1}^{N}%
\mathbb{E}\left[ u_{j}\left( X^{j}+Y^{j}\right) \right] =\max_{\mathbf{Y}\in 
\overline{\mathcal{B}_{A}}}\sum_{j=1}^{N}\mathbb{E}\left[ u_{j}\left(
X^{j}+Y^{j}\right) \right]  \label{eqminimaxapplied2}
\end{equation}%
\begin{equation}
=\min_{\mathbf{Q}\in \mathcal{Q}}\min_{\lambda \in {\mathbb{R}}_{++}}\left(
\lambda \left( \sum_{j=1}^{N}E_{{Q}^{j}}\left[ X^{j}\right] +A\right)
+\sum_{j=1}^{N}\mathbb{E}\left[ v_{j}\left( \lambda \frac{\mathrm{d}Q^{j}}{%
\mathrm{d}P}\right) \right] \right) .  \label{eqminimaxapplied3}
\end{equation}%
The minimization problem in \eqref{eqminimaxapplied3} admits a unique
optimum $(\widehat{\lambda },\widehat{\mathbf{Q}})\in {\mathbb{R}}%
_{++}\times \mathcal{Q}$ with $\widehat{\mathbf{Q}}\sim P$. The maximization
problem in (\ref{eqminimaxapplied2}) admits a unique optimum $\widehat{%
\mathbf{Y}}\in \overline{\mathcal{B}_{A}}$, given by 
\begin{equation}
\widehat{Y}^{j}=-X^{j}-v_{j}^{\prime }\left( \widehat{\lambda }\frac{\mathrm{%
d}\widehat{Q}^{j}}{\mathrm{d}P}\right) ,\text{ }j=1,...,N\text{,}  \label{YY}
\end{equation}%
which belongs to $\mathcal{B}_{A}$. In addition,%
\begin{equation}
\sum_{j=1}^{N}E_{\widehat{Q}^{j}}\left[ \widehat{Y}^{j}\right] =A\quad \text{%
and\quad }\sum_{j=1}^{N}E_{Q^{j}}\left[ \widehat{Y}^{j}\right] \leq A\quad
\forall \mathbf{Q}\in \mathcal{Q}_{v}\text{.}  \label{EQ}
\end{equation}
\end{theorem}

\begin{proof}
$\,$ We first prove the result for the case $A=0$.

\textbf{STEP 1}

We first show that 
\begin{equation}
\sup_{\mathcal{B}_{0}\cap M^{\Phi }}\sum_{j=1}^{N}\mathbb{E}\left[
u_{j}\left( X^{j}+Y^{j}\right) \right] <\sum_{j=1}^{N}v_{j}(0)=%
\sum_{j=1}^{N}u_{j}(+\infty )\,\,\,\forall \,\mathbf{X}\in \,M^{\Phi }
\label{eqlessthansum}
\end{equation}%
so that we will be able to apply Theorem \ref{thmminimax} with the choice $%
\mathcal{C}:=\mathcal{B}_{0}\cap M^{\Phi }$. We distinguish two possible
cases: $\sum_{j=1}^Nu_j(+\infty)=+\infty$ or $\sum_{j=1}^Nu_j(+\infty)<+%
\infty$.

For $\sum_{j=1}^Nu_j(+\infty)=+\infty$: observe that for any $\mathbf{Q}\in 
\mathcal{Q}_{v}$ (which is nonempty by Lemma \ref{remqvok}) and $\lambda >0$
we have 
\begin{eqnarray*}
\sum_{j=1}^{N}\mathbb{E}\left[ u_{j}\left( X^{j}+Y^{j}\right) \right] &\leq
&\sum_{j=1}^{N}\mathbb{E}\left[ (X^{j}+Y^{j})\left( \lambda \frac{\mathrm{d}%
Q^{j}}{\mathrm{d}P}\right) \right] +\sum_{j=1}^{N}\mathbb{E}\left[
v_{j}\left( \lambda \frac{\mathrm{d}Q^{j}}{\mathrm{d}P}\right) \right] \\
&\leq &\sum_{j=1}^{N}\mathbb{E}\left[ X^{j}\left( \lambda \frac{\mathrm{d}%
Q^{j}}{\mathrm{d}P}\right) \right] +\sum_{j=1}^{N}\mathbb{E}\left[
v_{j}\left( \lambda \frac{\mathrm{d}Q^{j}}{\mathrm{d}P}\right) \right] .
\end{eqnarray*}%
We exploited above Fenchel's Inequality and the definition of $\mathcal{Q}%
_{v}$. Observing that the last line does not depend on $\mathbf{Y}$ and is
finite, and using the well known relation $v_j(0)=u_j(+\infty),j=1,\dots,N$,
we conclude that 
\begin{equation*}
\sup_{\mathcal{B}_{0}\cap M^{\Phi }}\sum_{j=1}^{N}\mathbb{E}\left[%
u_{j}\left( X^{j}+Y^{j}\right) \right] <+\infty=\sum_{j=1}^{N}v_{j}(0).
\end{equation*}

For $\sum_{j=1}^Nu_j(+\infty)<+\infty$: if the inequality in %
\eqref{eqlessthansum} were not strict, for any maximizing sequence $(\mathbf{%
Y}_{m})_{m}$ we would have, by monotone convergence, that 
\begin{equation*}
\sum_{j=1}^{N}\mathbb{E}\left[ u_{j}\left( +\infty \right) \right]
-\sum_{j=1}^{N}\mathbb{E}\left[ u_{j}\left( X^{j}+Y_{m}^{j}\right) \right] =%
\mathbb{E}\left[ \left\vert \sum_{j=1}^{N}\left( u_{j}(+\infty
)-u_{j}(X^{j}+Y_{m}^{j})\right) \right\vert \right] \xrightarrow[m]{}0.
\end{equation*}%
Up to taking a subsequence we can assume the convergence is also almost
sure. Since all the terms in $\sum_{j=1}^{N}\left( u_{j}(+\infty
)-u_{j}(X^{j}+Y_{m}^{j})\right) $ are non negative, we also see that $%
u_{j}(X^{j}+Y_{m}^{j})\rightarrow _{m}u_{j}(+\infty )$ almost surely for
every $j=1,\dots ,N$. By strict monotonicity of the utilities, this would
imply that, for each $j,$ $Y_{m}^{j}\rightarrow _{m}+\infty $. This clearly
contradicts the constraint $\mathbf{Y}_{m}\in \mathcal{B}_{0}$. $\,$

\textbf{STEP 2}

We will prove equations \eqref{eqminimaxapplied2} and %
\eqref{eqminimaxapplied3}, with a supremum over $\overline{\mathcal{B}_A}$
in place of a maximum, since we will show in later steps (STEP 4) that this
supremum is in fact a maximum.

We observe that since $\mathcal{B}_{0}\cap M^{\Phi }\subseteq \overline{%
\mathcal{B}_{0}}$ 
\begin{equation*}
\sup_{\mathcal{B}_{0}\cap M^{\Phi }}\sum_{j=1}^{N}\mathbb{E}\left[
u_{j}\left( X^{j}+Y^{j}\right) \right] \leq \sup_{\overline{\mathcal{B}_{0}}%
}\sum_{j=1}^{N}\mathbb{E}\left[ u_{j}\left( X^{j}+Y^{j}\right) \right].
\end{equation*}%
Moreover, by the Fenchel inequality%
\begin{equation*}
\sup_{\overline{\mathcal{B}_{0}}}\sum_{j=1}^{N}\mathbb{E}\left[ u_{j}\left(
X^{j}+Y^{j}\right) \right] \leq \inf_{\lambda \in {\mathbb{R}}_{+},\,\mathbf{%
Q}\in \mathcal{Q}}\left( \lambda \sum_{j=1}^{N}E_{{Q}^{j}}\left[ X^{j}\right]
+\sum_{j=1}^{N}\mathbb{E}\left[ v_{j}\left( \lambda \frac{\mathrm{d}Q^{j}}{%
\mathrm{d}P}\right) \right] \right) .
\end{equation*}%
Equations \eqref{eqminimaxapplied2} and \eqref{eqminimaxapplied3} follow
from Theorem \ref{thmminimax} replacing there the convex cone $\mathcal{C}$
with $\mathcal{B}_{0}\cap M^{\Phi }$ and using equation %
\eqref{polarityisnice}, which shows that $(\mathcal{C}_{1}^{0})^{+}=\mathcal{%
Q}$.

\textbf{STEP 3}

We prove that if $\widehat{\lambda }$ and $\widehat{\mathbf{Q}}$ are optima
in equation \eqref{eqminimaxapplied3}, then $\widehat{Y}^{j}:=-X^{j}-v_{j}^{%
\prime }\left( \widehat{\lambda }\frac{\mathrm{d}\widehat{Q}^{j}}{\mathrm{d}P%
}\right) $ defines an element in $\overline{\mathcal{B}_{0}}$. Observe that $%
\widehat{\lambda }$ minimizes the function 
\begin{equation*}
{\mathbb{R}}_{++}\ni \gamma \mapsto \psi (\gamma ):=\sum_{j=1}^{N}\left(
\gamma E_{\widehat{Q}^{j}}\left[ X^{j}\right] +\mathbb{E}\left[ v_{j}\left(
\gamma \frac{\mathrm{d}\widehat{Q}^{j}}{\mathrm{d}P}\right) \right] \right)
\end{equation*}%
which is real valued and convex. Also we have by Monotone Convergence
Theorem and Lemma \ref{lemmaintegra}.1. that the right and left
derivatives, which exist by convexity, satisfy 
\begin{equation*}
\frac{\mathrm{d}^{\pm }\psi }{\mathrm{d}\gamma }(\gamma )=\sum_{j=1}^{N}%
\mathbb{E}\left[ X^{j}\frac{\mathrm{d}\widehat{Q}^{j}}{\mathrm{d}P}\right]
+\sum_{j=1}^{N}\mathbb{E}\left[ v_{j}^{\prime }\left( \gamma \frac{\mathrm{d}%
\widehat{Q}^{j}}{\mathrm{d}P}\right) \frac{\mathrm{d}\widehat{Q}^{j}}{%
\mathrm{d}P}\right],
\end{equation*}%
hence the function is differentiable. Since $\widehat{\lambda }$ is a
minimum for $\psi $, this implies $\psi ^{\prime }(\widehat{\lambda })=0$,
which can be rephrased as%
\begin{equation}
\sum_{j=1}^{N}\left( \mathbb{E}\left[ X^{j}\frac{\mathrm{d}\widehat{Q}^{j}}{%
\mathrm{d}P}\right] +\mathbb{E}\left[ v_{j}^{\prime }\left( \widehat{\lambda 
}\frac{\mathrm{d}\widehat{Q}^{j}}{\mathrm{d}P}\right) \frac{\mathrm{d}%
\widehat{Q}^{j}}{\mathrm{d}P}\right] \right) =0,  \label{Foclambda}
\end{equation}%
i.e.,%
\begin{equation}
\sum_{j=1}^{N}E_{\widehat{Q}^{j}}\left[ \widehat{Y}^{j}\right] =0.
\label{QY0}
\end{equation}%
Now consider minimizing 
\begin{equation*}
\mathbf{Q}\mapsto \sum_{j=1}^{N}\left( \widehat{\lambda }E_{{Q}^{j}}\left[
X^{j}\right] +\mathbb{E}\left[ v_{j}\left( \widehat{\lambda }\frac{\mathrm{d}%
Q^{j}}{\mathrm{d}P}\right) \right] \right)
\end{equation*}%
for fixed $\widehat{\lambda }$ and $\mathbf{Q}$ varying in $\mathcal{Q}_{v}$%
. Let again $\widehat{\mathbf{Q}}$, with $\widehat{\mathbf{\eta}}:=\frac{%
\mathrm{d}\widehat{\mathbf{Q}}}{\mathrm{d}P}$, be an optimum and consider
another $\mathbf{Q}\in \mathcal{Q}_{v}$, with ${\mathbf{\eta}}:=\frac{%
\mathrm{d}\mathbf{Q}}{\mathrm{d}P}$. By Assumption \ref{A00}, the expression 
$\sum_{j=1}^{N}\mathbb{E}\left[ v_{j}\left( \lambda \frac{\mathrm{d}Q^{j}}{%
\mathrm{d}P}\right) \right] $ is finite for all choices of $\lambda $.
Observe that by convexity and differentiability of $v_{j}$ we have 
\begin{equation*}
\widehat{\lambda }\eta ^{j}v_{j}^{\prime }\left( \widehat{\lambda }\widehat{%
\eta }^{j}\right) \leq \widehat{\lambda }\widehat{\eta }^{j}v_{j}^{\prime
}\left( \widehat{\lambda }\widehat{\eta }^{j}\right) +v_{j}\left( \widehat{%
\lambda }\eta ^{j}\right) -v_{j}\left( \widehat{\lambda }\widehat{\eta }%
^{j}\right)\,.
\end{equation*}%
Hence by Lemma \ref{lemmaintegra}.1. and $\widehat{\mathbf{Q}},\mathbf{Q}%
\in \mathcal{Q}_{v}$ we conclude that 
\begin{equation}
\left( \eta ^{j}v_{j}^{\prime }\left( \widehat{\lambda }\widehat{\eta }%
^{j}\right) \right) ^{+}\in L^{1}(P).  \label{eqpospartintegrable}
\end{equation}%
To prove that also the negative part is integrable, we take a convex
combination of $\widehat{\mathbf{Q}},\mathbf{Q}\in \mathcal{Q}_{v}$, which
still belongs to $\mathcal{Q}_{v}$. By optimality of $\widehat{\eta }$ the
function 
\begin{equation*}
x\mapsto \varphi (x):=\sum_{j=1}^{N}\left( \widehat{\lambda }\mathbb{E}\left[
X^{j}\left( (1-x)\widehat{\eta }^{j}+x\eta ^{j}\right) \right] +\mathbb{E}%
\left[ v_{j}\left( \widehat{\lambda }\left( (1-x)\widehat{\eta }^{j}+x\eta
^{j}\right) \right) \right] \right) \text{, }0\leq x\leq 1\text{,}
\end{equation*}%
has a minimum at $0$, thus the right derivative of $\varphi $ at $0$ must be
non negative, so that: 
\begin{equation}
\sum_{j=1}^{N}\frac{\mathrm{d}}{\mathrm{d}x}\Big|_{0}\left( (1-x)\widehat{%
\lambda }\mathbb{E}\left[ X^{j}\widehat{\eta }^{j}\right] +x\widehat{\lambda 
}\mathbb{E}\left[ X^{j}\eta ^{j}\right] \right) \geq -\sum_{j=1}^{N}\frac{%
\mathrm{d}}{\mathrm{d}x}\Big|_{0}\mathbb{E}\left[ v_{j}\left( (1-x)\widehat{%
\lambda }\widehat{\eta }^{j}+x\widehat{\lambda }\eta ^{j}\right) \right] .
\label{eqfocQbad}
\end{equation}%
Define $H_{j}(x):=v_{j}\left( (1-x)\widehat{\lambda }\widehat{\eta }^{j}+x%
\widehat{\lambda }\eta ^{j}\right) $ and observe that as $x\downarrow 0$ by
convexity 
\begin{equation}
0\leq \left( -\frac{1}{x}\left( H_{j}(x)-H_{j}(0)\right)
+H_{j}(1)-H_{j}(0)\right) \uparrow \left( -\widehat{\lambda }v_{j}^{\prime
}\left( \widehat{\lambda }\widehat{\eta }^{j}\right) \eta ^{j}+\widehat{%
\lambda }v_{j}^{\prime }\left( \widehat{\lambda }\widehat{\eta }^{j}\right) 
\widehat{\eta }^{j}+H_{j}(1)-H_{j}(0)\right).  \label{eqincrementalpositive}
\end{equation}%
Write now explicitly equation \eqref{eqfocQbad} in terms of incremental
ratios and add and subtract the real number $\mathbb{E}\left[
\sum_{j=1}^{N}\left( H_{j}(1)-H_{j}(0)\right) \right] $ to get 
\begin{eqnarray}
&&\lim_{x\downarrow 0} \sum_{j=1}^{N}\left( \frac{1}{x}\left[ \left( (1-x)%
\widehat{\lambda }\mathbb{E}\left[ X^{j}\widehat{\eta }^{j}\right] +x%
\widehat{\lambda }\mathbb{E}\left[ X^{j}\eta ^{j}\right] \right) -\widehat{%
\lambda }\mathbb{E}\left[ X^{j}\widehat{\eta }^{j}\right] \right] +\mathbb{E}%
\left[ H_{j}(1)-H_{j}(0)\right] \right)  \label{line1} \\
&\geq &\lim_{x\downarrow 0}\sum_{j=1}^{N}\left( \mathbb{E}\left[ -\frac{1}{x}%
\left( H_{j}(x)-H_{j}(0)\right) +H_{j}(1)-H_{j}(0)\right] \right).
\end{eqnarray}%
The first limit is trivial. Observe that by \eqref{eqincrementalpositive}
and Monotone Convergence Theorem we also may compute the second limit and
then deduce:%
\begin{eqnarray*}
&&\sum_{j=1}^{N}\left( \widehat{\lambda }\mathbb{E}\left[ X^{j}\left( \eta
^{j}-\widehat{\eta }^{j}\right) \right] +\mathbb{E}\left[ H_{j}(1)-H_{j}(0)%
\right] \right) \\
&\geq &\sum_{j=1}^{N}\mathbb{E}\left[ -\widehat{\lambda }v_{j}^{\prime
}\left( \widehat{\lambda }\widehat{\eta }^{j}\right) \eta ^{j}+\widehat{%
\lambda }v_{j}^{\prime }\left( \widehat{\lambda }\widehat{\eta }^{j}\right) 
\widehat{\eta }^{j}+H_{j}(1)-H_{j}(0)\right]
\end{eqnarray*}%
and therefore%
\begin{eqnarray*}
+\infty &>&\sum_{j=1}^{N}\widehat{\lambda }\mathbb{E}\left[ X^{j}\left( \eta
^{j}-\widehat{\eta }^{j}\right) \right] \geq \mathbb{E}\left[
\sum_{j=1}^{N}\left( -\widehat{\lambda }v_{j}^{\prime }\left( \widehat{%
\lambda }\widehat{\eta }^{j}\right) \eta ^{j}+\widehat{\lambda }%
v_{j}^{\prime }\left( \widehat{\lambda }\widehat{\eta }^{j}\right) \widehat{%
\eta }^{j}\right) \right] \\
&=&\mathbb{E}\left[ \sum_{j=1}^{N}\left( \widehat{\lambda }\left(
v_{j}^{\prime }\left( \widehat{\lambda }\widehat{\eta }^{j}\right) \eta
^{j}\right) ^{-}-\widehat{\lambda }\left( v_{j}^{\prime }\left( \widehat{%
\lambda }\widehat{\eta }^{j}\right) \eta ^{j}\right) ^{+}+\widehat{\lambda }%
v_{j}^{\prime }\left( \widehat{\lambda }\widehat{\eta }^{j}\right) \widehat{%
\eta }^{j}\right) \right]\,.
\end{eqnarray*}%
Since $\sum_{j=1}^{N}v_{j}^{\prime }\left( \widehat{\lambda }\widehat{\eta }%
^{j}\right) \widehat{\eta }^{j}\in L^{1}(P) $ by Lemma \ref{lemmaintegra}.1, and $\sum_{j=1}^{N}\left( v_{j}^{\prime }\left( \widehat{\lambda }%
\widehat{\eta }^{j}\right) \eta ^{j}\right) ^{+}\in L^{1}(P)$ by equation %
\eqref{eqpospartintegrable}, we deduce that $\sum_{j=1}^{N}\left(
v_{j}^{\prime }\left( \widehat{\lambda }\widehat{\eta }^{j}\right) \eta
^{j}\right) ^{-}\in L^{1}(P)$ so that 
\begin{equation*}
0\leq \left( v_{j}^{\prime }\left( \widehat{\lambda }\widehat{\eta }%
^{j}\right) \eta ^{j}\right) ^{-}\leq \sum_{j=1}^{N}\left( v_{j}^{\prime
}\left( \widehat{\lambda }\widehat{\eta }^{j}\right) \eta ^{j}\right)
^{-}\in L^{1}(P).
\end{equation*}%
We conclude that $v_{j}^{\prime }\left( \widehat{\lambda }\widehat{\eta }%
^{j}\right) \eta ^{j}$ defines a vector in $L^{1}(P)\times\dots\times L^1(P)$%
, hence 
\begin{equation}
\widehat{\mathbf{Y}}\in L^{1}(Q^{1})\times\dots\times L^{1}(Q^{N}) \quad
\forall \mathbf{Q}\in \mathcal{Q}_{v}.  \label{YL1}
\end{equation}%
Moreover equation \eqref{eqfocQbad} can be rewritten as:%
\begin{equation}
0\leq \sum_{j=1}^{N}\widehat{\lambda }\mathbb{E}\left[ X^{j}\left( \eta ^{j}-%
\widehat{\eta }^{j}\right) \right] +\sum_{j=1}^{N}\widehat{\lambda }\mathbb{E%
}\left[ v_{j}^{\prime }\left( \widehat{\lambda }\widehat{\eta }^{j}\right)
\left( \eta ^{j}-\widehat{\eta }^{j}\right) \right] .  \label{Focq}
\end{equation}%
Now rearrange the terms in \eqref{Focq} 
\begin{equation*}
0\leq -\sum_{j=1}^{N}\widehat{\lambda }\left( \mathbb{E}\left[ X^{j}\widehat{%
\eta }^{j}\right] +\mathbb{E}\left[ v_{j}^{\prime }\left( \widehat{\lambda }%
\widehat{\eta }^{j}\right) \widehat{\eta }^{j}\right] \right) +\sum_{j=1}^{N}%
\widehat{\lambda }\left( \mathbb{E}\left[ X^{j}\eta ^{j}\right] +\mathbb{E}%
\left[ v_{j}^{\prime }\left( \widehat{\lambda }\widehat{\eta }^{j}\right)
\eta ^{j}\right] \right)
\end{equation*}%
and use \eqref{Foclambda}: 
\begin{equation*}
0\leq 0-\sum_{j=1}^{N}\widehat{\lambda }\left( \mathbb{E}\left[ \left(
-X^{j}-v_{j}^{\prime }\left( \widehat{\lambda }\widehat{\eta }^{j}\right)
\right) \eta ^{j}\right] \right) =-\widehat{\lambda }\sum_{j=1}^{N}\mathbb{E}%
\left[ \widehat{Y}^{j}\frac{\mathrm{d}Q^{j}}{\mathrm{d}P}\right].
\end{equation*}%
This proves that 
\begin{equation}
\sum_{j=1}^{N}E_{Q^{j}}\left[ \widehat{Y}^{j}\right] \leq 0\quad \forall 
\mathbf{Q}\in \mathcal{Q}_{v}  \label{QY}
\end{equation}%
and then $\widehat{\mathbf{Y}}\in \overline{\mathcal{B}_{0}}$.

\textbf{STEP 4 (Optimality of }$\widehat{\mathbf{Y}}$\textbf{)}

Under our standing Assumption \ref{A00} it is well known that $u(-v^{\prime
}(y))=v(y)-yv^{\prime }(y),\,\forall y\geq 0$. As a consequence we get by
direct substitution 
\begin{equation*}
u_{j}(X^{j}+\widehat{Y}^{j})=u_{j}\left( -v_{j}^{\prime }\left( \widehat{%
\lambda }\frac{\mathrm{d}\widehat{Q}^{j}}{\mathrm{d}P}\right) \right) =-%
\widehat{\lambda }\frac{\mathrm{d}\widehat{Q}^{j}}{\mathrm{d}P}v_{j}^{\prime
}\left( \widehat{\lambda }\frac{\mathrm{d}\widehat{Q}^{j}}{\mathrm{d}P}%
\right) +v_{j}\left( \widehat{\lambda }\frac{\mathrm{d}\widehat{Q}^{j}}{%
\mathrm{d}P}\right)
\end{equation*}

and 
\begin{equation*}
\sum_{j=1}^{N}\mathbb{E}\left[ u_{j}\left( X^{j}+\widehat{Y}^{j}\right) %
\right] =\widehat{\lambda }\left( -\sum_{j=1}^{N}\mathbb{E}\left[ \frac{%
\mathrm{d}\widehat{Q}^{j}}{\mathrm{d}P}v_{j}^{\prime }\left( \widehat{%
\lambda }\frac{\mathrm{d}\widehat{Q}^{j}}{\mathrm{d}P}\right) \right]
\right) +\sum_{j=1}^{N}\mathbb{E}\left[ v_{j}\left( \widehat{\lambda }\frac{%
\mathrm{d}\widehat{Q}^{j}}{\mathrm{d}P}\right) \right] .
\end{equation*}%
Use now the expression in \eqref{Foclambda} to substitute in the first RHS
term: 
\begin{equation*}
\sum_{j=1}^{N}\mathbb{E}\left[ u_{j}\left( X^{j}+\widehat{Y}^{j}\right) %
\right] =\widehat{\lambda }\sum_{j=1}^{N}E_{\widehat{Q}^{j}}[X^{j}]+%
\sum_{j=1}^{N}\mathbb{E}\left[ v_{j}\left( \widehat{\lambda }\frac{\mathrm{d}%
\widehat{Q}^{j}}{\mathrm{d}P}\right) \right]\,.
\end{equation*}%
The optimality of $\widehat{\mathbf{Y}}$ follows then by optimality of $(%
\widehat{\lambda },\widehat{\mathbf{Q}})$ in \eqref{eqminimaxapplied3}.

Using now our findings in STEP 2 together with optimality of $\widehat{%
\mathbf{Y}}$, the proof of equation \eqref{eqminimaxapplied2} is now
complete.

\textbf{STEP 5 (}$\widehat{\mathbf{Y}}\in \mathcal{B}_{0}$\textbf{)}

By Lemma \ref{remqvok} there exists a $\mathbf{Q}\in \mathcal{Q}_{v}^{e}:=\{%
\mathbf{Q}\in \mathcal{Q}_{v}$ s.t. $\mathbf{Q}\sim P\}$ and from (\ref{YL1}%
) we know that $v_{j}^{\prime }\left( \lambda \frac{\mathrm{d}\widehat{Q}^{j}%
}{\mathrm{d}P}\right) \in L^{1}(Q^{j})$, $\lambda >0$. Also, for every $%
j=1,\dots ,N,$ $v_{j}^{\prime }(0+)=-\infty ,$ so that $Q^{j}\left( \frac{%
\mathrm{d}\widehat{Q}^{j}}{\mathrm{d}P}=0\right) =0$. As $\mathbf{Q}\sim P$,
this in turn implies $P\left( \frac{\mathrm{d}\widehat{Q}^{j}}{\mathrm{d}P}%
=0\right) =0$, for every $j=1,\dots ,N$. Hence $\widehat{\mathbf{Q}}\sim P$.
Theorem \ref{thmweirdclosure} now can be applied to $K:=(\mathcal{B}_{0}\cap
M^{\Phi })$ and $\mathcal{Q}_{v}^{e}$ to get 
\begin{equation}
\bigcap_{\mathbf{Q}\in \mathcal{Q}_{v}^{e}}cl_{Q}\left( (\mathcal{B}_{0}\cap
M^{\Phi })-L_{+}^{1}\left( \mathbf{Q}\right) \right) =\left\{ \mathbf{Z}\in
\bigcap_{\mathbf{Q}\in \mathcal{Q}_{v}^{e}}L^{1}\left( \mathbf{Q}\right) 
\text{ s.t. }\sum_{j=1}^{N}E_{{Q}^{j}}\left[ Z^{j}\right] \leq 0\,\,\forall
\,\mathbf{Q}\in \mathcal{Q}_{v}^{e}\right\}.  \label{eqweirdclosure}
\end{equation}%
As $\widehat{\mathbf{Y}}\in \overline{\mathcal{B}_{0}}$ and $\overline{%
\mathcal{B}_{0}}$ is included in the RHS of (\ref{eqweirdclosure}), we
deduce that $\widehat{\mathbf{Y}}$ belongs to the LHS of (\ref%
{eqweirdclosure}). Now by equation (\ref{QY0}) we see that $\widehat{\mathbf{%
Y}}$ satisfies $\sum_{j=1}^{N}\mathbb{E}\left[ \widehat{Y}^{j}\frac{\mathrm{d%
}\widehat{Q}^{j}}{\mathrm{d}P}\right] =0$, and this implies that:%
\begin{equation}
\widehat{\mathbf{Y}}\in cl_{\widehat{\mathbf{Q}}}\left( \mathcal{B}_{0}\cap
M^{\Phi }\right) ,  \label{Yclosure}
\end{equation}
the $L^{1}(\widehat{Q}^{1})\times\dots\times L^{1}(\widehat{Q}^{1}) $-(norm)
closure of $\mathcal{B}_{0}\cap M^{\Phi }$. In particular $\widehat{\mathbf{Y%
}}$ is a $\widehat{\mathbf{Q}}$ (hence $P$)- a.s. limit of elements in $%
\mathcal{B}_{0}$ which is closed in probability $P$, so that $\widehat{%
\mathbf{Y}}\in \mathcal{B}_{0}$.

\textbf{STEP 6}

The conditions in (\ref{EQ}) are proved in (\ref{QY0}) and (\ref{QY}). We
conclude with uniqueness. By the strict convexity of the utilities and the
convexity of $\overline{\mathcal{B}_{0}}$, it is evident that the
maximization problem given by $\sup_{\overline{\mathcal{B}_{0}}%
}\sum_{j=1}^{N}\mathbb{E}\left[ u_{j}\left( X^{j}+Y^{j}\right) \right] $
admits at most one optimum. Now clearly if $(\widehat{\lambda },\widehat{%
\mathbf{Q}})$ and $(\widetilde{\lambda },\widetilde{\mathbf{Q}})$ are optima
for the minimax expression \eqref{eqminimaxapplied3}, they both give rise to
two optima $\widehat{\mathbf{Y}},\,\widetilde{\mathbf{Y}}$ as in the previous
steps. Uniqueness of the solution for the primal problem implies $\widehat{%
\mathbf{Y}}=\widetilde{\mathbf{Y}}$. Under Assumption \ref{A00}.(a) the
functions $v_{1}^{\prime },\dots ,v_{N}^{\prime }$ are injective and
therefore we conclude that $\widehat{\lambda }\frac{\mathrm{d}\widehat{%
\mathbf{Q}}}{\mathrm{d}P}=\widetilde{\lambda }\frac{\mathrm{d}\widetilde{%
\mathbf{Q}}}{\mathrm{d}P}$. Taking expectations we get $\widehat{\lambda }=%
\widetilde{\lambda }$ and then $(\widehat{\lambda },\widehat{\mathbf{Q}})=(%
\widetilde{\lambda },\widetilde{\mathbf{Q}})$.

\textbf{Conclusion}

The more general case $A\neq 0$ can be obtained using Remark \ref%
{remarkfrom0toA}. We just sketch one step of the proof, as the other steps
follows similarly. Using $\mathbf{a_{0}}$ as in Remark \ref{remarkfrom0toA},
in STEP 3 we see that 
\begin{equation*}
0\leq -\widehat{\lambda }\sum_{j=1}^{N}\mathbb{E}\left[ \widehat{Y}^{j}\frac{%
\mathrm{d}Q^{j}}{\mathrm{d}P}\right] +\widehat{\lambda }%
\sum_{j=1}^{N}a_{0}^{j}
\end{equation*}%
which yields that $\widehat{\mathbf{Y}}\in \overline{\mathcal{B}_{A}}$.
\end{proof}

\begin{remark}
\label{fairsums}Notice that $\mathbf{Y}\in \mathcal{B}\cap M^{\Phi }$
implies that $\mathbf{Z}\in \mathcal{B}_{0}$, where $\mathbf{Z}$ is defined
by $Z^{j}:=Y^{j}-x^{j}\sum_{k=1}^{N}Y^{k}$ for any $\mathbf{x}\in {\mathbb{R}%
}^{N}$ such that $\sum_{j=1}^{N}x^{j}=1$. To see this, recall that we are
assuming that ${\mathbb{R}}^{N}+\mathcal{B}=\mathcal{B}$. As $%
\sum_{j=1}^{N}Y^{j}\in {\mathbb{R}}$, then $\mathbf{Z}\in \mathcal{B}$ and,
since also trivially integrability is preserved and $\sum_{j=1}^{N}Z^{j}=0$,
we conclude that $\mathbf{Z}\in \mathcal{B}_{0}$.
\end{remark}

\begin{proposition}
\label{propfairprob}For all $\mathbf{Y}\in\mathcal{B}\cap M^{\Phi }$
and $\,\mathbf{Q}\in\mathcal{Q}$ 
\begin{equation}
\sum_{j=1}^{N}E_{{Q}^{j}}\left[ Y^{j}\right] \leq \sum_{j=1}^{N}Y^{j}.
\label{fair}
\end{equation}%
Moreover, denoting by $cl_{\mathbf{Q}}\left( \mathcal{B}\cap M^{\Phi
}\right) $ the $L^{1}(Q^{1})\times\dots\times L^{1}(Q^{N})$-norm closure of $%
\mathcal{B}\cap M^{\Phi }$, inequality (\ref{fair}) holds for all $\,\mathbf{%
Y}\in cl_{\mathbf{Q}}\left( \mathcal{B}\cap M^{\Phi }\right) $ and
$\mathbf{Q}\in \mathcal{Q},$ $\mathbf{Q}\sim P.$ In particular, (\ref{fair})
holds for $\widehat{\mathbf{Q}}\sim P$ and $\widehat{\mathbf{Y}}\in cl_{%
\widehat{\mathbf{Q}}}\left( \mathcal{B}_{0}\cap M^{\Phi }\right) $ defined
in Theorem \ref{thmoptimumexists}.
\end{proposition}

\begin{proof}
Take $\mathbf{Y}\in \mathcal{B}\cap M^{\Phi }$ and argue as in Remark \ref%
{fairsums}, with the notation introduced there. By the definition of the
polar, $\sum_{j=1}^{N}\mathbb{E}\left[ Z^{j}\varphi ^{j}\right] \leq 0$ for
all $\mathbf{\varphi }\in (\mathcal{B}\cap M^{\Phi })^{0}$, and in
particular for all $\mathbf{Q}\in \mathcal{Q}$ 
\begin{equation*}
0\geq \sum_{j=1}^{N}\mathbb{E}\left[ Z^{j}\frac{\mathrm{d}Q^{j}}{\mathrm{d}P}%
\right] =\sum_{j=1}^{N}\mathbb{E}\left[ Y^{j}\frac{\mathrm{d}Q^{j}}{\mathrm{d%
}P}\right] -\sum_{j=1}^{N}\mathbb{E}\left[ x^{j}\left(
\sum_{k=1}^{N}Y^{k}\right) \frac{\mathrm{d}Q^{j}}{\mathrm{d}P}\right]
=\sum_{j=1}^{N}E_{{Q}^{j}}\left[ Y^{j}\right] -\sum_{j=1}^{N}Y^{j}\,.
\end{equation*}%
As to the second claim, take a sequence $(\mathbf{k}_{n})_{n}$ in $\mathcal{B%
}\cap M^{\Phi }$ converging both $\mathbf{Q}$-almost surely (hence $P$-a.s.) and in norm to $\mathbf{Y}$ and apply (\ref{fair}) to $\mathbf{k}_{n}$
to get 
\begin{equation}  \label{eqfairlimit}
\sum_{j=1}^{N}E_{{Q}^{j}}\left[ Y^{j}\right] =\lim_{n}\sum_{j=1}^{N}E_{{Q}%
^{j}}\left[ k_{n}^{j}\right] \overset{P\text{-a.s.}}{\leq }\liminf_{n}\left(
\sum_{j=1}^{N}k_{n}^{j}\right) \overset{P\text{-a.s.}}{=}\sum_{j=1}^{N}Y^{j}.
\end{equation}
\end{proof}

\begin{remark}
\label{remYA}In particular (\ref{fair}) shows that $\forall \,\mathbf{Q}\in 
\mathcal{Q}$ 
\begin{equation*}
\left\{ \mathbf{Y}\in \mathcal{B}\cap M^{\Phi }\mid \sum_{j=1}^{N}Y^{j}\leq
A\right\} \subseteq \left\{ \mathbf{Y}\in M^{\Phi }\mid \sum_{j=1}^{N}E_{{Q}%
^{j}}\left[ Y^{j}\right] \leq A\right\}
\end{equation*}%
and therefore $\pi (A)\leq \pi ^{\mathbf{Q}}(A).$
\end{remark}

\subsection{Utility Maximization with a fixed probability measure}

The following represents a counterpart to Theorem \ref{thmoptimumexists},
once a measure is fixed a priori.

\begin{proposition}
\label{proppiq} Fix $\mathbf{Q}\in \mathcal{Q}_{v}$. If $\pi ^{\mathbf{Q}%
}(A)<+\infty $, then%
\begin{eqnarray}
\pi ^{\mathbf{Q}}(A) &=&\Pi ^{\mathbf{Q}}(A)=\sup \left\{ \sum_{j=1}^{N}%
\mathbb{E}\left[ u_{j}\left( X^{j}+Y^{j}\right) \right] \,\middle|\,\mathbf{Y%
}\in L^{1}(\mathbf{Q}),\,\sum_{j=1}^{N}E_{{Q}^{j}}\left[ Y^{j}\right] \leq
A\right\}  \label{eqminimaxfixq1} \\
&=&\min_{\lambda \in {\mathbb{R}}_{+}}\left( \lambda \left( \sum_{j=1}^{N}E_{%
{Q}^{j}}\left[ X^{j}\right] +A\right) +\sum_{j=1}^{N}\mathbb{E}\left[
v_{j}\left( \lambda \frac{\mathrm{d}Q^{j}}{\mathrm{d}P}\right) \right]
\right) \,.  \notag
\end{eqnarray}%
If additionally any of the two expressions is strictly less than $%
\sum_{j=1}^{N}u_{j}(+\infty )$, then 
\begin{equation}
\pi ^{\mathbf{Q}}(A)=\min_{\lambda \in {\mathbb{R}}_{++}}\left( \lambda
\left( \sum_{j=1}^{N}E_{{Q}^{j}}\left[ X^{j}\right] +A\right) +\sum_{j=1}^{N}%
\mathbb{E}\left[ v_{j}\left( \lambda \frac{\mathrm{d}Q^{j}}{\mathrm{d}P}%
\right) \right] \right) .  \label{lasteq}
\end{equation}
\end{proposition}

\begin{proof}
Again, we prove the case $A=0$ since Remark \ref{remarkfrom0toA} can be used
to obtain the general case $A\neq 0$. From $M^{\Phi }\subseteq \mathcal{L}%
\subseteq L^{1}(\mathbf{Q})$ we obtain: 
\begin{align}
\pi ^{\mathbf{Q}}(0) &:=\sup \left\{ \sum_{j=1}^{N}\mathbb{E}\left[
u_{j}\left( X^{j}+Y^{j}\right) \right] \,\middle|\,\mathbf{Y}\in M^{\Phi
},\,\sum_{j=1}^{N}E_{{Q}^{j}}\left[ Y^{j}\right] \leq 0\right\} \leq \Pi ^{%
\mathbf{Q}}(0)  \notag \\
&\leq \sup \left\{ \sum_{j=1}^{N}\mathbb{E}\left[ u_{j}\left(
X^{j}+Y^{j}\right) \right] \,\middle|\,\mathbf{Y}\in L^{1}(\mathbf{Q}%
),\,\sum_{j=1}^{N}E_{{Q}^{j}}\left[ Y^{j}\right] \leq 0\right\}  \notag \\
&\leq \min_{\lambda \in {\mathbb{R}}_{+}}\left( \lambda \sum_{j=1}^{N}E_{{Q}%
^{j}}\left[ X^{j}\right] +\sum_{j=1}^{N}\mathbb{E}\left[ v_{j}\left( \lambda 
\frac{\mathrm{d}Q^{j}}{\mathrm{d}P}\right) \right] \right)  \label{supmin}
\end{align}%
by the Fenchel inequality. Define the convex cone 
\begin{equation*}
\mathcal{C}:=\left\{ \mathbf{Y}\in M^{\Phi }\mid \sum_{j=1}^{N}E_{{Q}^{j}}%
\left[ Y^{j}\right] \leq 0\right\} .
\end{equation*}%
The hypotheses on $\mathcal{C}$ of Theorem \ref{thmminimax} hold true and
inequality (\ref{supmin}) shows that $\pi ^{\mathbf{Q}}(0)<+\infty $ for all 
$\mathbf{X}\in M^{\Phi }$. The finite dimensional cone $\left\{ \lambda %
\left[ \frac{\mathrm{d}Q^{1}}{\mathrm{d}P},\dots ,\frac{\mathrm{d}Q^{N}}{%
\mathrm{d}P}\right] ,\,\,\lambda \geq 0\right\} \subseteq L^{\Phi ^{\ast }}$
is closed, and then by the Bipolar Theorem $\mathcal{C}^{0}=\left\{ \lambda %
\left[ \frac{\mathrm{d}Q^{1}}{\mathrm{d}P},\dots ,\frac{\mathrm{d}Q^{N}}{%
\mathrm{d}P}\right] ,\,\,\lambda \geq 0\right\} $. Hence the set $(\mathcal{C%
}_{1}^{0})^{+}$ in the statement of Theorem \ref{thmminimax} is exactly $%
\left\{ \left[ \frac{\mathrm{d}Q^{1}}{\mathrm{d}P},\dots ,\frac{\mathrm{d}%
Q^{N}}{\mathrm{d}P}\right] \right\} $ and Theorem \ref{thmminimax} proves
that $\pi ^{\mathbf{Q}}(0)$ is equal to the RHS of (\ref{supmin}). We can
similarly argue to prove \eqref{lasteq}.
\end{proof}

To conclude, we provide the minimax duality between the maximization
problems with and without a fixed measure

\begin{corollary}
\label{propminimum} The following holds: 
\begin{equation*}
\pi (A)=\min_{\mathbf{Q}\in \mathcal{Q}_{v}}\pi ^{\mathbf{Q}}(A)=\pi ^{%
\mathbf{\widehat{Q}}}(A)<+\infty\,,
\end{equation*}%
where $\mathbf{\widehat{Q}}$ is the minimax measure from Theorem \ref%
{thmoptimumexists}.
\end{corollary}

\begin{proof}
It is an immediate consequence of Theorem \ref{thmoptimumexists} and
Proposition \ref{proppiq}.
\end{proof}

\begin{lemma}
\label{lemmalinkpiS}For all $\mathbf{Q}\in \mathscr{\probq}$ we have $\Pi ^{%
\mathbf{Q}}(A)=S^{\mathbf{Q}}(A)$ and, if $\mathbf{\widehat{Q}}$ is the
minimax measure from Theorem \ref{thmoptimumexists}, then 
\begin{equation}
\pi (A)=\pi ^{\widehat{\mathbf{Q}}}(A)=\Pi ^{\widehat{\mathbf{Q}}}(A)=S^{%
\widehat{\mathbf{Q}}}(A).  \label{piequalPI}
\end{equation}
\end{lemma}

\begin{proof}
Let $\mathbf{Y}\in \mathcal{L}$, $\mathbf{Q}\in \mathscr{Q}$, $%
a^{n}:=E_{Q^{n}}[Y^{n}]$ and $Z^{n}:=Y^{n}-a^{n}$. As $\mathcal{L}+\mathbb{R}%
^{N}=\mathcal{L}$, $Z^{n}\in \mathcal{L}^{n}$ and 
\begin{eqnarray*}
\Pi ^{\mathbf{Q}}(\mathbf{A}) &=&\sup_{\mathbf{Y}\in \mathcal{L}}\left\{ 
\mathbb{E}\left[ \sum_{n=1}^{N}u_{n}(X^{n}+Y^{n})\right] \mid
\sum_{n=1}^{N}E_{Q^{n}}[Y^{n}]=A\right\} \\
&=&\sup_{\mathbf{a}\in \mathbb{R}^{N}\text{, }\mathbf{Z}\in \mathcal{L}%
}\left\{ \mathbb{E}\left[ \sum_{n=1}^{N}u_{n}(X^{n}+Z^{n}+a^{n})\right] \mid
E_{Q^{n}}[Z^{n}]=0,\sum_{n=1}^{N}a^{n}=A\right\} \\
&=&\sup_{\mathbf{a}\in \mathbb{R}^{N}\text{, }\sum_{n=1}^{N}a^{n}=A}\left\{
\sup_{\mathbf{Z}\in \mathcal{L}\text{ }:\text{ }\mathbb{E}%
_{Q^{n}}[Z^{n}]=0}\sum_{n=1}^{N}\mathbb{E}\left[ u_{n}(X^{n}+Z^{n}+a^{n})%
\right] \right\} \\
&=&\sup_{\mathbf{a}\in \mathbb{R}^{N}\text{, }\sum_{n=1}^{N}a^{n}=A}%
\sum_{n=1}^{N}\sup_{Z^{n}\in \mathcal{L}^{n}}\left\{ \mathbb{E}\left[
u_{n}(X^{n}+Z^{n}+a^{n})\right] \mid E_{Q^{n}}[Z^{n}]=0\right\} \\
&=&\sup_{\mathbf{a}\in \mathbb{R}^{N}\text{, }\sum_{n=1}^{N}a^{n}=A}%
\sum_{n=1}^{N}\sup_{Y^{n}\in \mathcal{L}^{n}}\left\{ \mathbb{E}\left[
u_{n}(X^{n}+Y^{n})\right] \mid E_{Q^{n}}[Y^{n}]=a^{n}\right\} \\
&=&\sup_{\mathbf{a}\in \mathbb{R}^{N}\text{ }\sum_{n=1}^{N}a^{n}=A}%
\sum_{n=1}^{N}U_{n}^{Q^{n}}(a^{n})=S^{\mathbf{Q}}(A)\,.
\end{eqnarray*}%
The first equality in (\ref{piequalPI}) follows from Corollary \ref%
{propminimum} and the second one from (\ref{eqminimaxfixq1}).
\end{proof}

\subsection{Main results}

\label{secmain}

\begin{theorem}
\label{thmsorteexistsA} Take $\mathscr{Q}=\mathcal{Q}_{v}$ and set $\mathcal{%
L}=\bigcap_{\mathbf{Q}\in \mathcal{Q}_{v}}L^{1}(\mathbf{Q})$. Under
Assumption \ref{A00}, for any $\mathbf{X}\in M^{\Phi }$ and any $A\in {%
\mathbb{R}}$ a SORTE exists, namely $(\widehat{\mathbf{Y}},\widehat{\mathbf{Q%
}})\in ${$\mathcal{B}_{A}\times $}$\mathcal{Q}_{v}$ defined in Theorem \ref%
{thmoptimumexists} and 
\begin{equation}
\widehat{a}^{n}:=E_{\widehat{Q}^{n}}[\widehat{Y}^{n}],\quad n=1,\dots ,N,
\label{aY}
\end{equation}%
satisfy:

\begin{enumerate}
\item {$\widehat{Y}^{n}$ is an optimum for $U_{n}^{\widehat{Q}^{n}}(\widehat{%
a}^{n})$, for each $n\in \{1,\dots ,N\}$,}

\item $\widehat{{\mathbf{a}}}${\ is an optimum for $S^{\widehat{\mathbf{Q}}%
}(A),$}

\item {$\widehat{\mathbf{Y}}\in \mathcal{B}$ and $\sum_{n=1}^{N}\widehat{Y}%
^{n}=A$ }$P$-a.s.
\end{enumerate}
\end{theorem}

\begin{proof}
$\,$

\textbf{1)}: We prove that $U_{n}^{\widehat{Q}^{n}}(\widehat{a}^{n})=\mathbb{%
E}\left[ u_{n}\left( X^{n}+\widehat{Y}^{n}\right) \right] <u_{n}(+\infty )$,
for all $n=1,\dots ,N$ , thus showing that {$\widehat{Y}^{n}$ is an optimum
for $U_{n}^{\widehat{Q}^{n}}(\widehat{a}^{n}).$} As $\widehat{Y}^{n}\in 
\mathcal{L}^{n}$ for all $n=1,\dots ,N$, then by definition of $U_{n}^{%
\widehat{Q}^{n}}(\widehat{a}^{n})$ we obtain:%
\begin{equation*}
\sup \left\{ \mathbb{E}\left[ u_{n}(X^{n}+Z)\right] \middle|Z\in \mathcal{L}%
^{n},\,E_{\widehat{Q}^{n}}[Z]\leq \widehat{a}^{n}\right\} =:U_{n}^{\widehat{Q%
}^{n}}(\widehat{a}^{n})\geq \mathbb{E}\left[ u_{n}\left( X^{n}+\widehat{Y}%
^{n}\right) \right] .
\end{equation*}%
If, for some index, the last inequality were strict we would obtain the
contradiction 
\begin{equation}
\pi ^{\widehat{\mathbf{Q}}}(A)=\pi (A)\overset{\text{Thm. }\ref%
{thmoptimumexists}}{=}\sum_{n=1}^{N}\mathbb{E}\left[ u_{n}\left( X^{n}+%
\widehat{Y}^{n}\right) \right] <\sum_{n=1}^{N}U_{n}^{\widehat{Q}^{n}}(%
\widehat{a}^{n})\leq S^{\widehat{\mathbf{Q}}}(A)=\pi ^{\widehat{\mathbf{Q}}%
}(A)\,,  \label{contradiction}
\end{equation}%
where we used \eqref{piequalPI} in the first and last equality.

In particular then $\mathbb{E}\left[ u_{n}\left( X^{n}+\widehat{Y}%
^{n}\right) \right] <u_{n}(+\infty )$, for all $n=1,\dots ,N$ $.$ Indeed, if
the latter were equal to $u_{n}(+\infty )$, then $u_{n}$ would attain its
maximum over a compact subset of ${\mathbb{R}}$, which is not the case.

\textbf{2)}: From (\ref{EQ}) we know that $A={\sum_{n=1}^{N}}E_{\widehat{Q}%
^{n}}[\widehat{Y}^{n}]=${$\sum_{n=1}^{N}\widehat{a}^{n}$}. From (\ref%
{piequalPI}) we have 
\begin{equation*}
S^{\widehat{\mathbf{Q}}}(A)=\pi (A)\overset{\text{Thm. }\ref%
{thmoptimumexists}}{=}\sum_{n=1}^{N}\mathbb{E}\left[ u_{n}\left( X^{n}+%
\widehat{Y}^{n}\right) \right] =\sum_{n=1}^{N}U_{n}^{\widehat{Q}^{n}}(%
\widehat{a}^{n})\leq S^{\widehat{\mathbf{Q}}}(A).
\end{equation*}

\textbf{3)}: We already know that $\widehat{\mathbf{Y}}\in ${$\mathcal{B}_{A}%
{:=\mathcal{B}\cap \{\mathbf{Y}\in (L^{0}(P))^{N}\mid
\sum_{n=1}^{N}Y^{n}\leq A\}}$}. From Proposition \ref{propfairprob} we
deduce 
\begin{equation*}
A={\sum_{n=1}^{N}}E_{\widehat{Q}^{n}}[\widehat{Y}^{n}]\leq {\sum_{n=1}^{N}%
\widehat{Y}^{n}\leq A.}
\end{equation*}
\end{proof}

We now turn our attention to uniqueness and Pareto optimality, but we will
need an additional property and an auxiliary result.

\begin{definition}[Def. 4.18 in \protect\cite{bffm}]
\label{DefTrunc}We say that $\mathcal{B}\subseteq (L^{0}(P))^{N}$ is \emph{%
closed under truncation} if for each $\mathbf{Y}\in \mathcal{B}$ there
exists $m_{Y}\in \mathbb{N}$ and $\mathbf{c}_{Y}=(c_{Y}^{1},...,c_{Y}^{N})%
\in \mathbb{R}^{N}$ such that $\sum_{n=1}^{N}c_{Y}^{n}=%
\sum_{n=1}^{N}Y^{n}:=c_{Y}\in \mathbb{R}$ and for all $m\geq m_{Y}$ 
\begin{equation}
\mathbf{Y}_{m}:=\mathbf{Y}I_{\{\cap _{n=1}^{N}\{|Y^{n}|<m\}\}}+\mathbf{c}%
_{Y}I_{\{\cup _{n=1}^{N}\{|Y^{n}|\geq m\}\}}\in \mathcal{B}.  \label{EqTrunc}
\end{equation}
\end{definition}

\begin{remark}
We stress the fact that all the sets introduced in Example \ref{exCh}
satisfy closedness under truncation.
\end{remark}

\begin{lemma}
\label{lemmaclosed}

Let $\mathcal{B}$ be closed under truncation. Then for every $A\in {\mathbb{R%
}}$ 
\begin{equation*}
\mathcal{B}_{A}\cap \mathcal{L}\subseteq \overline{\mathcal{B}_{A}}.
\end{equation*}
\end{lemma}

\begin{proof}
Fix any $\mathbf{Q}\in\mathcal{Q}_v$ and argue as in Proposition 4.20 in 
\cite{bffm}: let $\mathbf{Y}\in \mathcal{B}_A\cap\mathcal{L}\subseteq L^1(%
\mathbf{Q})$ and consider $\mathbf{Y}_{m}$ for $m\in \mathbb{N}$ as defined
in \eqref{EqTrunc}, where w.l.o.g. we assume $m_{Y}=1$. Note that $%
\sum_{n=1}^{N}Y_{m}^{n}=c_{Y}(=\sum_{n=1}^{N}Y^{n}\leq A)$ for all $m\in 
\mathbb{N} $. By boundedness of $\mathbf{Y}_{m}$ and (\ref{EqTrunc}), we
have $\mathbf{Y}_{m}\in \mathcal{B}\cap M^{\Phi }$ for all $m\in \mathbb{N}$%
. \ Further, $\mathbf{Y}_{m}\rightarrow \mathbf{Y}$ $P$-a.s. for $%
m\rightarrow \infty $ , and thus, since $|\mathbf{Y}_{m}|\leq \max \{|%
\mathbf{Y}|,|\mathbf{c}_{Y}|\}\in L^{1}(\mathbf{Q})$ for all $m\in \mathbb{N}
$, also $\mathbf{Y}_{m}\rightarrow \mathbf{Y}$ in $L^{1}(\mathbf{Q})$ for $%
m\rightarrow \infty $ by dominated convergence.

Now, if $\mathbf{Q}\sim P$ we can directly apply Proposition \ref%
{propfairprob} to get that $\sum_{n=1}^N E_{Q^n}[Y^n]\leq
\sum_{n=1}^NY^n\leq A$. If we only have $\mathbf{Q}\ll P$ we can see that %
\eqref{eqfairlimit} still holds, with the particular choice of $(\mathbf{Y}%
_{m})_m$ in place of $(\mathbf{k_n})_n$, because the construction of $%
\mathbf{Y}_{m}$ is made $P$-almost surely.
\end{proof}

Define 
\begin{equation}
\Pi (A):=\sup \left\{ \mathbb{E}\left[ \sum_{n=1}^{N}u_{n}(X^{n}+Y^{n})%
\right] \mid \mathbf{Y}\in {\mathcal{L}\cap \mathcal{B}},\,%
\sum_{n=1}^{N}Y^{n}\leq A\right\}.  \label{PI}
\end{equation}

\begin{lemma}
\label{lemmalinkpiSPi2}Let $\mathcal{B}$ be closed under truncation. If $%
\mathbf{\widehat{Q}}$ is the minimax measure from Theorem \ref%
{thmoptimumexists}, then 
\begin{equation}
\pi (A)=\Pi (A)=\pi ^{\widehat{\mathbf{Q}}}(A)=\Pi ^{\widehat{\mathbf{Q}}%
}(A)=S^{\widehat{\mathbf{Q}}}(A).  \label{piequalPI2}
\end{equation}
\end{lemma}

\begin{proof}
It is clear that since $\mathcal{B}_{A}\cap M^{\Phi }\subseteq \mathcal{B}%
_{A}\cap \mathcal{L}$ we have $\pi (A)\leq \Pi (A)$ just by definitions %
\eqref{defpi} and \eqref{PI}. Now observe that by Lemma \ref{lemmaclosed} we
have $\mathcal{B}_{A}\cap \mathcal{L}\subseteq \overline{\mathcal{B}_{A}}$,
so that $\Pi (A)\leq \Pi ^{\widehat{\mathbf{Q}}}(A)$. The chain of
equalities then follows by Lemma \ref{lemmalinkpiS}.
\end{proof}

\begin{theorem}
\label{thmsorteuniqueA}Let $\mathcal{B}$ be closed under truncation. Under
the same assumptions of Theorem \ref{thmsorteexistsA}, for any $\mathbf{X}%
\in M^{\Phi }$ and $A\in \mathbb{R}$ the SORTE is unique and is a Pareto
optimal allocation for both the sets 
\begin{equation}
\mathscr{V}=\left\{ \mathbf{Y}\in \mathcal{L}\cap \mathcal{B}\,\,\middle%
|\,\,\sum_{n=1}^{N}Y^{n}\leq A\,\,P\text{-a.s.}\right\} \,\,\,\,\,\text{and}%
\,\,\,\,\,\,\,\,\mathscr{V}=\left\{ \mathbf{Y}\in \mathcal{L}\,\,\middle%
|\,\,\sum_{n=1}^{N}E_{\widehat{Q}^{n}}\left[ Y^{n}\right] \leq A\right\} .
\label{Vdef}
\end{equation}
\end{theorem}

\begin{proof}
Use Proposition \ref{proppiq} and Corollary \ref{propminimum} to get that
for any $\mathbf{Q}\in \mathcal{Q}_{v}$%
\begin{equation}
\Pi ^{\mathbf{Q}}(A)=\pi ^{\mathbf{Q}}(A)\geq \pi (A).  \label{ineqchain1A1}
\end{equation}%
Let $(\widetilde{\mathbf{Y}},\widetilde{\mathbf{Q}},\widetilde{\mathbf{a}})$
be a SORTE and $(\widehat{\mathbf{Y}},\widehat{\mathbf{Q}},\widehat{\mathbf{a%
}})$ be the one from Theorem \ref{thmsorteexistsA}.

By 1) and 2) in the definition of SORTE, together with Lemma \ref%
{lemmalinkpiS}, we see that $\widetilde{\mathbf{Y}}$ is an optimum for $\Pi
^{\widetilde{\mathbf{Q}}}(A)=S^{\widetilde{\mathbf{Q}}}(A)$. Also, $%
\widetilde{\mathbf{Y}}\in \mathcal{B}_{A}\cap \overline{\mathcal{B}_{A}}$ by
Lemma \ref{lemmaclosed}. We can conclude by equation %
\eqref{eqminimaxapplied2} that 
\begin{equation*}
\pi (A)\geq \sum_{n=1}^{N}\mathbb{E}\left[ u_{n}\left( X^{n}+\widetilde{Y}%
^{n}\right) \right] =\Pi ^{\widetilde{\mathbf{Q}}}(A)\overset{\text{eq.}%
\eqref{eqminimaxfixq1}}{=}\pi ^{\widetilde{\mathbf{Q}}}(A)\overset{\text{Cor.%
}\ref{propminimum}}{\geq }\pi (A),
\end{equation*}%
which tells us that $\pi (A)=\pi ^{\widetilde{\mathbf{Q}}}(A)=\sum_{n=1}^{N}%
\mathbb{E}\left[ u_{n}\left( X^{n}+\widetilde{Y}^{n}\right) \right] .$

By Theorem \ref{thmoptimumexists}, we also have $\pi (A)=\sum_{n=1}^{N}%
\mathbb{E}\left[ u_{n}\left( X^{n}+\widehat{Y}^{n}\right) \right] $. Then $%
\widehat{\mathbf{Y}},\widetilde{\mathbf{Y}}\in \overline{\mathcal{B}_{A}}$
(Lemma \ref{lemmaclosed}) and $\Pi (A)=\pi (A)$ (Lemma \ref{lemmalinkpiSPi2}%
) imply that both $\widehat{\mathbf{Y}},\widetilde{\mathbf{Y}}$ are optima
for $\Pi (A)$. By strict concavity of the utilities $u_{1},\dots ,u_{N}$, $%
\Pi (A)$ has at most one optimum. From this, together with uniqueness of the
minimax measure (see Theorem \ref{thmoptimumexists}), we get $(\widetilde{%
\mathbf{Y}},\widetilde{\mathbf{Q}})=(\widehat{\mathbf{Y}},\widehat{\mathbf{Q}%
})$. We infer from equation \eqref{aY} and Remark \ref{remarkmonot} that
also $\widetilde{\mathbf{a}}=\widehat{\mathbf{a}}$.

To prove the Pareto optimality observe that Theorem \ref{thmoptimumexists}
proves that $\widehat{\mathbf{Y}}\in \overline{\mathcal{B}_{A}}\mathcal{%
\subseteq L}$ is the unique optimum for $\Pi (A)$ (see Lemma \ref%
{lemmalinkpiSPi2}) and so it is also the unique optimum for $\Pi ^{\widehat{%
\mathbf{Q}}}(A)$. Pareto optimality then follows from Proposition \ref{ABQ},
noticing that $\Pi(\mathscr{V})$ for the two sets in \eqref{Vdef} are $%
\Pi(A) $ and $\Pi ^{\widehat{\mathbf{Q}}}(A)$ respectively.
\end{proof}

\subsection{Dependence of the SORTE on $\mathbf{X}$ and on $\mathcal{B}$}

\label{SecXB}

We see from the proof of Theorem \ref{thmsorteexistsA} that the triple
defining the SORTE (obviously) depends on the choice of $A$. We now focus on
the study of how such triple depends on $\mathbf{X}$. To this end, we first
specialize to the case $\mathcal{B}=\mathcal{C}_{\mathbb{R}}$.

\begin{proposition}
\label{propYfsumX} Under the hypotheses of Theorem \ref{thmsorteexistsA} and
for $\mathcal{B}=\mathcal{C}_{\mathbb{R}}$, the variables $\frac{\mathrm{d}%
\widehat{\mathbf{Q}}}{\mathrm{d}P}$ and $\mathbf{X}+\widehat{\mathbf{Y}}$
are $\sigma(X^1+\dots+X^N)$ (essentially) measurable.
\end{proposition}

\begin{proof}
By Theorem \ref{thmsorteexistsA} and Theorem \ref{thmsorteuniqueA} we have
that $(\widehat{\lambda },\widehat{\mathbf{Q}})$ is an optimum of the RHS of
equation \eqref{eqminimaxapplied3}. Notice that in this specific case $%
\mathbf{Y}:=\mathbf{e_{i}}1_{A}-\mathbf{e_{j}}1_{A}\in \mathcal{B}\cap
M^{\Phi }$ for all $i,j$ and all measurable sets $A\in \mathcal{F}$. Let $%
\mathbf{Q}\in \mathcal{Q}$. Then from (\ref{fair}) $%
\sum_{n=1}^{N}(E_{Q^{n}}[Y^{n}]-Y^{n})\leq 0$ and so $%
Q^{i}(A)-1_{A}-Q^{j}(A)+1_{A}\leq 0,$ i.e., $Q^{i}(A)-Q^{j}(A)\leq 0$.
Similarly taking $\mathbf{Y}:=-\mathbf{e_{i}}1_{A}+\mathbf{e_{j}}1_{A}\in 
\mathcal{B}$, we get $Q^{j}(A)-Q^{i}(A)\leq 0$. Hence all the components of
vectors in $\mathcal{Q}$ are equal. Let $\mathcal{G}:=\sigma (X^{1}+\dots
+X^{N})$. Then for any $\lambda \in {\mathbb{R}}_{++}$ and any $\mathbf{Q}%
=[Q,\dots,Q]\in \mathcal{Q}$ we have: 
\begin{equation*}
\lambda \left( \sum_{n=1}^{N}E_{{Q}^{n}}\left[ X^{n}\right] +A\right)
+\sum_{n=1}^{N}\mathbb{E}\left[ v_{n}\left( \lambda \frac{\mathrm{d}Q^{n}}{%
\mathrm{d}P}\right) \right] =\lambda \left( \mathbb{E}\left[ \left(
\sum_{n=1}^{N}X^{n}\right) \frac{\mathrm{d}Q}{\mathrm{d}P}\right] +A\right)
+\sum_{n=1}^{N}\mathbb{E}\left[ v_{n}\left( \lambda \frac{\mathrm{d}Q}{%
\mathrm{d}P}\right) \right] =
\end{equation*}%
\begin{equation*}
\lambda \left( \mathbb{E}\left[ \left( \sum_{n=1}^{N}X^{n}\right) \mathbb{E}%
\left[ \frac{\mathrm{d}Q}{\mathrm{d}P}\middle|\mathcal{G}\right] \right]
+A\right) +\sum_{n=1}^{N}\mathbb{E}\left[ \mathbb{E}\left[ v_{n}\left(
\lambda \frac{\mathrm{d}Q}{\mathrm{d}P}\right) \middle|\mathcal{G}\right] %
\right] \geq
\end{equation*}%
\begin{equation*}
\lambda \left( \sum_{n=1}^{N}\mathbb{E}\left[ X^{n}\mathbb{E}\left[ \frac{%
\mathrm{d}Q}{\mathrm{d}P}\middle|\mathcal{G}\right] \right] +A\right)
+\sum_{n=1}^{N}\mathbb{E}\left[ v_{n}\left( \lambda \mathbb{E}\left[ \frac{%
\mathrm{d}Q}{\mathrm{d}P}\middle|\mathcal{G}\right] \right) \right] \,,
\end{equation*}%
where in the last inequality we exploited the tower property and Jensen
inequality, as $v_{1},\dots ,v_{N}$ are convex. Notice now that $\mathbb{E}%
\left[ \frac{\mathrm{d}Q}{\mathrm{d}P}\middle|\mathcal{G}\right] $ defines
again a probability measure (on the whole $\mathcal{F}$, the initial sigma
algebra) and that this measure still belongs to $\mathcal{Q}$ since all its
components are equal. As a consequence, the minimum in equation (\ref%
{eqminimaxapplied3}) can be equivalently taken over $\lambda \in {\mathbb{R}}%
_{++}$ (as before) and $\mathbf{Q}\in \mathcal{Q}\cap \left( L^{0}(\Omega ,%
\mathcal{G},P)\right) ^{N}$. The claim for $\widehat{\mathbf{Y}}$ follows
from (\ref{YY}).
\end{proof}

It is interesting to notice that this dependence on the componentwise sum of 
$\mathbf{X}$ also holds in the case of B\"{u}hlmann's equilibrium (see \cite%
{Buhlmann} page 16 and \cite{Borch}).

\begin{remark}
In the case a cluster of agents, see the Example \ref{exCh}, the above
result can be clearly generalized: the $i$-th component of the vector $%
\widehat{\mathbf{Q}}$, for $i$ belonging to the $m$-th group, only depends
on the sum of those components of $\mathbf{X}$ whose corresponding indexes
belong to the $m$-th group itself. It is also worth mentioning that if we
took $\mathcal{B}^{(\mathbf{I})}=\mathbb{R}^{N}$, we would see that each
component of $\widehat{\mathbf{Q}}$ and of $\widehat{\mathbf{Y}}$ is a measurable function of the corresponding component of $\mathbf{X}$. This is reasonable
since, in this case, at the final time each agent would be only allowed to share and exchange
risk with herself/himself and the systemic features of the model we are
considering would be lost.
\end{remark}

We provide now some additional examples, to the ones in Example \ref{exCh},
of possible feasible sets $\mathcal{B}$ and study the dependence of the
probability measures from $\mathcal{B}$.

\begin{example}
\label{exmixtures} Consider a measurable partition $A_{1},\dots ,A_{K}$ of $%
\Omega $ and a collection of partitions $\mathbf{I}^{1},\dots ,\mathbf{I}%
^{K} $ of $\{1,\dots ,N\}$ as in Example \ref{exCh}. Take the associated
clusterings $\mathcal{B}^{(\mathbf{I}^{1})},\dots ,\mathcal{B}^{(\mathbf{I}%
^{K})}$ defined as in \eqref{C0}. Then the set 
\begin{equation}
\mathcal{B}:=\left( \sum_{i=1}^{K}\mathcal{B}^{(\mathbf{I}%
^{i})}1_{A_{i}}\right) \cap \mathcal{C}_{\mathbb{R}}  \label{randomcluster}
\end{equation}%
satisfies Assumptions \ref{A00} and is closed under truncation, as it can be
checked directly.
\end{example}

The set in (\ref{randomcluster}) can be seen as a scenario-dependent
clustering. A particular simple case of \eqref{randomcluster} is the
following. For a measurable set $A_1\in\mathcal{F}$ take $%
A_2=\Omega\setminus A_1$. Then set $\mathcal{C}_{\mathbb{R}} 1_{A_1}+{%
\mathbb{R}}^N 1_{A_2}$ is of the form \eqref{randomcluster} and consists of
all the $\mathbf{Y}\in (L^{0})^N$ such that (i) there exists a real number $%
\sigma\in{\mathbb{R}}$ with $\sum_{n=1}^NY^n=\sigma$ $P-$a.s. on $A_1$, (ii)
there exists a vector $\mathbf{b}\in{\mathbb{R}}^N$ such that $\mathbf{Y}=%
\mathbf{b}$ $P-$a.s. on $A_2$ and (iii) $\sigma=\sum_{n=1}^Nb^n$ (recall
that $\mathbf{Y}\in\mathcal{C}_{\mathbb{R}}$ by \eqref{randomcluster}).

Let us motivate Example \ref{exmixtures} with the following practical
example. Suppose for each bank $i$ a regulator establishes an excessive
exposure threshold $D^i$. If the position of bank $i$ falls below such
threshold, we can think that it is too dangerous for the system to let that
bank take part to the risk exchange. As a consequence, in the clustering
example, on the event $\{X^i\leq D^i\}$ we can require the bank to be left
alone. Also the symmetric situation can be considered: a bank $j$ whose
position is too good, say exceeding a value $A^j$, will not be willing to
share risk with all others, thus entering the game only as isolated
individual or as a member of the groups of \textquotedblleft
safer\textquotedblright \text{ }banks. Both these requirements, and many others (say
considering random thresholds) can be modelled with the constraints
introduced in Example \ref{exmixtures}.

It is interesting to notice that, as in Example \ref{exCh}, assuming a
constraint set of the form given in Example \ref{exmixtures} forces a
particular behavior on the probability vectors in $\mathcal{Q}_{v}$.

\begin{lemma}
\label{lemmaalleqonf}Let $\mathcal{B}$ be as in Example \ref{exmixtures} and
let $\mathbf{Q\in }\mathcal{Q}_{v}$. Fix any $i\in \left\{ 1,...,K\right\} $
and any group $\mathbf{I}_{m}^{i}$ of the partition $\mathbf{I}^{i}=(\mathbf{%
I}_{m}^{i})_{m}$. Then all the components $Q^{j}$, $j\in \mathbf{I}_{m}^{i}$%
, agree on $\mathcal{F}|_{A_{i}}:=\{F\cap A_{i},\,F\in \mathcal{F}\}$.
\end{lemma}

\begin{proof}
We think it is more illuminating to prove the statement in a simplified
case, rather than providing a fully formal proof (which would require
unnecessairly complicated notation). This is "without loss of generality" in
the sense that it is clear how to generalize the method. To this end, let us
consider the case $K=2$ (i.e. $A_2=A_1^c$) and $\mathcal{B}^{(\mathbf{I^1})}
:=\mathcal{C}_{\mathbb{R}}$, $\mathcal{B}^{(\mathbf{I^2})}:={\mathbb{R}}^N$.
For any $F\in\mathcal{F}$ and $i,j\in\{1,\dots,N\}$ we can take $\mathbf{Y}%
:=\left(1_{F}(\mathbf{e_i}-\mathbf{e_j})\right)1_{A_1}+\mathbf{0}1_{A_2}$ to
obtain $\mathbf{Y}\in\mathcal{C}_{\mathbb{R}} 1_{A_1}+{\mathbb{R}}^N1_{A_2}$%
, $\sum_{j=1}^NY^j=0$. By definition of $\mathcal{Q}_v$ we get for any $%
\mathbf{Q}\in\mathcal{Q}_v$ that $Q^i(A\cap F)-Q^j(A\cap F)\leq 0$, and
interchanging $i,j$ yields $Q^i(A\cap F)=Q^j(A\cap F)$ for any $i,j=1\dots,N$%
, $F\in\mathcal{F}$.
\end{proof}

\section{Exponential Case}
\label{sectionexp}
 We now specialize our analysis to the exponential setup,
where 
\begin{equation}
u_{n}(x):=1-\exp (-\alpha _{n}x),\,n=1,\dots ,N\,\,\,\,\,\text{ for }%
\,\,\,\,\,\alpha _{1},\dots ,\alpha _{N}>0.  \label{exponut}
\end{equation}%
This allows us to provide explicit formulas for a wide range of constraint
sets $\mathcal{B}$ (namely, all those introduced in Example \ref{exCh}) and
so the stability properties of SORTE, with respect to a different weighting
of utilities, will be evident.

\subsection{Explicit formulas}

We consider a set of constraints of the form $\mathcal{B}=\mathcal{B}^{(%
\mathbf{I})}$ as given in Example \ref{exCh}. Given $\mathbf{X}\in M^{\Phi }$
and $m\in \{1,\dots ,h\}$, we set: 
\begin{eqnarray*}
\beta _{m} &:&=\sum_{n\in I_{m}}\frac{1}{\alpha _{n}}\,\,\,\,\,\,\,\beta
:=\sum_{n=1}^{N}\frac{1}{\alpha _{n}}\,\,\,\,\,\,\,\,\overline{X}%
_{m}:=\sum_{n\in I_{m}}X^{n}, \\
R(n) &:&=\frac{\frac{1}{\alpha _{n}}}{\sum_{k=1}^{N}\frac{1}{\alpha _{k}}}%
\text{, }n=1,...N\text{, }\alpha :=(\alpha _{1},...,\alpha _{N}),\text{ }%
E_{R}\left[ \ln (\alpha )\right] =\sum_{n=1}^{N}R(n)\ln (\alpha _{n}).
\end{eqnarray*}

\begin{theorem}
\label{thmformulasexp} Take $u_{1},\dots ,u_{N}$ as given by \eqref{exponut}
and $\mathcal{B}=\mathcal{B}^{\mathbf{(I)}}$ as in Example \ref{exCh}. For $%
\mathcal{L}$ and $\mathscr{Q}$ defined in Theorem \ref{thmsorteexistsA}, the
SORTE is given by 
\begin{equation}
\begin{cases}
\widehat{Y}^{k}=-X^{k}+\frac{1}{\alpha _{k}}\left( \frac{\overline{X}_{m}}{%
\beta _{m}}-d_{m}(\mathbf{X})\right) +\frac{1}{\alpha _{k}}\left[ \frac{A}{%
\beta }+\ln \left( \alpha _{k}\right) -E_{R}\left[ \ln (\alpha )\right] %
\right] & \,\,\,\,\,k\in I_{m} \\ 
\frac{\mathrm{d}\widehat{Q}^{k}}{\mathrm{d}P}=\frac{\exp \left( -\frac{%
\overline{X}_{m}}{\beta _{m}}\right) }{\mathbb{E}\left[ \exp \left( -\frac{%
\overline{X}_{m}}{\beta _{m}}\right) \right] }=:\frac{\mathrm{d}\widehat{Q}%
^{m}}{\mathrm{d}P} & \,\,\,\,\,k\in I_{m}\label{defqexpgen} \\ 
\widehat{a}^{k}=E_{\widehat{Q}^{k}}[\widehat{Y}^{k}] & \,\,\,\,\,k=1,\dots ,N%
\end{cases}%
\end{equation}%
where 
\begin{equation*}
d_{m}(\mathbf{X}):=\left[ \sum_{j=1}^h\frac{\beta _{j}}{\beta }\ln \left( 
\mathbb{E}\left[ \exp \left( -\frac{\overline{X}_{j}}{\beta _{j}}\right) %
\right] \right) \right] -\ln \left( \mathbb{E}\left[ \exp \left( -\frac{%
\overline{X}_{m}}{\beta _{m}}\right) \right] \right) \,.
\end{equation*}
\end{theorem}

\begin{proof}
The utilities in \eqref{exponut} satisfy Assumption \ref{A00} (a) and $%
\mathcal{B}$ satisfies Asssumption \ref{A00} (b) and closedness under
truncation, hence Theorems \ref{thmsorteexistsA} and \ref{thmsorteuniqueA}
guarantee existence and uniqueness. Recall that from this choice of $%
\mathcal{B}$ we have that for each $\mathbf{Q}\in \mathcal{Q}_{v}$, all the
components of $\mathbf{Q}$ are equal in each index subset $I_{m}$.

It is easy to check that 
\begin{equation}
v_{n}(\lambda y)=\frac{\lambda y}{\alpha _{n}}\ln \frac{\lambda }{\alpha _{n}%
}+\frac{\lambda }{\alpha _{n}}y\ln y-\frac{\lambda }{\alpha _{n}}y+1\,.
\label{eqvgen}
\end{equation}%
Substitute now $y=\frac{\mathrm{d}Q^{n}}{\mathrm{d}P}\in \mathcal{Q}_{v}$ in
the above expressions and take expectations to get 
\begin{equation}
\mathbb{E}\left[ v_{n}\left( \lambda \frac{\mathrm{d}Q^{n}}{\mathrm{d}P}%
\right) \right] =\phi _{n}(\lambda )+\frac{\lambda }{\alpha _{n}}\mathbb{E}%
\left[ \frac{\mathrm{d}Q^{n}}{\mathrm{d}P}\ln \left( \frac{\mathrm{d}Q^{n}}{%
\mathrm{d}P}\right) \right] ,\,\,\,\,\phi _{n}(\lambda )=\frac{\lambda }{%
\alpha _{n}}\ln \frac{\lambda }{\alpha _{n}}-\frac{\lambda }{\alpha _{n}}%
+1\,.  \label{expectvgen}
\end{equation}%
Let $K\left( \lambda ,\frac{\mathrm{d}\mathbf{Q}}{\mathrm{d}P}\right) $ be
the functional to be optimized in \eqref{eqminimaxapplied3}. Set 
\begin{equation*}
\xi :=\sum_{n=1}^{N}\frac{1}{\alpha _{n}}\ln \left( \frac{1}{\alpha _{n}}%
\right) ,\,\,\phi (\lambda )=\sum_{n=1}^{N}\phi _{n}(\lambda )=\lambda \xi
+\beta \lambda \ln \lambda -\lambda \beta +N.\,\ 
\end{equation*}%
Then from \eqref{expectvgen} we deduce 
\begin{equation}
K\left( \lambda ,\frac{\mathrm{d}\mathbf{Q}}{\mathrm{d}P}\right) ={\lambda }%
\left( \sum_{n=1}^{N}E_{{Q}^{n}}\left[ X^{n}\right] +A\right) +\phi ({%
\lambda })+\sum_{n=1}^{N}\frac{{\lambda }}{\alpha _{n}}\mathbb{E}\left[ 
\frac{\mathrm{d}Q^{n}}{\mathrm{d}P}\ln \left( \frac{\mathrm{d}Q^{n}}{\mathrm{%
d}P}\right) \right] .  \label{formulaKgen}
\end{equation}%
Set 
\begin{equation}
\mu :=\sum_{n=1}^{N}\frac{1}{\alpha _{n}}\mathbb{E}\left[ \frac{\mathrm{d}%
\widehat{Q}^{n}}{\mathrm{d}P}\ln \left( \frac{\mathrm{d}\widehat{Q}^{n}}{%
\mathrm{d}P}\right) \right] +A+\sum_{n=1}^{N}\mathbb{E}_{\widehat{Q}%
^{n}}[X^{n}].  \label{defmugen}
\end{equation}%
From \eqref{formulaKgen} and \eqref{defmugen} 
\begin{equation*}
K\left( \lambda ,\frac{\mathrm{d}\widehat{\mathbf{Q}}}{\mathrm{d}P}\right)
=\lambda \mu +\lambda \left( \xi +\beta \ln (\lambda )-\beta \right) +N\,.
\end{equation*}%
The associated first order condition obtained differentiating in $\lambda $
yields the unique solution 
\begin{equation*}
\widehat{\lambda }=\exp \left( -\frac{\mu +\xi }{\beta }\right) 
\end{equation*}%
which can be substituted in $K\left( \cdot ,\frac{\mathrm{d}\widehat{\mathbf{%
Q}}}{\mathrm{d}P}\right) $ yielding 
\begin{equation}
K\left( \widehat{\lambda },\frac{\mathrm{d}\widehat{\mathbf{Q}}}{\mathrm{d}P}%
\right) =-\widehat{\lambda }\beta +N\,.\label{formulaklambdahatgen}
\end{equation}%
We now \textbf{guess} that the vector of measures $\widehat{\mathbf{Q}}$
defined via \eqref{defqexpgen} is optimal and compute the associated $\mu $: 
\begin{equation*}
\mu =\sum_{n=1}^{N}{E}_{\widehat{Q}^{n}}[X^{n}]+A+\sum_{n=1}^{N}\frac{1}{%
\alpha _{n}}\mathbb{E}\left[ \frac{\mathrm{d}\widehat{Q}^{n}}{\mathrm{d}P}%
\ln \left( \frac{\mathrm{d}\widehat{Q}^{n}}{\mathrm{d}P}\right) \right]
=A+\sum_{j=1}^{h}\mathbb{E}\left[ \left( \overline{X}_{j}\right) \frac{%
\mathrm{d}\widehat{Q}^{j}}{\mathrm{d}P}\right] +
\end{equation*}%
\begin{equation*}
\sum_{j=1}^{h}\beta _{j}\mathbb{E}\left[ \frac{\mathrm{d}\widehat{Q}^{j}}{%
\mathrm{d}P}\ln \left( \exp \left( -\frac{\overline{X}_{j}}{\beta _{j}}%
\right) \right) \right] +\sum_{j=1}^{h}\beta _{j}\mathbb{E}\left[ \frac{%
\mathrm{d}\widehat{Q}^{j}}{\mathrm{d}P}\ln \left( \frac{1}{\mathbb{E}\left[
\exp \left( -\frac{\overline{X}_{j}}{\beta _{j}}\right) \right] }\right) %
\right] \,.
\end{equation*}%
Hence 
\begin{equation}
\mu =A-\sum_{j=1}^{h}\beta _{j}\ln \left( \mathbb{E}\left[ \exp \left( -%
\frac{\overline{X}_{j}}{\beta _{j}}\right) \right] \right)   \label{eqmugen}
\end{equation}%
and substituting \eqref{formulaklambdahatgen} in the explicit formula for $%
\widehat{\lambda }$ we get 
\begin{equation}
K\left( \widehat{\lambda },\frac{\mathrm{d}\widehat{\mathbf{Q}}}{\mathrm{d}P}%
\right) =-\beta \exp \left( -\frac{1}{\beta }\left( A+\xi
+\sum_{j=1}^{h}\beta _{j}\ln \left( \mathbb{E}\left[ \exp \left( -\frac{%
\overline{X}_{j}}{\beta _{j}}\right) \right] \right) \right) \right) +N\,.
\label{eqkcandidategen}
\end{equation}%
Using equation \eqref{YY} we \textbf{define}, for the measure given in %
\eqref{defqexpgen}, 
\begin{equation*}
\widehat{Y}^{k}=-X^{k}-v_{n}^{\prime }\left( \widehat{\lambda }\frac{\mathrm{%
d}\widehat{Q}}{\mathrm{d}P}\right) \,\,\,\,k=1,\dots ,N\,.
\end{equation*}%
By \eqref{eqvgen} (with $\lambda =1$) we obtain, for $k\in I_{m},$ $%
v_{k}^{\prime }(y)=\frac{1}{\alpha _{k}}\ln \left( \frac{y}{\alpha _{k}}%
\right) $ and 
\begin{equation*}
v_{k}^{\prime }\left( \widehat{\lambda }\frac{\mathrm{d}\widehat{Q}}{\mathrm{%
d}P}\right) =\frac{1}{\alpha _{k}}\ln \left( \frac{1}{\alpha _{k}}\right) +%
\frac{1}{\alpha _{k}}\ln \left( \frac{\exp \left( -\frac{\overline{X}_{m}}{%
\beta _{m}}-\frac{A+\mu }{\beta }\right) }{\mathbb{E}\left[ \exp \left( -%
\frac{\overline{X}_{m}}{\beta _{m}}\right) \right] }\right) 
\end{equation*}%
\begin{equation*}
=\frac{1}{\alpha _{k}}\ln \left( \frac{1}{\alpha _{k}}\right) -\frac{1}{%
\alpha _{k}}\left( \frac{\overline{X}_{m}}{\beta _{m}}+\frac{A+\mu }{\beta }%
\right) -\frac{1}{\alpha _{k}}\ln \left( \mathbb{E}\left[ \exp \left( -\frac{%
\overline{X}_{m}}{\beta _{m}}\right) \right] \right) \overset{\text{Eq.}%
\eqref{eqmugen}}{=}
\end{equation*}%
\begin{equation*}
\frac{1}{\alpha _{k}}\ln \left( \frac{1}{\alpha _{k}}\right) -\frac{1}{%
\alpha _{k}}\left( \frac{\overline{X}_{m}}{\beta _{m}}+\frac{A+\xi }{\beta }%
\right) +\frac{1}{\alpha _{k}}d_{m}(\mathbf{X})\,.
\end{equation*}%
Hence for $k\in I_{m}$ we have 
\begin{equation*}
\widehat{Y}^{k}=-X^{k}+\frac{1}{\alpha _{k}}\left( \frac{\overline{X}_{m}}{%
\beta _{m}}+\frac{A+\xi }{\beta }-d_{m}(\mathbf{X})\right) -\frac{1}{\alpha
_{k}}\ln \left( \frac{1}{\alpha _{k}}\right) \,.
\end{equation*}%
A simple computation yields $\widehat{\mathbf{Y}}\in M^{\Phi }$, $\sum_{k\in
I_{m}}Y^{k}\in {\mathbb{R}}$ and $\sum_{n=1}^{N}\widehat{Y}^{n}=A$, so that $%
\widehat{\mathbf{Y}}\in \mathcal{B}_{A}\cap M^{\Phi }$.

Moreover%
\begin{equation*}
\exp \left( -\left( X^{k}+\widehat{Y}^{k}\right) \right) =\exp \left(
-\alpha _{k}\left( \frac{1}{\alpha _{k}}\left( \frac{\overline{X}_{m}}{\beta
_{m}}+\frac{A+\xi }{\beta }-d_{m}(\mathbf{X})\right) -\frac{1}{\alpha _{k}}%
\ln \left( \frac{1}{\alpha _{k}}\right) \right) \right) 
\end{equation*}%
\begin{equation*}
=\frac{1}{\alpha _{k}}\exp \left( -\frac{\overline{X}_{m}}{\beta _{m}}%
\right) \exp \left( -\frac{A+\xi }{\beta }\right) \exp \left( d_{m}(\mathbf{X%
})\right) 
\end{equation*}%
\begin{equation*}
=\frac{1}{\alpha _{k}}\frac{\exp \left( -\frac{\overline{X}_{m}}{\beta _{m}}%
\right) }{\mathbb{E}\left[ \exp \left( -\frac{\overline{X}_{m}}{\beta _{m}}%
\right) \right] }\exp \left( -\frac{A+\xi }{\beta }\right) \exp \left(
\sum_{j=1}^{h}\frac{\beta _{j}}{\beta }\ln \left( \mathbb{E}\left[ \exp
\left( -\frac{\overline{X}_{j}}{\beta _{j}}\right) \right] \right) \right)
\,.
\end{equation*}%
As a consequence 
\begin{equation*}
\sum_{n=1}^{N}\mathbb{E}\left[ 1-\exp \left( -\alpha _{n}\left( X^{n}+%
\widehat{Y}^{n}\right) \right) \right] 
\end{equation*}%
\begin{equation}
=-\sum_{n=1}^{N}\frac{1}{\alpha _{n}}\exp \left( -\frac{1}{\beta }\left(
A+\xi +\sum_{j=1}^{h}\beta _{j}\ln \left( \mathbb{E}\left[ \exp \left( -%
\frac{\overline{X}_{j}}{\beta _{j}}\right) \right] \right) \right) \right) +N%
\overset{\text{Eq.}\eqref{eqkcandidategen}}{=}K\left( \widehat{\lambda },%
\frac{\mathrm{d}\widehat{\mathbf{Q}}}{\mathrm{d}P}\right) 
\label{valueofkopt}
\end{equation}%
which implies%
\begin{equation}
\sum_{n=1}^{N}\mathbb{E}\left[ u_{n}\left( X^{n}+\widehat{Y}^{n}\right) %
\right] =K\left( \widehat{\lambda },\frac{\mathrm{d}\widehat{\mathbf{Q}}}{%
\mathrm{d}P}\right) \,.  \label{eqexputiskcandidate}
\end{equation}%
To sum up we have 
\begin{equation*}
K\left( \widehat{\lambda },\frac{\mathrm{d}\widehat{\mathbf{Q}}}{\mathrm{d}P}%
\right) \overset{\text{Eq.}\eqref{eqexputiskcandidate}}{=}\sum_{n=1}^{N}%
\mathbb{E}\left[ u_{n}\left( X^{n}+\widehat{Y}^{n}\right) \right] \overset{%
\widehat{\mathbf{Y}}\in \mathcal{B}_{A}\cap M^{\Phi }}{\leq }
\end{equation*}%
\begin{equation*}
\sup_{\mathbf{Y}\in \mathcal{B}_{A}\cap M^{\Phi }}\sum_{n=1}^{N}\mathbb{E}%
\left[ u_{n}(X^{n}+Y^{n})\right] \overset{\text{Thm.}\eqref{thmoptimumexists}%
}{=}\min_{\substack{ \lambda >0 \\ \mathbf{Q}\in \mathcal{Q}_{v}}}K\left(
\lambda ,\frac{\mathrm{d}\mathbf{Q}}{\mathrm{d}P}\right) \leq K\left( 
\widehat{\lambda },\frac{\mathrm{d}\widehat{\mathbf{Q}}}{\mathrm{d}P}\right)
\,.
\end{equation*}%
Consequently $\widehat{\mathbf{Y}}$ is the (unique) optimum for the
optimization problem in LHS of \eqref{eqminimaxapplied2}, and $(\widehat{%
\lambda },\widehat{\mathbf{Q}})$ is the (unique) optimum to the minimization
problem in \eqref{eqminimaxapplied3}.

Moreover, setting $\widehat{a}^{n}:=E_{\widehat{Q}^{n}}[\widehat{Y}%
^{n}],\quad n=1,\dots ,N,$ the SORTE (which, as already argued, exists and
is unique) is given by $\left( \widehat{\mathbf{Y}},\widehat{\mathbf{Q}},%
\widehat{\mathbf{a}}\right) $.
\end{proof}

\begin{remark}
We observe that in the terminal part of the proof above we also got an
explicit formula for the maximum systemic utility: 
\begin{equation}  \label{valueofsupexp}
\sup_{\mathbf{Y}\in \mathcal{B}_{A}\cap M^{\Phi }}\sum_{n=1}^{N}\mathbb{E}%
\left[ u_{n}(X^{n}+Y^{n})\right] \overset{\text{Thm.}\eqref{thmoptimumexists}%
}{=}K\left( \widehat{\lambda },\frac{\mathrm{d}\widehat{\mathbf{Q}}}{\mathrm{%
d}P}\right)
\end{equation}
where $K\left( \widehat{\lambda },\frac{\mathrm{d}\widehat{\mathbf{Q}}}{%
\mathrm{d}P}\right)$ is given in \eqref{valueofkopt}.
\end{remark}

\subsection{A toy Example\label{Ex}}

In the following two examples we compare a B\"{u}hlmann's Equilibrium with a
SORTE in the simplest case where $\mathbf{X}=\mathbf{0}:=(0,...,0)$ and $%
A=0. $ In the formula below we use the well known fact:%
\begin{equation*}
\sup_{Y\in L^{1}(Q)}\left\{ \mathbb{E}\left[ u_{n}(Y)\right] \mid
E_{Q}[Y]\leq x\right\} =1-e^{-\alpha _{n}x-H(Q,P)}\,,
\end{equation*}%
where $H(Q,P)=E[\frac{dQ}{dP}\ln (\frac{dQ}{dP})]$ is the relative entropy,
for $Q\ll P$.

\begin{example}[B\"{u}hlmann's equilibrium solution]
As $\mathbf{X}:=\mathbf{0}$ then $\overline{X}_{N}=\sum_{k=1}^{N}X^{k}=0$
and therefore the optimal probability measure $Q_{\mathbf{X}}$ defined in B%
\"{u}hlmann is: 
\begin{equation}
\frac{dQ_{\mathbf{X}}}{d\mathbb{P}}:=\frac{e^{-\frac{1}{\beta }\overline{X}%
_{N}}}{\mathbb{E}\left[ e^{-\frac{1}{\beta }\overline{X}_{N}}\right] }=1%
\text{,}  \label{Q}
\end{equation}%
i.e. $Q_{\mathbf{X}}=P.$ Take $\mathbf{a=0}=(0,...,0)$. We compute%
\begin{equation*}
U_{n}^{Q_{\mathbf{X}}}(0)=U_{n}^{P}(0):=\sup \left\{ \mathbb{E}\left[
u_{n}(0+Y)\right] \mid E_{P}[Y]\leq 0\right\} =1-e^{-\alpha
_{n}0-H(P,P)}=1-1=0\,,
\end{equation*}%
as $H(P,P)=0$, so that 
\begin{equation*}
\sum_{n=1}^{N}U_{n}^{P}(0)=0.
\end{equation*}%
As a consequence, and as $u_{n}(0)=0$, the optimal solution for each single $%
n$ is obviously $Y_{\mathbf{X}}^{n}=0.$

\textbf{Conclusion:} The B\"{u}hlmann's equilibrium solution associated to $%
\mathbf{X}:=\mathbf{0}$ (and $A=0)$ is the couple $(\mathbf{Y}_{\mathbf{X}},%
\mathbf{Q}_{\mathbf{X}})=(\mathbf{0},P)$. Here the vector $\mathbf{a}$ is
taken a priori to be equal to $(0,...,0).$
\end{example}

\begin{example}[SORTE]
\label{ex2}From Theorem \ref{thmformulasexp} with $\mathbf{X}:=\mathbf{0}$
and $A=0$ we obtain for the SORTE that: the optimal probability measure $%
\widehat{\mathbf{Q}}$ coincides again with $P$; the optimal $\widehat{Y}$ is:%
\begin{equation}
\widehat{Y}^{n}=\frac{1}{\alpha _{n}}\left[ \ln (\alpha _{n})-E_{R}\left[
\ln (\alpha )\right] \right] :=\widehat{a}^{n}.  \label{aa}
\end{equation}%
Recalling that $\widehat{\mathbf{Q}}$ is in fact a minimax measure for the
optimization problem $\pi _{0}(\mathbf{0})$ (see the proof of Theorem \ref%
{thmsorteexistsA}), we can say that%
\begin{equation}
S^{P}(0)=S^{\widehat{\mathbb{Q}}}(A)\overset{\text{Lemma.}\ref%
{lemmalinkpiSPi2}}{=}\pi _{0}(\mathbf{0})\overset{\eqref{valueofkopt},%
\eqref{valueofsupexp}}{=}N-\beta e^{-\frac{\xi }{\beta }}  \label{S}
\end{equation}%
Notice that if the $\alpha _{n}$ are equal for all $n,$ then $S^{P}(0)=0,$
but in general 
\begin{equation*}
S^{P}(0)=N-\beta e^{-\frac{\xi }{\beta }}\geq 0.
\end{equation*}%
Indeed, by Jensen inequality:%
\begin{equation*}
e^{-\frac{\xi }{\beta }}=e^{E_{R}[\ln (\alpha )]}\leq E_{R}[e^{\ln (\alpha
)}]=E_{R}[\alpha ]:=\sum_{n=1}^{N}\frac{\frac{1}{\alpha _{n}}\alpha _{n}}{%
\sum_{k=1}^{N}\frac{1}{\alpha _{k}}}=\frac{N}{\beta }.
\end{equation*}%
From (\ref{aa}) we deduce that the $\alpha _{n}$ are equal for all $n$ if
and only if $\widehat{a}^{n}=0$ for all $n$, but in general $\widehat{a}^{n}$
may differ from $0$. As $\widehat{Y}^{n}=\widehat{a}^{n}$, the same holds
also for the optimal solution $\widehat{Y}$. When $\widehat{a}^{n}<0$ a
violation of Individual Rationality occurs.

\textbf{Conclusion:} The SORTE solution associated to $\mathbf{X}:=\mathbf{0}
$ (and $A=0)$ is the triplet $(\widehat{\mathbf{Y}},P,\widehat{\mathbf{a}})$
where $\widehat{\mathbf{Y}}=\widehat{\mathbf{a}}$ is assigned in equation (%
\ref{aa}).
\end{example}

The above comparison shows that a SORTE is not a B\"{u}hlmann equilibrium,
even when $\mathbf{X}:=\mathbf{0}$ and $A=0$. When the $\alpha _{n}$ are all
equal, then the B\"{u}hlmann and the SORTE solution coincide, as all agents\
are assumed to have the same risk aversion.

\begin{remark}
In this example, notice that we may control the risk sharing components $%
Y^{n}$ of agent $n$ in the SORTE by:%
\begin{equation*}
|Y^{n}|\leq \frac{1}{\alpha _{\min }}\left[ \ln (\alpha _{\max })-\ln
(\alpha _{\min })\right] .
\end{equation*}%
Suppose that $\alpha _{\min }<\alpha _{\max }$ and consider the expression
for $\widehat{Y}^{n}=\widehat{a}^{n}$ in (\ref{aa}). If $\alpha _{j}=\alpha
_{\min }$ then the corresponding $\widehat{Y}^{j}<0$ is in absolute value
relatively large (divide by $\alpha _{\min }$), while if $\alpha _{k}=\alpha
_{\max }$ the corresponding $\widehat{Y}^{k}>0$ is in absolute value
relatively small (divide by $\alpha _{\max }$).
\end{remark}

\subsection{Dependence on weights and stability}

\label{secdeponweight} We now provide a detailed study of the dependence on
weights, as introduced in Remark \ref{remgamma}, in the exponential case.
Given $\gamma _{n}\in (0,+\infty ),\,n=1,\dots ,N$ and $u_{1},\dots ,u_{N}$
satisfying Assumption \ref{A00} (a), we recall that $u_{n}^{\gamma
}(x):=\gamma _{n}u_{n}(x),\,n=1,\dots ,N$ and we denote by $v_{n}^{\gamma
}(\cdot )$ their convex conjugates. These functions $u_{n}^{\gamma }$
satisfy Assumption \ref{A00} (a).

In our exponential setup and under closedness under truncation, a different
weighting only results in a translation of both allocations at initial and
terminal time of a SORTE, without affecting the optimal measure:

\begin{proposition}
\label{proptrasl} Consider $u_{1},\dots ,u_{N}$ as given in \eqref{exponut}
and take the associated $u_{1}^{\gamma },\dots ,u_{N}^{\gamma }$ as above.
Suppose $\mathcal{B}$ satisfies Assumption \ref{A00} (b) and is closed under
truncation. Call $\left( \widehat{\mathbf{Y}},\widehat{\mathbf{Q}},\widehat{%
\mathbf{a}}\right) $ the unique SORTE associated to $u_{1},\dots ,u_{N}$,
and similarly define $\left( \widehat{\mathbf{Y}}_{\gamma },\widehat{\mathbf{%
Q}}_{\gamma },\widehat{\mathbf{a}}_{\gamma }\right) $ as the unique SORTE
associated to $u_{1}^{\gamma },\dots ,u_{N}^{\gamma }$. Then%
\begin{equation}
\begin{cases}
\widehat{Y}_{\gamma }^{k}=\widehat{Y}^{k}+g_{k}(\gamma ) & 
\,\,\,\,\,k=1,\dots ,N \\ 
\frac{\mathrm{d}\widehat{Q}_{\gamma }^{k}}{\mathrm{d}P}=\frac{\mathrm{d}%
\widehat{Q}^{k}}{\mathrm{d}P} & \,\,\,\,\,k=1,\dots ,N\notag \\ 
\widehat{a}_{\gamma }^{k}=\widehat{a}^{k}+g_{k}(\gamma ) & 
\,\,\,\,\,k=1,\dots ,N%
\end{cases}
\label{eqsortetrasl}
\end{equation}%
where 
\begin{equation*}
g_{k}(\gamma ):= \frac{1}{\alpha _{k}}\frac{\sum_{n=1}^{N}\frac{1}{\alpha
_{n}}\ln \left( \frac{1}{\gamma _{n}}\right) }{\sum_{n=1}^{N}\frac{1}{\alpha
_{n}}}-\frac{1}{\alpha _{k}}\ln \left( \frac{1}{\gamma _{k}}\right)=\frac{1}{%
\alpha_k}\left(\ln (\gamma _{n})-E_{R}[\ln (\gamma
)]\right)\,\,\,\,\,k=1,\dots ,N\,.
\end{equation*}
\end{proposition}

\begin{proof}
For a general set $\mathcal{B}$, we here provide only a sketch of the proof.
Using the formulas for $v_{1},\dots ,v_{N}$, after some computations one can
write explicitly the minimax expression \eqref{eqminimaxapplied3}. Then use
the gradient formula (\ref{YY}) to deduce (\ref{eqsortetrasl}). A more
direct proof, that works only for sets $\mathcal{B}$ in the form described
in Example \ref{exCh}, is based on the observation that 
\begin{equation*}
u_{n}^{\gamma }(x):=\gamma _{n}u_{n}(x)=\gamma _{n}-\gamma _{n}\exp (-\alpha
_{n}x)=\gamma _{n}-\exp \left( -\alpha _{n}\left[ x-\frac{1}{\alpha _{n}}\ln
(\gamma _{n})\right] \right) .
\end{equation*}%
Hence, $\left( \widehat{\mathbf{Y}}_{\gamma },\widehat{\mathbf{Q}}_{\gamma },%
\widehat{\mathbf{a}}_{\gamma }\right) $ can be obtained by a straightforward
computation from the solution $\left( \widehat{\mathbf{Y}},\widehat{\mathbf{Q%
}},\widehat{\mathbf{a}}\right) $, which is explicitly given in Theorem \ref%
{thmformulasexp}, using $X^{n}-\frac{1}{\alpha _{n}}\ln (\gamma
_{n}),\,n=1,\dots ,N$ in place of $\mathbf{X}$.
\end{proof}

\appendix

\section{Appendix}

\subsection{Orlicz Spaces and Utility Functions\label{secorlicz}}

We consider the utility maximization problem defined on Orlicz spaces, see 
\cite{RaoZen} for further details on Orlicz spaces. This presents several
advantages. From a mathematical point of view, it is a more general setting
than $L^{\infty }$, but at the same time it simplifies the analysis, since
the topology is order continuous and there are no singular elements in the
dual space. Furthermore, it has been shown in \cite{BF08AAP} that the Orlicz
setting is the natural one to embed utility maximization problems, as the
natural integrability condition $\mathbb{E}[u(X)]>-\infty $ is implied by $%
\mathbb{E}[\phi (X)]<+\infty $.

Let $u:\mathbb{R}\rightarrow \mathbb{R}$ be a concave and increasing
function satisfying $\lim_{x\rightarrow -\infty }\frac{u(x)}{x}=+\infty .$
Consider $\phi (x):=-u(-|x|)+u(0).$ Then $\phi :\mathbb{R}\rightarrow
\lbrack 0,+\infty )$ is a strict Young function, i.e., it is finite valued,
even and convex on $\mathbb{R}$ with $\phi (0)=0$ and $\lim_{x\rightarrow
+\infty }\frac{\phi (x)}{x}=+\infty $. The Orlicz space $L^{\phi }$ and
Orlicz Heart $M^{\phi }$ are respectively defined by 
\begin{align}
L^{\phi }& :=\left\{ X\in L^{0}(\mathbb{R})\mid \mathbb{E}[\phi (\alpha
X)]<+\infty \text{ for some }\alpha >0\right\} ,  \label{Lphi} \\
M^{\phi }& :=\left\{ X\in L^{0}(\mathbb{R})\mid \mathbb{E}[\phi (\alpha
X)]<+\infty \text{ for all }\alpha >0\right\} ,  \label{Mphi}
\end{align}%
and they are Banach spaces when endowed with the Luxemburg norm. %
%
%
%
%
The topological dual of $M^{\phi }$ is the Orlicz space $L^{\phi ^{\ast }},$
where the convex conjugate $\phi ^{\ast }$ of $\phi $, defined by 
\begin{equation*}
\phi ^{\ast }(y):=\sup_{x\in \mathbb{R}}\left\{ xy-\phi (x)\right\} ,\ y\in 
\mathbb{R},
\end{equation*}%
is also a strict Young function. Note that 
\begin{equation}
\mathbb{E}[u(X)]>-\infty \text{\ if \ }\mathbb{E}[\phi (X)]<+\infty .
\label{orliz_cond}
\end{equation}

\begin{remark}
\label{remOrlicz}It is well known that $L^{\infty }(\probp;\mathbb{R}%
)\subseteq M^{\phi }\subseteq L^{\phi }\subseteq L^{1}(\probp;\mathbb{R}%
) $. In addition, from the Fenchel inequality $xy\leq \phi (x)+\phi ^{\ast
}(y) $ we obtain 
\begin{equation*}
(\alpha |X|)\left( \lambda \frac{dQ}{d\probp}\right) \leq \phi (\alpha
|X|)+\phi ^{\ast }\left( \lambda \frac{dQ}{d\probp}\right)
\end{equation*}%
for some probability measure $Q\ll P$, and we immediately deduce
that $\frac{dQ}{d\probp}\in L^{\phi ^{\ast }}$ implies $L^{\phi
}\subseteq L^{1}(Q;\mathbb{R})$.
\end{remark}

%
%

Given the utility functions $u_{1},\cdots ,u_{N}:\mathbb{R}\rightarrow 
\mathbb{R}$, satisfying the above conditions, with associated Young
functions $\phi _{1},\cdots ,\phi _{N}$, we define 
\begin{equation}
M^{\Phi }:=M^{\phi _{1}}\times \dots \times M^{\phi _{N}},\quad L^{\Phi
}:=L^{\phi _{1}}\times \dots \times L^{\phi _{N}}\,.  \label{orly-product}
\end{equation}

\subsection{Auxiliary results}

\begin{lemma}
\label{lemmaintegra} Let $v:[0,+\infty )\rightarrow {\mathbb{R}}\cup
\{+\infty \}$ be a convex function, and suppose that its restriction to $%
(0,+\infty )$ is real valued and differentiable. Let $Q\ll P$ be a given
probability measure with $v\left( \lambda \frac{\mathrm{d}Q}{\mathrm{d}P}%
\right) \in L^1(P)$ for all $\lambda >0$. Then

\begin{enumerate}
\item $v^{\prime }$ is defined on $(0,+\infty)$ and real valued there and
extendable to $[0,+\infty)$ by taking $\lim_{x\rightarrow 0}v^{\prime }(x)\in%
{\mathbb{R}}\cup\{-\infty\}$. Also, $\frac{\mathrm{d}Q}{\mathrm{d}P}%
v^{\prime }\left(\lambda \frac{\mathrm{d}Q}{\mathrm{d}P}\right)\in L^1(P)$
for all $\lambda>0$.

\item If $g$ is such that $g+\frac1g\in L^\infty_+(P)$, then $v\left(g \frac{%
\mathrm{d}Q}{\mathrm{d}P}\right)\in L^1(P)$.

\item If $v^{\prime }(0+)=-\infty$, $v^{\prime }(+\infty)=+\infty$ and $v$
is strictly convex $F(\gamma):=\mathbb{E} \left[\frac{\mathrm{d}Q}{\mathrm{d}%
P}v^{\prime }\left(\gamma \frac{\mathrm{d}Q}{\mathrm{d}P}\right)\right]$ is
a well defined bijection between $(0,+\infty)$ and ${\mathbb{R}}$.
\end{enumerate}
\end{lemma}

\begin{proof}
Lemma 2 of \cite{bf}.
\end{proof}

The following dual representation holds:

\begin{theorem}
\label{thmminimax} Let $u_{1}\dots ,u_{n}:{\mathbb{R}}\rightarrow {\mathbb{R}%
}$ be strictly increasing and concave functions. Let $\mathcal{C}\subseteq
M^{\Phi }$ be a convex cone such that for every $i,j=1,\dots ,N,$ $\mathbf{%
e_{i}}-\mathbf{e_{j}}\in \mathcal{C}$. Denote by $\mathcal{C}^{0}$ the polar
of the cone $\mathcal{C}$ in the dual pair $(M^{\Phi },L^{\Phi ^{\ast }})$ 
\begin{equation*}
\mathcal{C}^{0}:=\left\{ \mathbf{Z}\in L^{\Phi ^{\ast }}\text{ s.t. }%
\sum_{j=1}^{N}\mathbb{E}\left[ Y^{j}Z^{j}\right] \leq 0\,\,\forall \,\mathbf{%
Y}\in \mathcal{C}\right\} .
\end{equation*}%
Set 
\begin{equation*}
\mathcal{C}_{1}^{0}:=\left\{ \mathbf{Z}\in \mathcal{C}^{0}\text{ s.t. }%
\mathbb{E}\left[ Z^{1}\right] =\dots =\mathbb{E}\left[ Z^{N}\right]
=1\right\} \text{,\quad }(\mathcal{C}_{1}^{0})^{+}:=\left\{ \mathbf{Z}\in 
\mathcal{C}_{1}^{0}\text{ s.t. }Z^{j}\geq 0\text{ for all }j\right\}
\end{equation*}%
and suppose that 
\begin{equation*}
\sup_{\mathbf{Y}\in \mathcal{C}}\left( \sum_{j=1}^{N}\mathbb{E}\left[
u_{j}\left( X^{j}+Y^{j}\right) \right] \right) <+\infty \,\,\,\,\forall 
\mathbf{X}\in M^{\Phi }.
\end{equation*}%
Then%
\begin{equation*}
\sup_{\mathbf{Y}\in \mathcal{C}}\left( \sum_{j=1}^{N}\mathbb{E}\left[
u_{j}\left( X^{j}+Y^{j}\right) \right] \right) =\min_{\lambda \in {\mathbb{R}%
}_{+},\,\mathbf{Q}\in (\mathcal{C}_{1}^{0})^{+}}\left( \lambda \sum_{j=1}^{N}%
\mathbb{E}\left[ X^{j}\frac{\mathrm{d}Q^{j}}{\mathrm{d}P}\right]
+\sum_{j=1}^{N}\mathbb{E}\left[ v_{j}\left( \lambda \frac{\mathrm{d}Q^{j}}{%
\mathrm{d}P}\right) \right] \right) .
\end{equation*}%
If any of the two expressions above is strictly smaller than $%
\sum_{j=1}^{N}u_{j}(+\infty )$, then 
\begin{equation*}
\sup_{\mathbf{Y}\in \mathcal{C}}\left( \sum_{j=1}^{N}\mathbb{E}\left[
u_{j}\left( X^{j}+Y^{j}\right) \right] \right) =\min_{\lambda \in {\mathbb{R}%
}_{++},\,\mathbf{Q}\in (\mathcal{C}_{1}^{0})^{+}}\left( \lambda
\sum_{j=1}^{N}\mathbb{E}\left[ X^{j}\frac{\mathrm{d}Q^{j}}{\mathrm{d}P}%
\right] +\sum_{j=1}^{N}\mathbb{E}\left[ v_{j}\left( \lambda \frac{\mathrm{d}%
Q^{j}}{\mathrm{d}P}\right) \right] \right) .
\end{equation*}
\end{theorem}

\begin{proof}
$\,$


Observe first that $\mathbf{X}\mapsto \rho (\mathbf{X}):=-\sup_{\mathbf{Y}%
\in \mathcal{C}}\left( \sum_{j=1}^{N}\mathbb{E}\left[ u_{j}\left(
X^{j}+Y^{j}\right) \right] \right) $ is a non increasing, finite valued,
convex functional on the Fr\'{e}chet lattice $M^{\Phi }$. Only convexity is
non-evident: to show it, consider $\mathbf{X},\mathbf{Z}\in M^{\Phi }$ and $%
\mathbf{Y},\mathbf{W}\in \mathcal{C}$. For any $0\leq \lambda \leq 1,$ we
have by concavity 
\begin{eqnarray*}
&&\lambda \sum_{j=1}^{N}\mathbb{E}\left[ u_{j}\left( X^{j}+Y^{j}\right) %
\right] +(1-\lambda )\sum_{j=1}^{N}\mathbb{E}\left[ u_{j}\left(
Z^{j}+W^{j}\right) \right] \\
&\leq &\sum_{j=1}^{N}\mathbb{E}\left[ u_{j}\left( \lambda
(X^{j}+Y^{j})+(1-\lambda )(Z^{j}+W^{j})\right) \right] \\
&=&\sum_{j=1}^{N}\mathbb{E}\left[ u_{j}\left( \lambda X^{j}+(1-\lambda
)Z^{j}+\left( \lambda Y^{j}+(1-\lambda )W^{j}\right) \right) \right] \leq
-\rho (\lambda \mathbf{X}+(1-\lambda )\mathbf{Z})
\end{eqnarray*}%
as $\lambda \mathbf{Y}+(1-\lambda )\mathbf{W}\in \mathcal{C}$. Thus taking
suprema over $\mathbf{Y},\mathbf{W}\in \mathcal{C}$ we get 
\begin{equation*}
\lambda (-\rho (\mathbf{X}))+(1-\lambda )(-\rho (\mathbf{Z}))\leq -\rho
(\lambda \mathbf{X}+(1-\lambda )\mathbf{Z}).
\end{equation*}%
Now the Extended Namioka-Klee Theorem (see \cite{bfnam} Theorem A.3) can be
applied and we obtain 
\begin{equation*}
\rho (\mathbf{X})=\max_{\mathbf{0}\leq \mathbf{Z}\in L^{\Phi ^{\ast
}}}\left( \sum_{j=1}^{N}E_{\mathbb{P}}\left[ X^{j}(-Z^{j})\right] -\alpha (%
\mathbf{Z})\right) ,
\end{equation*}%
where 
\begin{align}
\alpha (\mathbf{Z}) &:=\sup_{\mathbf{X}\in M^{\Phi }}\left(
\sum_{j=1}^{N}E_{\mathbb{P}}\left[ X^{j}(-Z^{j})\right] -\rho (\mathbf{X}%
)\right)\notag \\
&=\sup_{\mathbf{X}\in M^{\Phi }}\left( \sum_{j=1}^{N}E_{\mathbb{P}}%
\left[ X^{j}(-Z^{j})\right] +\sup_{\mathbf{Y}\in \mathcal{C}}\left(
\sum_{j=1}^{N}\mathbb{E}\left[ u_{j}\left( X^{j}+Y^{j}\right) \right]
\right) \right)  \notag \\
&=\sup_{\mathbf{Y}\in \mathcal{C}}\left( \sup_{\mathbf{X}\in M^{\Phi
}}\left( \sum_{j=1}^{N}E_{\mathbb{P}}\left[ X^{j}(-Z^{j})\right] +\left(
\sum_{j=1}^{N}\mathbb{E}\left[ u_{j}\left( X^{j}+Y^{j}\right) \right]
\right) \right) \right)  \notag \\
&=\sup_{\mathbf{Y}\in \mathcal{C}}\left( \sum_{j=1}^{N}E_{\mathbb{P}}\left[
Y^{j}(Z^{j})\right] +\sup_{\mathbf{W}\in M^{\Phi }}\left( \sum_{j=1}^{N}E_{%
\mathbb{P}}\left[ W^{j}(-Z^{j})\right] +\left( \sum_{j=1}^{N}\mathbb{E}\left[
u_{j}\left( W^{j}\right) \right] \right) \right) \right) .  \label{alpha}
\end{align}%
Observe now that $-U(\mathbf{z}):=\sum_{j=1}^{N}-u_{j}(z^{j})$ for $\mathbf{z%
}\in {\mathbb{R}}^{N}$ defines a continuous, convex, proper function whose
Fenchel transform is 
\begin{equation*}
(-U)^{\ast }(\mathbf{w}):=\sup_{\mathbf{z}\in {\mathbb{R}}^{N}}\left(
\langle \mathbf{z},\mathbf{w}\rangle -(-U(\mathbf{z}))\right) =\sup_{\mathbf{%
z}\in {\mathbb{R}}^{N}}\left( \langle \mathbf{z},\mathbf{w}\rangle +U(%
\mathbf{z})\right) =\sup_{\mathbf{z}\in {\mathbb{R}}^{N}}\left( U(\mathbf{z}%
)-\langle \mathbf{z},-\mathbf{w}\rangle \right) =\sum_{j=1}^{N}v_{j}(-w^{j}).
\end{equation*}%
Now we apply Corollary on page 534 of \cite{Rockafellar} with $L=M^{\Phi }$, 
$L^{\ast }=L^{\Phi ^{\ast }}$, $F(\mathbf{x})=-U(\mathbf{x})$ to see that 
\begin{equation*}
\sup_{\mathbf{W}\in M^{\Phi }}\left( \sum_{j=1}^{N}E_{\mathbb{P}}\left[
W^{j}(-Z^{j})\right] +\sum_{j=1}^{N}\mathbb{E}\left[ u_{j}\left(
W^{j}\right) \right] \right) =E_{\mathbb{P}}\left[ \sum_{j=1}^{N}v_{j}(Z^{j})%
\right]
\end{equation*}%
and replacing this in (\ref{alpha}) we get: 
\begin{equation*}
\alpha (\mathbf{Z})=\sup_{\mathbf{Y}\in \mathcal{C}}\left( \sum_{j=1}^{N}E_{%
\mathbb{P}}\left[ Y^{j}Z^{j}\right] +E_{\mathbb{P}}\left[
\sum_{j=1}^{N}v_{j}(Z^{j})\right] \right) .
\end{equation*}%
Now observe that there are two possibilities:

\begin{itemize}
\item either $\mathbf{Z}\in\mathcal{C}^0$, and in this case $\alpha(\mathbf{Z%
})=E_\mathbb{P} \left[\sum_{j=1}^Nv_j(Z^j)\right]$ since $\mathbf{0}\in%
\mathcal{C}$

\item or $\alpha (\mathbf{Z})=+\infty $, since $v_{1},\dots ,v_{N}$ are
bounded from below.
\end{itemize}

Hence 
\begin{eqnarray}
-\sup_{\mathbf{Y}\in \mathcal{C}}\left( \sum_{j=1}^{N}\mathbb{E}\left[
u_{j}\left( X^{j}+Y^{j}\right) \right] \right) &=&\max_{\mathbf{0}\leq 
\mathbf{Z}\in L^{\Phi ^{\ast }}}\left( \sum_{j=1}^{N}E_{\mathbb{P}}\left[
X^{j}(-Z^{j})\right] -\alpha (\mathbf{Z})\right)  \notag \\
&=&\max_{\mathbf{0}\leq \mathbf{Z}\in \mathcal{C}^{0}}\left( -\left(
\sum_{j=1}^{N}E_{\mathbb{P}}\left[ X^{j}Z^{j}\right] +E_{\mathbb{P}}\left[
\sum_{j=1}^{N}v_{j}(Z^{j})\right] \right) \right)  \notag \\
&=&-\min_{\mathbf{0}\leq \mathbf{Z}\in \mathcal{C}^{0}}\left(
\sum_{j=1}^{N}E_{\mathbb{P}}\left[ X^{j}Z^{j}\right] +E_{\mathbb{P}}\left[
\sum_{j=1}^{N}v_{j}(Z^{j})\right] \right) .  \label{eqminimaxugly1}
\end{eqnarray}%
Moreover, since for every $i,j=1,\dots ,N$ $\mathbf{e_{i}}-\mathbf{e_{j}}\in 
\mathcal{C}$ we can argue as in Lemma \ref{rempolarisnice} to deduce that $%
\mathcal{C}^{0}\cap (L_{+}^{0})^{N}={\mathbb{R}_{+}}\cdot (\mathcal{C}%
_{1}^{0})^{+}$. Replacing this in the expression \eqref{eqminimaxugly1} we
get 
\begin{equation*}
\sup_{\mathbf{Y}\in \mathcal{C}}\left( \sum_{j=1}^{N}\mathbb{E}\left[
u_{j}\left( X^{j}+Y^{j}\right) \right] \right) =\min_{\lambda \in {\mathbb{R}%
}_{+},\,\mathbf{Q}\in (\mathcal{C}_{1}^{0})^{+}}\left( \lambda \sum_{j=1}^{N}%
\mathbb{E}\left[ X^{j}\frac{\mathrm{d}Q^{j}}{\mathrm{d}P}\right]
+\sum_{j=1}^{N}\mathbb{E}\left[ v_{j}\left( \lambda \frac{\mathrm{d}Q^{j}}{%
\mathrm{d}P}\right) \right] \right) .
\end{equation*}%
$\,$To prove the last claim, observe that if the optimum $\lambda $ in the
right hand side was $0$, we would have 
\begin{equation*}
\sup_{\mathbf{Y}\in \mathcal{C}}\left( \sum_{j=1}^{N}\mathbb{E}\left[
u_{j}\left( X^{j}+Y^{j}\right) \right] \right) =\sum_{j=1}^{N}v_{j}\left(
0\right) =\sum_{j=1}^{N}u_{j}(+\infty ),
\end{equation*}%
which contradicts our hypotheses.
\end{proof}

\begin{theorem}
\label{thmweirdclosure} Let $u_{1},\dots ,u_{N}$ satisfy Assumption \ref{A00}%
. Let $K\subseteq M^{\Phi }$ be a convex cone such that for all $i,j\in
\{1,\dots ,N\}$ $\mathbf{e_{i}}-\mathbf{e_{j}}\in K$ and suppose that $%
\mathcal{Q}_{v}^{e}\neq \emptyset $, where 
\begin{equation*}
\mathcal{Q}_{v}^{e}:=\left\{ \mathbf{Q}\sim P\mid \frac{\mathrm{d}Q^{j}}{%
\mathrm{d}P}\in L^{\Phi _{j}^{\ast }},\mathbb{E}\left[ v_{j}\left( \frac{%
\mathrm{d}Q^{j}}{\mathrm{d}P}\right) \right] <+\infty ,\sum_{j=1}^{N}E_{{Q}%
^{j}}\left[ k^{j}\right] \leq 0\text{ \ }\forall \mathbf{k}\in K\right\}
\subseteq L^{\Phi ^{\ast }}.
\end{equation*}%
Then denoting by $cl_{\mathbf{Q}}(\dots )$ the closure in $L^{1}\left(
Q^{1}\right)\times\dots\times L^{1}\left( Q^{N}\right) $ with respect to the
norm $\left\Vert \mathbf{X}\right\Vert _{\mathbf{Q}}:=\sum_{j=1}^{N}\left%
\Vert X^{j}\right\Vert _{L^{1}(Q^{j})}$ we have 
\begin{equation*}
\bigcap_{\mathbf{Q}\in \mathcal{Q}_{v}^{e}}cl_{\mathbf{Q}}\left(
K-L_{+}^{1}\left( \mathbf{Q}\right) \right) =\left\{ \mathbf{W}\in \bigcap_{%
\mathbf{Q}\in \mathcal{Q}_{v}^{e}}L^{1}\left(\mathbf{Q}\right) \mid
\sum_{j=1}^{N}E_{{Q}^{j}}\left[ W^{j}\right] \leq 0\,\,\forall \,\mathbf{Q}%
\in \mathcal{Q}_{v}^{e}\right\} .
\end{equation*}
\end{theorem}

\begin{proof}
We modify the procedure in \cite{bf} Theorem 4. The inclusion ($LHS\subseteq
RHS$) can be checked directly. As to the opposite one ($RHS\subseteq LHS$),
suppose we had a $\mathbf{k}\in RHS$ and a $\mathbf{Q}\in \mathcal{Q}%
_{v}^{e} $ with $\mathbf{k}\notin cl_{\mathbf{Q}}\left( K-L_{+}^{1}\left( 
\mathbf{Q}\right) \right) $, that is $\mathbf{k}\notin LHS$. We stress that
by construction 
\begin{equation}
\sum_{j=1}^{N}E_{{Q}^{j}}\left[ k^{j}\right] \leq 0\,\,\,\forall \mathbf{Q}%
\in \mathcal{Q}_{v}^{e}\,.  \label{eqforcontrad}
\end{equation}%
In the dual system 
\begin{equation*}
\left( L^{1}(\mathbf{Q}),L^{\infty}(\mathbf{Q})\right)
\end{equation*}%
the set $cl_{\mathbf{Q}}\left( K-L_{+}^{1}\left(\mathbf{Q}\right) \right) $
is convex and $\sigma \left( L^{1}(\mathbf{Q}),L^{\infty}(\mathbf{Q})\right)$%
-closed by compatibility of the latter topology with the norm topology. Thus
we can use Hahn-Banach Separation Theorem to get a class $\widehat{\mathbf{%
\xi }}\in L^{\infty }(\mathbf{Q})$ with 
\begin{equation}
0=\sup_{\mathbf{W}\in (K-L_{+}^{1}\left( \mathbf{Q}\right) )}\left(
\sum_{j=1}^{N}\mathbb{E}\left[ \widehat{\xi }^{j}W^{j}\frac{\mathrm{d}Q^{j}}{%
\mathrm{d}P}\right] \right) <\sum_{j=1}^{N}\mathbb{E}\left[ \widehat{\xi }%
^{j}k^{j}\frac{\mathrm{d}Q^{j}}{\mathrm{d}P}\right].  \label{eqseparation}
\end{equation}%
We now work componentwise. First observe that 
\begin{equation*}
\lbrack -1_{\widehat{\xi }^{j}<0}]_{j=1}^{N}\in 0-L_{+}^{\infty }\left( 
\mathbf{Q}\right) \subseteq K-L_{+}^{1}\left( \mathbf{Q}\right),
\end{equation*}%
so that $\widehat{\xi }^{j}\geq 0$ $Q^{j}$-a.s. for every $j=1,\dots ,N$.
Hence $\widehat{\xi }^{j}\frac{\mathrm{d}Q^{j}}{\mathrm{d}P}\geq 0$ $P$-a.s.
for every $j=1,\dots ,N$.

Moreover, since for all $i,j\in \{1,\dots ,N\}$ $\mathbf{e_{i}}-\mathbf{e_{j}%
}\in K$, we have 
\begin{equation}
\mathbb{E}\left[ \widehat{\xi }^{1}\frac{\mathrm{d}Q^{1}}{\mathrm{d}P}\right]
=\dots =\mathbb{E}\left[ \widehat{\xi }^{N}\frac{\mathrm{d}Q^{N}}{\mathrm{d}P%
}\right]\,.  \label{eqchainequalities}
\end{equation}%
It follows that for every $j=1,\dots ,N$ 
\begin{equation*}
P\left( \widehat{\xi }^{j}\frac{\mathrm{d}Q^{j}}{\mathrm{d}P}>0\right) >0
\end{equation*}%
since if this were not the case all the terms in equation %
\eqref{eqchainequalities} would be null, which would yield $\widehat{\xi }%
^{1}\frac{\mathrm{d}Q^{1}}{\mathrm{d}P}=\dots =\widehat{\xi }^{N}\frac{%
\mathrm{d}Q^{N}}{\mathrm{d}P}=0$, a contradiction with \eqref{eqseparation}.

Hence the vector 
\begin{equation*}
\frac{\mathrm{d}Q_{1}^{j}}{\mathrm{d}P}:=\frac{1}{\mathbb{E}\left[ \widehat{%
\xi }^{j}\frac{\mathrm{d}Q^{j}}{\mathrm{d}P}\right] }\widehat{\xi }^{j}\frac{%
\mathrm{d}Q^{j}}{\mathrm{d}P}
\end{equation*}%
is well defined and identifies a vector of probability measures $%
[Q_{1}^{1},\dots ,Q_{1}^{N}]$. We trivially have that 
\begin{equation*}
Q_{1}^{j}\ll P,\frac{\mathrm{d}Q_{1}^{j}}{\mathrm{d}P}\in L^{\Phi _{j}^{\ast
}}\,,
\end{equation*}%
and by equation \eqref{eqseparation}, together with \eqref{eqchainequalities}
\begin{equation}
\sup_{\mathbf{W}\in K}\left( \sum_{j=1}^{N}\mathbb{E}\left[ W^{j}\frac{%
\mathrm{d}Q_{1}^{j}}{\mathrm{d}P}\right] \right) \leq 0<\sum_{j=1}^{N}%
\mathbb{E}\left[ k^{j}\frac{\mathrm{d}Q_{1}^{j}}{\mathrm{d}P}\right]\,.
\label{eqQ1inpolar}
\end{equation}%
We observe that if we could prove $\mathbf{Q}_{1}\in \mathcal{Q}_{v}^{e}$,
we would get a contradiction with (\ref{eqforcontrad}). However this needs
not to be true, since we cannot guarantee $Q_{1}^{1},\dots ,Q_{1}^{N}\sim P$.

As $\mathbf{Q}\in \mathcal{Q}_{v}^{e}$, we have $\mathbf{Q}\sim P$, and for $%
\mathbf{Q}_{1}$ above we have $\mathbf{Q}_{1}\ll \mathbf{Q}$, $\frac{\mathrm{%
d}Q_{1}^{k}}{\mathrm{d}Q^{k}}\in L^{\infty }(Q^{k})=L^{\infty }(P)$. Take $%
\lambda \in (0,1]$ and define $\mathbf{Q}_{\lambda }$ via 
\begin{equation*}
\frac{\mathrm{d}Q_{\lambda }^{k}}{\mathrm{d}P}:=\lambda \frac{\mathrm{d}Q^{k}%
}{\mathrm{d}P}+(1-\lambda )\frac{\mathrm{d}Q_{1}^{k}}{\mathrm{d}P}\,.
\end{equation*}%
We now prove that $\mathbf{Q_{\lambda }}\in \mathcal{Q}_{v}^{e}$. It is easy
to check that 
\begin{equation*}
0<\lambda \leq \frac{\mathrm{d}Q_{\lambda }^{k}}{\mathrm{d}Q^{k}}\leq
(1-\lambda )\frac{\mathrm{d}Q_{1}^{k}}{\mathrm{d}Q^{k}}+\lambda\,,
\end{equation*}%
so that Lemma \ref{lemmaintegra}.2. with $g=g^{k}:=\frac{\mathrm{d}%
Q_{\lambda }^{k}}{\mathrm{d}Q^{k}}$, together with $\mathbb{E}\left[
v_{k}\left( \frac{\mathrm{d}Q^{k}}{\mathrm{d}P}\right) \right] <+\infty
\,\forall \,k=1,\dots ,N$ ($\mathbf{Q}\in \mathcal{Q}_{v}^{e}$ by
construction), yields 
\begin{equation*}
\mathbb{E}\left[ v_{k}\left( \frac{\mathrm{d}Q_{\lambda }^{k}}{\mathrm{d}P}%
\right) \right] =\mathbb{E}\left[ v_{k}\left( \frac{\mathrm{d}Q_{\lambda
}^{k}}{\mathrm{d}Q^{k}}\frac{\mathrm{d}Q^{k}}{\mathrm{d}P}\right) \right] =%
\mathbb{E}\left[ v_{k}\left( g^{k}\frac{\mathrm{d}Q^{k}}{\mathrm{d}P}\right) %
\right] <+\infty ,\,\,\,\forall \,k\in \{1,\dots ,N\},\,\,\lambda \in
(0,1]\,.
\end{equation*}%
Moreover $\mathbf{Q}\in \mathcal{Q}_{v}^{e}$ and $\lambda >0$ imply $\mathbf{%
Q}_{\lambda }^{k}\sim P$ for all $k=1,\dots ,N$. This, together with
equation \eqref{eqQ1inpolar}, yields 
\begin{equation*}
\sum_{j=1}^{N}\mathbb{E}\left[ W^{j}\frac{\mathrm{d}Q_{\lambda }^{j}}{%
\mathrm{d}P}\right] \leq 0\,\,\,\forall \mathbf{W}\in K,\forall \,\lambda
\in (0,1]\,.
\end{equation*}%
We can conclude that $\mathbf{Q}_{\lambda }\in \mathcal{Q}_{v}^{e},\,\forall
\lambda \in (0,1]$. At the same time 
\begin{equation*}
\sum_{j=1}^{N}\mathbb{E}\left[ k^{j}\frac{\mathrm{d}Q_{\lambda }^{j}}{%
\mathrm{d}P}\right] =\lambda \sum_{j=1}^{N}\mathbb{E}\left[ k^{j}\frac{%
\mathrm{d}Q^{j}}{\mathrm{d}P}\right] +(1-\lambda )\sum_{j=1}^{N}\mathbb{E}%
\left[ k^{j}\frac{\mathrm{d}Q_{1}^{j}}{\mathrm{d}P}\right] %
\xrightarrow[\lambda\rightarrow 0 ]{}\sum_{j=1}^{N}\mathbb{E}\left[ k^{j}%
\frac{\mathrm{d}Q_{1}^{j}}{\mathrm{d}P}\right] \overset{\text{Eq.}%
\eqref{eqQ1inpolar}}{>}0\,,
\end{equation*}%
which gives a contradiction with Equation \eqref{eqforcontrad}. We conclude
that $RHS\subseteq LHS$.
\end{proof}

\bibliographystyle{abbrv}
\bibliography{BibSORTE}

\begin{thebibliography}{10}

\bibitem{Acciaio}
B.~Acciaio.
\newblock Optimal risk sharing with non-monotone monetary functionals.
\newblock {\em Finance and Stochastics}, 11(2):267--289., 2007.

\bibitem{pedersen}
V.~V. Acharia, L.~H. Pedersen, T.~Philippon, and M.~Richardson.
\newblock Measuring systemic risk.
\newblock {\em The Review of Financial Studies}, 30(1):2--47, 2017.

\bibitem{Drapeau}
Y.~Armenti, S.~Crepey, S.~Drapeau, and A.~Papapantoleon.
\newblock Multivariate shortfall risk allocation and systemic risk.
\newblock {\em SIAM Journal on Financial Mathematics}, 9(1):90--126, 2018.

\bibitem{BE05}
P.~Barrieu and N.~{El Karoui}.
\newblock Inf-convolution of risk measures and optimal risk transfer.
\newblock {\em Finance and Stochastics}, 9:269--298, 2005.

\bibitem{BF02}
F.~Bellini and M.~Frittelli.
\newblock On the existence of minimax martingale measures.
\newblock {\em Mathematical Finance}, 12(1):1--21, 2002.

\bibitem{bffm0}
F.~Biagini, J.-P. Fouque, M.~Frittelli, and T.~Meyer-Brandis.
\newblock A unified approach to systemic risk measures via acceptance sets.
\newblock {\em Mathematical Finance}, 29:329--367, 2019.

\bibitem{bffm}
F.~Biagini, J.-P. Fouque, M.~Frittelli, and T.~Meyer-Brandis.
\newblock On fairness of systemic risk measures.
\newblock {\em Finance Stoch.}, 24(2):513--564, 2020.

\bibitem{bf}
S.~Biagini and M.~Frittelli.
\newblock Utility maximization in incomplete markets for unbounded processes.
\newblock {\em Finance and Stochastics}, 9(4):493--517, 2005.

\bibitem{BF08AAP}
S.~Biagini and M.~Frittelli.
\newblock A unified framework for utility maximization problems: an orlicz
  space approach.
\newblock {\em Annals of Applied Probability}, 18(3):929--966, 2008.

\bibitem{bfnam}
S.~Biagini and M.~Frittelli.
\newblock On the extension of the {N}amioka-{K}lee theorem and on the {F}atou
  property for risk measures.
\newblock In {\em Optimality and risk---modern trends in mathematical finance},
  pages 1--28. Springer, Berlin, 2009.

\bibitem{Borch}
K.~Borch.
\newblock Equilibrium in a reinsurance market.
\newblock {\em Econometrica}, 30(3):424--444, 1962.

\bibitem{Buhlmann1}
H.~B\"{u}hlmann.
\newblock An economic premium principle.
\newblock {\em Astin Bulletin}, 11(1):52--60, 1980.

\bibitem{Buhlmann}
H.~B\"{u}hlmann.
\newblock The general economic premium principle.
\newblock {\em Astin Bulletin}, 14(1):13--21, 1982.

\bibitem{BJ79}
H.~B\"{u}hlmann and W.~S. Jewell.
\newblock Optimal risk exchange.
\newblock {\em Astin Bulletin}, 10:243--262, 1979.

\bibitem{Carlier1}
G.~Carlier and R.-A. Dana.
\newblock Pareto optima and equilibria when preferences are incompletely known.
\newblock {\em Journal of Economic Theory}, 148(4):1606--1623, 2013.

\bibitem{Carlier2}
G.~Carlier, R.-A. Dana, and A.~Galichon.
\newblock Pareto efficiency for the concave order and multivariate
  comonotonicity.
\newblock {\em Journal of Economic Theory}, 147(1):207--229, 2012.

\bibitem{chen}
C.~Chen, G.~Iyengar, and C.~Moallemi.
\newblock An axiomatic approach to systemic risk.
\newblock {\em Management Science}, 56(6):1373--1388, 2013.

\bibitem{CK}
J.~Cvitanic and I.~Karatzas.
\newblock Convex duality in constrained portfolio optimization.
\newblock {\em Annals of Applied Probability}, 2:718--767, 1992.

\bibitem{Dana_LeVan}
R.~A. Dana and C.~{Le Van}.
\newblock Overlapping sets of priors and the existence of efficient allocations
  and equilibria for risk measures.
\newblock {\em Mathematical Finance}, 20(3):327--339, 2010.

\bibitem{Debreu}
G.~Debreu.
\newblock {\em Theory of value : an axiomatic analysis of economic
  equilibrium}.
\newblock Cowles Foundation for research in Economics at Yale University,
  Monograph 17, John Wiley \& Sons, 1959.

\bibitem{ELW182}
P.~Embrechts, H.~Liu, T.~Mao, and R.~Wang.
\newblock Quantile-based risk sharing with heterogeneous beliefs.
\newblock {\em Mathematical Programming}, 2018.

\bibitem{ELW18}
P.~Embrechts, H.~Liu, and R.~Wang.
\newblock Quantile-based risk sharing.
\newblock {\em Operations Research}, 66(4):893--1188, 2018.

\bibitem{FeinsteinRudloffWeber}
Z.~Feinstein, B.~Rudloff, and S.~Weber.
\newblock Measures of systemic risk.
\newblock {\em SIAM Journal on Financial Mathematics}, 8(1):672--708, 2017.

\bibitem{Kupper}
D.~Filipovi\'{c} and M.~Kupper.
\newblock Optimal capital and risk transfers for group diversification.
\newblock {\em Mathematical Finance}, 18(1):55--76, 2008.

\bibitem{Filipovic_Svindland}
D.~Filipovi\'{c} and G.~Svindland.
\newblock Optimal capital and risk allocations for law- and cash-invariant
  convex functions.
\newblock {\em Finance and Stochastics}, 12(3):423 -- 439, 2008.

\bibitem{Follmer}
H.~F\"{o}llmer and A.~Schied.
\newblock {\em Stochastic Finance: An Introduction In Discrete Time}.
\newblock De Gruyter, 4th edition, 2016.

\bibitem{JP_Langsam}
J.-P. Fouque and J.~A. Langsam.
\newblock {\em Handbook on Systemic Risk}.
\newblock Cambridge, 2013.

\bibitem{Heath_Ku}
D.~Heath and H.~Ku.
\newblock Pareto equilibria with coherent measures of risk.
\newblock {\em Mathematical Finance}, 14(2):163--172, 2004.

\bibitem{Hurd}
T.~R. Hurd.
\newblock {\em Contagion! Systemic Risk in Financial Networks}.
\newblock Springer, 2016.

\bibitem{JST07}
E.~Jouini, W.~Schachermayer, and N.~Touzi.
\newblock Optimal risk sharing for law invariant monetary utility functions.
\newblock {\em Mathematical Finance}, 18(2):269--292, 2008.

\bibitem{KLSX}
I.~Karatzas, J.~Lehoczky, S.~Shreve, and G.~Xu.
\newblock Martingale and duality methods for utility maximization in an
  incomplete market.
\newblock {\em SIAM Journal on Control \& Optimization}, 29:702--730, 1991.

\bibitem{Kromer}
E.~Kromer, L.~Overbeck, and K.~Zilch.
\newblock Systemic risk measures on general measurable spaces.
\newblock {\em Mathematical Methods of Operation Research}, 84:323--357, 2016.

\bibitem{LS18}
F.~B. Liebrich and G.~Svindland.
\newblock Risk sharing for capital requirements with multidimensional security
  markets.
\newblock {\em Finance and Stochastics}, 4, 2019.

\bibitem{MasColell}
A.~{Mas Colell} and W.~Zame.
\newblock {\em Equilibrium Theory in infinite dimensional spaces.}
\newblock In: handbook of mathematical economics IV, Elsevir Science Publisher,
  1991.

\bibitem{MRG14}
E.~Mastrogiacomo and E.~{Rosazza Gianin}.
\newblock Pareto optimal allocations and optimal risk sharing for quasiconvex
  risk measures.
\newblock {\em Mathematics and Financial Economics}, 9:149--167, 2015.

\bibitem{RaoZen}
M.~M. Rao and Z.~D. Ren.
\newblock {\em Theory of Orlicz Spaces}.
\newblock Marcel Dekker Inc. N.Y, 1991.

\bibitem{Rockafellar}
R.~T. Rockafeller.
\newblock Integrals which are convex functionals.
\newblock {\em Pacific Journal of Mathematics}, 24(3):525--540, 1968.

\bibitem{S01}
W.~Schachermayer.
\newblock Optimal investment in incomplete markets when wealth may become
  negative.
\newblock {\em Annals of Applied Probability}, 11(3):694--734, 2001.

\bibitem{Tsanakas}
A.~Tsanakas.
\newblock To split or not to split: Capital allocation with convex risk
  measures.
\newblock {\em Insurance: Mathematics and Economics}, 44(2):268--277, 2009.

\bibitem{Weber17}
S.~Weber.
\newblock Solvency {II}, or how to sweep the downside risk under the carpet.
\newblock {\em Insurance: Mathematics and Economics}, 82:191--200, 2018.

\end{thebibliography}


\begin{thebibliography}{APPR11}
\bibitem[APPR11]{pedersen} {V. V. Acharia, L. H. Pedersen, T. Philippon, M.
Richardson \emph{Measuring Systemic Risk}. AFA 2011 Denver Meetings Paper,
2010.}

\bibitem[AB05]{Aliprantis} {C. D. Aliprantis and K. C. Border, \emph{%
Infinite Dimensional Analysis}. Springer, 2005 (Third Edition).}

\bibitem[ACDP18]{Drapeau} {Y. Armenti, S. Crepey, S. Drapeau, A.
Papapantoleon \emph{Multivariate shortfall risk allocation and systemic risk}%
. SIAM Journal on Financial Mathematics, 9(1):90-126, 2018.}

\bibitem[B01]{Bauer} {H. Bauer, \emph{Measure and Integration Theory}. De
Gruyter, 2001.}

\bibitem[BF02]{BF02} F. Bellini and M. Frittelli, \emph{On the existence of
minimax martingale measures}, Mathematical Finance, Vol. 12/1, pp. 1-21,
2002.

\bibitem[BFFM18]{bffm} {F. Biagini, J.-P. Fouque, M. Frittelli, T.
Meyer-Brandis, \emph{On Fairness of Systemic Risk Measures}. Finance and Stochastics,
24(2):513-564, 2020}

\bibitem[BFFM19]{bffm0} {F. Biagini, J.-P. Fouque, M. Frittelli, T.
Meyer-Brandis, \emph{A unified approach to systemic risk measures via
acceptance sets}. Mathematical Finance, 29:329-367, 2019.}

\bibitem[BF05]{bf} {S. Biagini and M. Frittelli, \emph{Utility maximization
in Incomplete Markets for Unbounded Processes}, Finance and Stochastics Vol.
9/4 pp. 493-517, 2005.}

\bibitem[BF08]{BF08AAP} {S. Biagini and M. Frittelli, \emph{A unified
framework for utility maximization problems: an orlicz Space approach}.
Annals of Applied Probability, 18(3):929-966, 2008.}

\bibitem[BF09]{bfnam} {S. Biagini and M. Frittelli, \emph{On the extension
of the Namioka-Klee theorem and on the Fatou property for Risk Measures },
Optimality and risk: modern trends in mathematical Finance. The Kabanov
Festschrift, 2009.}

\bibitem[BFG11]{bfg} {S. Biagini and M. Frittelli, M. Grasselli \emph{%
Indifference Price with General Semimartingales}, Mathematical Finance,
21(3):423-446, 2011.}

\bibitem[B62]{Borch} {K. Borch, \emph{Equilibrium in a Reinsurance Market},
Econometrica Vol. 30, No. 3, 1962.}

\bibitem[B82]{Buhlmann} {H. B\"{u}hlmann, \emph{The general Economic Premium
principle}, Astin Bulletin Vol. 14, No.1, 1982.}

\bibitem[B80]{Buhlmann1} {H. B\"{u}hlmann, \emph{An Economic Premium
Principle}, Astin Bulletin 11, 52-60, 1980.}

\bibitem[CIM13]{chen} {C. Chen, G. Iyengar, C. Moallemi \emph{An axiomatic
approach to systemic risk}. Management Science, 56(6):1373-1388, 2013.}

\bibitem[FS11]{Follmer} {H. F\"{o}llmer, A. Schied, \emph{Stochastic
Finance: An Introduction In Discrete Time}. De Gruyter, 2011(Third Edition).}

\bibitem[L69]{Luen} {D. G. Luenberger, \emph{Optimization by Vector Space
Methods}, Wiley-Interscience, 1969. }

\bibitem[RR91]{RaoZen} {M. M. Rao and Z.D. Ren, \emph{Theory of Orlicz Spaces%
}. Marcel Dekker Inc. N.Y., 1991.}

\bibitem[R68]{Rockafellar} {R. T. Rockafeller, \emph{Integrals which are
convex functionals}. Pacific Journal of Mathematics, 1968.}

\bibitem[S01]{S01} {W. Schachermayer, \emph{Optimal investment in incomplete
markets when wealth may become negative}. Annals of Applied Probability,
11:694-734, 2001.}

\bibitem[..]{..}{..................................................}

\bibitem[BE05]{BE05}{P. Barrieu1, N. El Karoui, \emph{Inf-convolution of risk measures and optimal risk transfer}. Finance and Stochastics Vol. 9 pp. 269–298, 2005.}

\bibitem[BJ79]{BJ79}{H. B\"{u}hlmann, W. S. Jewell \emph{Optimal Risk Exchange}. Astin Bulletin Vol. 10 pp. 243-262, 1979.}

\bibitem[CRG18]{CR18}{F. Centrone, E. Rosazza Gianin, \emph{Capital allocation àla Aumann–Shapley for non-differentiable risk measures}. European Journal of Operational Research Vol. 267 pp. 667–675, 2018.}

\bibitem[D01]{D01}{M. Denault, \emph{Coherent allocation of risk capital}.  Journal of Risk Vol. 4 No. 1 pp. 1-34, 2001.}


\bibitem[ELW18]{ELW18}{P. Embrechts, H. Liu, R. Wang, \emph{Quantile-based Risk Sharing}.  Operations Research Vol. 66 No. 4 pp. 893-1188, 2018.}

\bibitem[ELW18_2]{ELW18_2}{P. Embrechts, H. Liu, T. Mao, R. Wang, \emph{Quantile-based Risk Sharing with Heterogeneous Beliefs}.  Preprint, 2018.}

\bibitem[FK08]{FK08}{D. Filipovi\'{c}, M. Kupper, \emph{Optimal Capital and Risk Transfer for Group Diversification}.  Mathematical Finance Vol. 10 No.1 pp. 55-76, 2008.}

\bibitem[Hurd]{Hurd}{
T. R. Hurd, \emph{Contagion! Systemic Risk in Financial Networks},
Springer, 2016.
}

\bibitem[JST07]{JST07}{E. Jouini, W. Schachermayer, N. Touzi, \emph{Optimal risk sharing for law invariant monetary utility functions}.  Mathematical Finance Vol. 18 No. 2 pp. 269-292, 2008.}

\bibitem[JL]{JP_Langsam}{J.-P.  Fouque, J. A. Langsam, editors, \emph{Handbook on Systemic Risk},
Cambridge, 2013.}

\bibitem[LS18]{LS18}{F. B. Liebrich, G. Svindland, \emph{Risk sharing for capital requirements with multidimensional security markets}. Preprint: arXiv:809.10015, 2018.}


\bibitem[MRG14]{MRG14}{E. Mastrogiacomo, E. Rosazza Gianin, \emph{Pareto optimal allocations and optimal risk sharing for quasiconvex risk measures}. Mathematics and Financial Economics 9:149–167, 2015.}
\end{thebibliography}

\end{document}